\documentclass[12pt,oneside,reqno]{article} 
\usepackage{graphics}
\usepackage[a4paper]{geometry}

\usepackage[english]{babel}
\usepackage{amsmath}
\usepackage{amsthm}
\usepackage{graphicx}
\usepackage[dvipsnames]{xcolor}
\usepackage{empheq}

\usepackage{amsfonts}
\newtheorem{proposition}{Proposition}
\usepackage{graphicx}
\usepackage{dcolumn}
\usepackage{bm}

\graphicspath{{./Figure/}}

\begin{document}


\title{Agents' beliefs and economic regimes polarization\\ in interacting markets}

\begin{author}
{Fausto Cavalli$^{\,\,\rm a}$\footnote{E-mail address: fausto.cavalli@unicatt.it}\,\,,
Ahmad Naimzada$^{\,\,\rm b}$\footnote{E-mail address: ahmad.naimzada@unimib.it}\,\,,
Nicol\'o Pecora$^{\,\,\rm c}$\footnote{E-mail address: nicolo.pecora@unicatt.it}\,\,, 
Marina Pireddu$^{\,\,\rm d}$\footnote{E-mail address: marina.pireddu@unimib.it}\\
{\footnotesize{${}^{\rm a\,}$Dept. of Mathematical Sciences, Mathematical Finance and
  Econometrics,}}\\ {\footnotesize{Catholic University, Via Necchi 9, 20123 Milano, Italy.}}\\
{\footnotesize{${}^{\rm b\,}$Dept. of Economics,
Management and Statistics, University of Milano-Bicocca,}}\\ {\footnotesize{U6 Building, Piazza dell'Ateneo Nuovo 1, 20126 Milano, Italy.}}\\
{\footnotesize{${}^{\rm c\,}$Dept. of Economics and
  Social Sciences, Catholic University,}}\\ {\footnotesize{Via Emilia Parmense 84, 29122 Piacenza, Italy.}}\\
	{\footnotesize{${}^{\rm d\,}$Dept. of Mathematics and its
Applications, University of Milano-Bicocca,}}\\
{\footnotesize{U5 Building, Via Cozzi 55, 20125 Milano, Italy.}}}
\end{author}
\date{}
\maketitle


\begin{abstract}
 In the present paper a model of a market consisting of real and
  financial interacting sectors is studied. Agents populating the
  stock market are assumed to be not able to observe the true
  underlying fundamental, and their beliefs are biased by either
  optimism or pessimism. Depending on the relevance they give to
  beliefs, they select the best performing strategy in an evolutionary
  perspective. The real side of the economy is described within a
  multiplier-accelerator framework with a nonlinear, bounded
  investment function. We show that strongly polarized beliefs in an
  evolutionary framework can introduce multiplicity of steady states,
  which, consisting in enhanced or depressed levels of income, reflect
  and reproduce the optimistic or pessimistic nature of the agents'
  beliefs. The polarization of these steady states, which coexist with
  an unbiased steady state, positively depends on that of the beliefs
  and on their relevance. Moreover, with a mixture of analytical and
  numerical tools, we show that such static characterization is
  inherited also at the dynamical level, with possibly complex
  attractors that are characterized by endogenously fluctuating
  pessimistic and optimistic levels of national income and price. This
  framework, when stochastic perturbations are included, is able to
  account for stylized facts commonly observed in real financial
  markets, such as fat tails and excess volatility in the returns
  distributions, as well as bubbles and crashes for stock prices.
\end{abstract}

\vspace{2mm}


%

\maketitle

\bigskip

\bigskip

\noindent
{\it The 2008-2009 crisis stimulated new reflections on the origins and the
causes of turmoil in the economic environment. Accordingly, it turned
out to be relevant the understanding of how instability can reinforce
and spread among different sectors, and of what are the factors that may
hinder their interaction.  As a consequence of such complexity, agents
populating the economic system are often better represented as
boundedly rational actors, mispricing risks and inappropriately evaluating the
economic variables. This is usually at the source of
instabilities, that can make the economic system evolve in
unpredictable ways, especially when non-deterministic components are
acting.}\\
{\it The research goal of the present paper is to improve the
understanding of the effect that a biased evaluation of the economic
variables may have on the resulting economic dynamics and of whether an
increase in the integration between the real and financial sectors is harmful or
beneficial to the economy. Relying on analytical investigations
complemented by numerical simulations, we study how a boundedly
rational agents framework can give rise to multistability phenomena in the
economic system, and how the eventual related instability may spread
between the integrated markets.}


\section{Introduction}\label{sec:intro}
Seeking the explanations of the 2008-2009 crisis only within the real
or the financial side of the economy leads to a dead end way. The
global response to the crisis in the form of fiscal stimulus indeed
awakened the interest in the multiplier-accelerator approach dating
back to \cite{samuelson}, from which the modern business
cycle theory originated. According to this scheme, endogenous
fluctuations spring in the economic activity as a consequence of the
mutual interplay between consumption, investments and national income
changes, driven either by a multiplicative component that links
expenditures to national income (the so-called Keynesian multiplier)
or by the principle through which induced investments depend on the
pace of growth in the economic activity (the accelerator
principle). The original model of \cite{samuelson} has been refined
and developed through the decades, emphasizing from time to time
different aspects, ranging from the role of the monetary sector or of the
inventory adjustment to that of expectations in the national income,
or to the effect of bounded investment functions. A burgeoning
literature has been widely disseminating: we just mention
the contributions in \cite{CD93, hicks,  hommes, LW06, naim1, GPS05}.\\
However, the increasing relevance and influence of the financial side
of the economy on the real sector (the so-called phenomenon of
\textit{financialization} of the real economy) can not be neglected,
with the consequent need to investigate the interactions between the
real and stock subsystems. The related literature can be roughly
divided into two research strands, where the former sustains and
describes the reverberations of the financial sector into the real one
(as in \cite{chen,long,naim3,west}), while the latter, on the
contrary, contains a support for the opposite influence direction (see
e.g. \cite{bask,kontonikas,long,naim3}).

In the wake of the aforementioned discussion, the recent financial
crisis strengthened the need to understand the causes of turmoil in
the economic environment, keeping a focus on the role of the financial
system. Among several factors, a systematic mispricing of risk is
commonly regarded as one of the underlying sources of instability for
the economic environment. The presence of such a bias, indeed, reflects
a boundedly rational behavior of the involved economic actors, who may
have a rough knowledge of the economic fundamentals. This can lead agents to
systematically overestimate or underestimate them, and to possibly
stick to or switch among different forecasting rules on the basis of
simple heuristics. It is well-known (see \cite{hommes}) that bounded
rationality together with heterogeneity is a potential source of
endogenous fluctuations in the dynamics of the relevant economic
variables.
 An interesting consequence of this modeling framework is also the capability to exhibit simulated time
series characterized by several qualitative properties, as leptokurtic
returns distributions, large and clustered volatility, that are
distinctive features of real time series.

The present contribution belongs to the research strand on the
study of interacting markets, and embraces a modelling approach similar
to those adopted in \cite{cnp17,naim2,west}. In particular, in each of
these works, the sources of complex dynamics are always clearly linked
to and explained through particular mathematical phenomena occurring
in the model.\\ For example, in \cite{west} a Keynesian-type goods real
market is combined with a stock market populated by heterogeneous
fundamentalist and chartist agents. The main result is that the system
is characterized by different coexisting steady states, that can
become unstable and give rise to complex dynamics. From the
mathematical point of view, we have persistent multistability of
stable steady states or of attractors characterized by similar degrees
of complexity.\\  Conversely, in \cite{naim2} it was deeply investigated
the effect of coupling either stable and unstable isolated real and
stock markets. In this case, a Keynesian good market interacts with a
stock market characterized by optimistic and pessimistic agents, whose
fractions can change over time according to an evolutionary switching mechanism.
Surprisingly, instabilities
can not only spread from one market to the other, but they can even arise by increasing
the level of interaction between markets that are stable
when isolated. From the mathematical viewpoint, this is explained
investigating the loss of stability and the kinds of consequent
dynamics arising, depending on the degree of interaction or on the
degree of optimism/pessimism. Moreover, the analysis of the resulting
time series of the economic variables shows a significant deviation
from normality.\\
The work in \cite{cnp17} provides evidence that even a potentially stable
economic system can show endogenously fluctuating dynamics, since a
locally stable steady state can coexist with other complex
attractors, so that the long-run dynamics are strongly path-dependent
and can result in sequences of alternating periods of persisting large or small volatility.

In all the above-mentioned papers, the agents' heterogeneity and their
boundedly rational nature have no effects on the set of possible
steady states, which always consists of a unique steady state
(\cite{naim2,cnp17}) or of three coexisting steady states
(\cite{west}).\\
However, it is reasonable to wonder whether the boundedly rational
nature of the agents can also affect the final outcomes of the system
in terms of existing steady states and their position. In other words,
we are interested in inquiring if, as a consequence of a biased
evaluation of economic variables, ``biased'' steady states can
emerge. For example, may an excess of optimistic/pessimistic bias
foster the emergence of ``optimistic'' and ``pessimistic'' steady
states, which, in turn, can evolve into optimistic and
pessimistic complex attractors?  Additionally, to which extent the
birth of biased steady states is affected by the markets'
integration?

In order to investigate these aspects, building on the approach
pursued in the above mentioned contributions, we analyze an economy
made by two, possibly interdependent, subsystems, namely a real and a
financial sector. The real sector essentially recalls the economy
presented in \cite{cnp17} and \cite{naim1}, in which a
multiplier-accelerator model with a bounded investment function is
studied, and it is coupled with a stock market populated by optimistic
and pessimistic fundamentalists, modeled by following the approach in
\cite{naim2}. In so doing, we shall focus our attention on the degree
of interaction between markets, which positively depends on the parameter
$\omega\in[0,1],$ as well as we shall take into account the role of
agents' beliefs, represented by a symmetric optimistic/pessimistic
bias $b>0,$ and their willingness to switch toward the best
performing strategy while operating in the stock market, described by
the so-called intensity of choice $\beta>0$ (see e.g. \cite{hommes})
that regulates the evolutionary selection between forecasting
rules. The resulting dynamics are described by a three-dimensional
nonlinear system of difference equations for which we analyze the
existence and the stability of steady states and, through global
analysis, we investigate the presence of multistability phenomena.

The results we come up with can be read under different
perspectives. Our most relevant finding concerns the
static analysis of the set of possible steady states. We show that,
independently of the structural features of the sectors under
consideration and of their degree of integration, the steady state final
outcomes depend on the optimistic/pessimistic bias and on the
intensity of choice of the evolutionary process. If both $b$ and
$\beta$ are suitably small, a unique steady state exists, whose
position is independent of $b$ and $\beta.$ Conversely, if either
$b$ or $\beta$ increases, two new biased steady states emerge,
whose position depends on such parameters and which coexist with the
previous one. With respect to it, the two new steady states respectively exhibit larger or
smaller values in the level of national income and price. This means that
in some sense, such final outcomes ``keep track'' of the
optimistic/pessimistic biased belief of the agents, and depart from the unbiased equilibrium.
The intuition for the existence of these biased steady states can be read as follows: when the realizations of the true fundamental value of the traded asset are close to some biased beliefs, if the intensity of choice is high enough, then almost all agents will adopt the optimistic or pessimistic predictor, which is the best performing predictor in terms of forecasting error, leading the dynamics to converge to a biased equilibrium. In the same way, when the intensity of choice is sufficiently large, even small differences in predictors' performances may lead agents to massively switch among forecasting rules. Such alternating waves of optimism and pessimism can drive the system to periodic cycles.\\
From the mathematical viewpoint, the emergence of additional steady states
follows a pitchfork bifurcation, whose consequent multistability
is further investigated through global analysis, in order to have a
complete picture of the possible dynamical scenarios. The
independence of the emergence (but not of the position) of the biased steady states from $\omega$ makes possible that the same economic
setting can simultaneously take advantage or be hindered from
isolated or fully integrated markets, with a consequent increase in
the level of national income and prices, or, conversely, a
reduction in the level of national income like, for instance, during a
financial shock within an integrated system. The latter framework can
be read in the wake of the recent financial crisis where, if on the one
hand an integrated economy may facilitate exchange and production
through the intermediation of financial instruments, on the other hand
no real benefits to the society deriving from the activities of
increased financialization may be appraised.

Carrying on the investigation and moving to a dynamical perspective, the
analysis we perform allows us to determine the local stability
conditions for each steady state, from which we infer a possible
ambiguous role for each relevant parameter. Deepening the
investigation with the help of numerical simulations, we show that
periodic, quasi-periodic and complex dynamics can occur, from a global viewpoint, either when a
unique or three steady states exist. Accordingly,
oscillations arising in our model are endogenously generated and
stoked by the acceleration principle acting in the real subsystem
which, in turn, may be triggered by the nonlinearity of the real subsystem,
as found in
\cite{naim1} too. In addition, we show the possible coexistence of pairs
of complex attractors, arising by the loss of stability of the two
new steady states, which give rise to endogenous fluctuations around
either the pessimistic or the optimistic steady state. We also find
that, even when the fundamental steady state is locally stable, other
attractors may coexist with it, making the choice of initial condition
crucial for the course of the resulting economic dynamics. Dealing with a purely deterministic version of the model, the convergence
toward a particular attractor is determined in advance by the
choice of the initial datum. Conversely, when a stochastic version of
the model is considered, in which the agents' demands are
buffeted with noise, trajectories can approach different
attractors from time to time, so that, for example, periods of
persistent optimistic and pessimistic fluctuations can interchange,
giving rise to alternating booms and busts. In general, the emergence
of stylized facts is mathematically connected with an increasing
complexity of the generated dynamics. On the contrary, we show that this is
possible in our model even in the presence of stable steady states, when the
pitchfork bifurcation occurs. That level of investigation turns out
to be relevant to reproduce realistic dynamics observed in the real
financial markets together with several stylized facts regarding the
return of stock prices (and output as well), such as positive
autocorrelation, volatility clustering and non-normal distribution
characterized by a high kurtosis and fat tails, with
dynamics resulting from the combination of nonlinear forces and
random shocks. Therefore, the analytical and numerical results we retrieve from
the deterministic framework flow into a better understanding of the
stochastic framework.

To summarize, a growing polarization of beliefs gives rise to an
increasing richness in the economic regimes that may originate due to the evolutive framework. In such regimes the optimistic and
pessimistic hallmark of the belief is not lost but it is mixed up by the complexity of the scenarios, and
marks the emerging regimes which are characterized by low, medium or high levels of economic activity.

The remainder of the paper is organized as follows: Section 2 describes the
model economy highlighting the features of the two markets; in Section 3 we present the
analytical results on the steady states and their local stability; in Section 4 the simulations
regarding the deterministic framework as well as the stochastic
version of the model are performed and analyzed; in Section 5 we draw some conclusions on our findings;
the Appendix gathers all the proofs.

\section{The model economy}\label{sec:model}
In the present section we shall introduce the baseline model made up
of two interacting sectors, i.e., the real market and the stock
market, which are separately outlined in the following subsections.
Since the real market we consider essentially consists in the same
closed economy taken into account in \cite{cnp17}, while the stock
market is very close to that depicted in \cite{naim2},
we limit ourselves to report the constitutive and distinctive elements of each
market, addressing the interested reader to the above-mentioned works
for a more detailed description. In what follows, we are making use of
sigmoid functions, namely functions
$g:\mathbb{R}\rightarrow\mathbb{R}$ that satisfy the following
assumptions:
\begin{subequations}
  \label{eq:sp}
  \begin{align}
    \bullet&\,\, \text{vanish at } z=0 \text{ and are strictly
      increasing}; \label{eq:spa}\\
    \bullet&\,\, \text{are convex on }
    (-\infty,0] \text{, concave on } [0,+\infty)  \text{ and attain
      maximum slope } 1 \text{ at } z=0;\label{eq:spb}\\
    \bullet&\,\, \text{are bounded above and below.}\label{eq:spc}
  \end{align}
\end{subequations}

\subsection{The real market}
As usual, at each time $t,$ the national income $Y_t$ can be expressed
through the macroeconomic equilibrium condition
\begin{equation}
  Y_t = C_t + G_t + I_t, \label{equil}
\end{equation}
in which we assume constant exogenous government expenditures
$G_t\equiv \bar{G}$ and aggregate consumption $C_t$ depending on
autonomous consumption $\bar{C}$ and on the last period national
income, setting $C_t=\bar{C} + cY_{t - 1},$ where $c \in (0, 1)$
represents the marginal propensity to consume. Private investments are
expressed through
\begin{equation}
  I_t = \bar{I}  + \gamma g_I (Y_{t - 1} - Y_{t - 2}) + \omega h P_{t - 1},
  \label{inv}
\end{equation}
where $\bar{I}$ are exogenous investments, while the last two terms
respectively encompass the accelerator principle and the dependence of
the real market on the stock sector, with $\omega\in[0,1]$
representing the degree of interaction between the two markets, $h>0$
the marginal propensity to invest from the stock market wealth and
$P_{t-1}$ the price of the financial asset. As noted in \cite{west}, a
better performance of the financial side of the economy, here
encompassed in $P,$ has a beneficial effect on the patrimonial
situation of the private sector, promoting greater investments.\\
Concerning the second term, $\gamma > 0$ is the accelerator parameter
and $g_I:\mathbb{R}\rightarrow\mathbb{R}$ is a function that satisfies
assumptions \eqref{eq:sp}. Thanks to assumption \eqref{eq:spa}, if the
national income raises or decline, after a gestation lag, the
component of investments depending on national income variation
$\Delta Y_{t-1}$ increases or decreases accordingly, by a bounded
quantity, thanks to assumption \eqref{eq:spc}. Finally, thanks to
assumption \eqref{eq:spb}, the effect of both positive and negative
national income variations are less and less significant as
$|\Delta Y_{t-1}|\rightarrow+\infty,$  attaining its
maximum at $\Delta Y_{t-1}=0,$ where the slope coincides with
accelerator $\gamma.$ We stress that all the previous aspects are in
line with the classic macroeconomic literature of the 1930s-1950s (see
e.g.  \cite{goodwin,kaldor,kalecki}).

Taking into account the expressions for the government expenditure,
aggregate consumption and private investments, from \eqref{equil} we
obtain
\[
Y_{t}=A+cY_{t-1}+\gamma g_I(Y_{t-1}-Y_{t-2}) +\omega hP_{t-1},
\]
where we defined $A\equiv\bar{C}+\bar{I}+\bar{G}$ as the sum of the
autonomous components. We stress that if $\omega = 0$ the stock market
has no influence on the national income, while if $\omega = 1$ the
two markets are fully integrated.
\subsection{The stock market}
We assume that the stock market is populated by fundamentalists, who
buy (sell) the asset if its price is lower (higher) than the value
estimated as the fundamental asset price but, due to some form of
bounded rationality, do not know the true fundamental value of the
asset price, trying to form forecasts about it. On the basis of
these beliefs, they operate in the stock market. In particular,
we suppose there are two groups of fundamentalists: optimists, who
typically overestimate the true fundamental asset price, and
pessimists, who underestimate it.  This approach has been adopted
before in \cite{degrauwe} and \cite{hommes}, assuming that
optimists (resp. pessimists) expect a constant price above
  (resp. below) the fundamental price. The asset price is set by a
market maker on the basis of the following nonlinear adjustment
mechanism
\begin{equation}
  P_{t}=P_{t-1}+\sigma g_P\left( \alpha _{t}D_{t-1}^{O}+\left( 1-\alpha_t \right)D_{t-1}^{P}\right), \label{mm}
\end{equation}
where $\alpha_{t}$ represents the fraction of optimistic
fundamentalists, $D_{t-1}^{O}$ and $D_{t-1}^{P}$ denote optimistic and
pessimistic demand, respectively, while $\sigma >0$ measures the
reactivity of the price variation to aggregate excess demand and
$g_P:\mathbb{R}\rightarrow\mathbb{R}$ is a function fulfilling
assumptions \eqref{eq:sp}. We stress that a similar price adjustment
mechanism has been deeply investigated in \cite{NP15}, where it was
introduced to prevent an overreaction of price variation to a large
excess demand with the consequent unrealistic uncontrolled growth or
negativity of prices. In the present contribution, an overreaction
would also lead the real sector variables to diverge or
to become economically meaningless.

The demand of optimistic and pessimistic agents is proportional to the
gap between the estimated fundamental asset price and the realized
price, i.e.,
$D_{t}^{O} =\mu \left( F_{t}^{O}-P_{t}\right),\, D_{t}^{P} =\mu \left(
  F_{t}^{P}-P_{t}\right),$
where $\mu>0$ represents the agents' reactivity to such gap. In an
environment in which agents, due to their bounded rationality, roughly
know the fundamental stock price, they can disagree about the correct
value of the fundamental price and can make different estimates on it. In particular, we assume that
$F_{t}^{O} =F_{t}^{\ast }+b$ and $F_{t}^{P} =F_{t}^{\ast }-b,$ where
$b$ is a positive parameter representing the bias on the
  fundamental price. On raising $b,$ agents become more and more
  distant in their beliefs, which are increasingly polarized
  toward high level of optimism/pessimism. Hence, parameter $b$
  also portrays the degree of agents' heterogeneity\footnote{As
    it will become evident along the paper, $b$ is one of the
    most relevant parameters from both the interpretative and the
    mathematical point of view. To keep the investigation analytically
    tractable and the interpretation more straightforward, we decided
    to limit ourselves to
    consider the symmetric situation in which $F_{t}^{O}$ and $F_{t}^{P}$ lie at the same distance
		$b$ from $F_{t}^{\ast }.$}.\\
As concerns the true, unobserved, fundamental value $F^{*}_{t},$ it is described by a
weighted average between an exogenous fundamental value $F^{*}$ and an
endogenous component proportional to the national income $Y_t,$
resulting in $F_t^{\ast } = (1 - \omega) F^{\ast } + \omega dY_t,$ where
$d>0$ captures the strength of the linear dependence between the
fundamental value and the current national income level.  In
this way, $F_t^{\ast }$ is linked to the course of the real economy
through $\omega \in [0, 1],$ i.e., the interaction degree intensity
parameter which has been introduced for the investments' function in
\eqref{inv}. It is worth noticing that if $\omega =0$ the true
unobserved fundamental value is completely exogenous, while if
$\omega =1$ it is endogenously
determined by national income, as in \cite{west}.

Therefore, the demand functions for both optimists and pessimists become
\[
\begin{split}
D_{t}^{O} &=\mu \left( \left( 1-\omega \right) F^{\ast }+\omega
dY_{t}+b-P_{t}\right) \\
D_{t}^{P} &=\mu \left( \left( 1-\omega \right) F^{\ast }+\omega
dY_{t}-b-P_{t}\right),
\end{split}
\]
and, replacing them into the price equation \eqref{mm}, we obtain the
following expression for the price
  dynamics
\[
P_{t}=P_{t-1}+\sigma g_P\big(\mu \left( \left( 1-\omega \right) F^{\ast }+\omega
dY_{t-1}-P_{t-1}+b\left( 2\alpha _{t}-1\right) \right)\big).
\]

In order to discipline the evolution of the agents' share employing
a certain speculative rule, we assume that a fraction of optimistic
agents can switch from period to period to the other behavior, and
vice versa. In the present model, differently from \cite{naim2} and similarly to \cite{CNP17ee}, such
an evolutionary process is governed by the squared forecasting error
evaluated on the last realization of the fundamental value and of the
 stock price, namely
\[
SE_{t}^{i}=\left( F_{t-1}^{i}-P_{t-1}\right)^{2}, \hspace{5mm} i\in\{O,P\}.
\]
Following \cite{brock}, we assume the fractions evolve according to a
multinomial logit
rule, that is
\[
\alpha _{t}=\frac{e^{-\beta \left(
F_{t-1}^{O}-P_{t-1}\right) ^{2}}}{e^{-\beta \left( F_{t-1}^{O}-P_{t-1}\right) ^{2}}+e^{-\beta
\left( F_{t-1}^{P}-P_{t-1}\right) ^{2}}},
\]
where the positive parameter $\beta,$ also called intensity of choice, measures how fast the mass of optimistic fundamentalists
will switch to the optimal prediction strategy. In the limit case
$\beta=0$, fractions will be constant and equal to 0.5, and traders never
switch strategy; in the other extreme case $\beta\rightarrow+\infty$,
in each period all traders use the same, optimal strategy. The
latter case represents the highest possible degree of rationality with respect to
strategy selection based upon past performance in a heterogeneous
agents environment (see \cite{hommes} for further
discussion on the topic). Moreover, $\beta$ also encompasses the
  agents' perception of the beliefs relevance. Independently of the moderate or strong polarization of the
  beliefs (small or large value of $b$), if the agents give a small relevance
  to them ($\beta$ is small), they will more likely stick to their
  current forecasting behavior, while if they give a serious
  consideration to them ($\beta$ is large), their tendency to trust and
  switch to the best performing belief will result strengthened.

In so doing, the equation that describes the dynamics of the asset
price reads as
\[
P_{t}=P_{t-1}+\sigma g_P\left(\mu \left( \left( 1-\omega \right) F^{\ast }+\omega
dY_{t-1}-P_{t-1}+b\left( \frac{2}{1+e^{-4b\beta \left( P_{t-1}-\left(
1-\omega\right) F^{\ast }-\omega dY_{t-1}\right) }}-1\right) \right)\right).
\]

Summarizing, defining $Z_{t}\equiv Y_{t-1}$ in view of the subsequent
analysis and taking into account both the {real and financial}
sectors, we can introduce the function $G=(G_1,G_2,G_3):\mathbb R_+^3\to\mathbb R^3,\,
({Y_t,P_t,Z_t})\mapsto(G_1(Y_t,P_t,Z_t),G_2(Y_t,P_t,Z_t),G_3(Y_t,P_t,Z_t)),$
which {describes} the functioning of the whole economy:

\begin{subequations}\label{map}
\begin{empheq}[left={\empheqlbrace\,}]{align}
  &Y_{t+1} =G_1(Y_t,P_t,Z_t) = A+cY_{t}+\gamma g_I(Y_{t}-Z_{t}) +\omega hP_{t},\label{mapY} \\
  &P_{t+1}=G_2(Y_t,P_t,Z_t)= P_{t}+\sigma g_P\big(\mu \left( \left(
      1-\omega \right) F^{\ast }+\omega
    dY_{t}-P_{t}+b\left( 2\alpha_{t+1}-1\right) \right)\big),\label{mapP}\\
  &Z_{t+1}=G_3(Y_t,P_t,Z_t)= Y_t,\label{mapZ}
\end{empheq}
\end{subequations}
where
\[
  \alpha_{t+1}=\dfrac{1}{1+e^{-4b\beta \left( P_{t}-\left( 1-\omega\right)
F^{\ast }-\omega dY_{t}\right) }}.
\]

\section{Analytical results on steady states and their stability}\label{sec:an}
In this section we investigate the existence of equilibria for the
model in \eqref{map}, studying the effect on the position and on the
stability of the steady states played by the relevant parameters,
paying specific attention to the role of the intensity of choice $\beta,$ of the bias $b$ and of the
interaction parameter $\omega.$ We
start by investigating the number and the expression of possible
steady states of the map in \eqref{map}.
\begin{proposition}\label{ss}
  System \eqref{map} has
  \begin{itemize}
  \item[a)] a unique steady state
    $S^{\ast}=(Y^{\ast}, P^{\ast}, Z^{\ast}) = \left(\frac{A+\omega \left(
          1-\omega \right) hF^{\ast }}{1-c-\omega^{2}dh},
      \frac{(1-\omega)(1-c) F^{\ast }+\omega d A}{1-c-\omega ^{2}dh},
      \frac{A+\omega \left( 1-\omega \right) hF^{\ast
        }}{1-c-\omega^{2}dh}\right)$
    if
    \begin{equation}\label{cond}
      b\le\frac{1}{\sqrt{2\beta}}.
    \end{equation}
    Such steady state is well defined for any interaction degree value
    $\omega\in [0,1]$ provided that
    \begin{equation}
      1 - c - h d > 0. \label{a1}
    \end{equation}
  \item[b)] three steady states
    $S^{\ast}=(Y^{\ast}, P^{\ast}, Z^{\ast}),\, S^{L}=(Y^{L}, P^{L},
    Z^{L})$
    and $S^{H}=(Y^{H}, P^{H}, Z^{H})$ if $b>\frac{1}{\sqrt{2\beta}}.$
In particular, $S^{L}$ and
    $S^{H}$ are symmetric w.r.t.  $S^{\ast},$ which is still well
    defined for any $\omega$ provided that \eqref{a1} is
    fulfilled. Moreover, it holds that $I^{L}<I^{\ast}<I^{H},\,$ for
    $I\in\{Y,\,P,\,Z\},$
    and 
    \begin{equation}
      \begin{gathered}
        P_{\ell}=P^{\ast} - \frac{1-c}{1 - c - \omega^2 dh} b < P^L < P^{\ast} < P^H <
        P^{\ast} + \frac{1 - c}{1 - c - \omega^2 dh} b=P_{h}.\\
        Y_{\ell}=Y^{\ast} - \frac{ h \omega}{1-c-\omega^2 d h }b < Y^L < Y^{\ast} < Y^H <
        Y^{\ast} + \frac{ h \omega}{1-c-\omega^2 d h }b=Y_{h}.
        \label{bounds}
      \end{gathered}
    \end{equation}

    All the components of such new steady states are strictly positive
    if
    \begin{equation}
      b<\min\left\{\dfrac{dA}{1-c},F^{\ast}\right\}. \label{a2}
    \end{equation}
  \end{itemize}
\end{proposition}

From Proposition \ref{ss}, it is possible to have either one or three
steady states, whose existence is triggered by sufficiently large
values of $\beta$
and $b.$
Steady state $S^{\ast},$
in which $Y^{\ast}$
reduces to the same steady state of the Samuelson model when
$\omega=0,$ is in agreement with
those found in \cite{cnp17,naim3,naim2}, so that all the comments therein apply to
$S^{\ast},$ as well.

We stress that the steady state $S^{\ast}$ represents an \textit{unbiased}
  final outcome of the economic system, since it is affected neither by
the bias, nor by the evolutionary pressure.  Conversely,
  suitably strong beliefs (namely a suitably strong joint effect of
  belief polarization and relevance given to them by the agents)
  trigger the emergence of the two supplementary \textit{biased}
  steady states $S^L$ and $S^H.$ Moreover, recalling \eqref{bounds},
  $S^{\ast}$ consists in intermediate steady price and national
  income levels, while $S^L$ and $S^H$ represent polarized steady states,
  with respectively small and large steady values. Proposition
  \ref{ss} then shows that sufficiently strong beliefs can drive,
  independently of the economic setting and its features, the
  system toward optimistic (like $S^H$) or pessimistic (like
  $S^L$) steady states. This result is absent in the related
  literature, where the only possibilities are given either by a unique steady state
  (\cite{cnp17,naim2}) or by always existing multiple steady states
  (\cite{west}), which are then independent of the
  underlying characteristics of the economy and of the agents. In
  agreement with a Keynesian line of thought, Proposition \ref{ss}
  clearly shows how beliefs, influencing the stock market performance
  that, in an integrated framework, has in turn a direct influence on
  the investments, can be the essential cause and driving force that
  introduces and leads to final outcomes consisting of persistent
  flourishing or depressed income levels. The optimistic/pessimistic
  polarization of beliefs, under the evolutionary perspective,
  reflects on the ``sentiment'' of the steady states, and can drive
  the dynamics of the economy toward a more optimistic or pessimistic
  final output. We stress that also in such situations, the
  intermediate outcome $S^{\ast}$ is still relevant, as it will become
  evident from the dynamical analysis, while in \cite{west} it
  attracts just a set with null measure.

  We recall that $b$ encompasses the boundedly rational nature of the
  agents, while the presence of the switching mechanism, that is
  crucial as well, is regulated by the parameter $\beta$, that is
  positively connected to the rationality of the agents in the
  selection of the forecasting rule, as the larger is the intensity
  of choice, the larger is the fraction of agents that ``rationally''
  turns to the speculative rule that proved to be the most effective
  in forecasting the actual stock price. In this sense, we can say that
	the emergence of the pessimistic and optimistic steady states
	is fostered by the joint action of an irrational component ($b$) and of a
	rational component ($\beta$) in agents' behavior.
	The roles of $b$ and $\beta,$
  which we already noted to be very interconnected from the
  interpretative point of view, can not be completely disentangled
  also from the mathematical viewpoint, and we will see that both
  parameters have a similar influence on the results.


From \eqref{bounds}, the optimistic or pessimistic divergence of
  the biased steady states from the unbiased one is potentially more
  relevant as the polarization of the beliefs increases. Our next
  level of investigation is then devoted to study the effect of the
  main parameters on the steady states position. Since
  $Z^i=Y^i,\,i\in\{\ast,H,L\}$ and $Z$ has not its own relevance from the
  economic point of view, we do not deal with it, recalling, however,
  that it behaves exactly as $Y.$

In order to have the steady state $S^{\ast}$ well defined for any
interaction degree value $\omega\in [0,1],$ in the remainder of the
paper we shall assume that \eqref{a1} is fulfilled, even when not
explicitly mentioned, as well as we also assume \eqref{a2} when
dealing with $S^H$ and $S^L$, to guarantee the economic meaningfulness
of each
steady state value.\\
In the following comparative static result, we investigate the effects
of $b$ and $\beta$ on the steady state values $Y^i$ and
$P^i,\,i\in\{\ast,H,L\}.$
\begin{proposition}\label{scbeta}
  Under assumption \eqref{a1}, $S^{\ast}$ remains
  unchanged on increasing $\beta$ or $b.$\\
  Let $b<\min\{dA/(1-c),F^{\ast }\}.$ 
  Then, on increasing $\beta\in(1/(2b^2),+\infty),$ we have that $P^H$
  and $Y^H$ are strictly increasing, while $P^L$ and $Y^L$ are
  strictly decreasing. Moreover,
  $\lim_{\beta\rightarrow+\infty} P^L(\beta)=P_{\ell}$ and
  $\lim_{\beta\rightarrow+\infty} P^H(\beta)=P_{h},$
  $\lim_{\beta\rightarrow+\infty} Y^L(\beta)=Y_{\ell}$ and
  $\lim_{\beta\rightarrow+\infty} Y^H(\beta)=Y_{h},$
  with $P_{\ell},\,Y_{\ell}$ and $P_h,\,Y_h$ representing respectively the lower and upper bounds defined in \eqref{bounds}.\\
  On increasing $b\in (1/{\sqrt{2\beta}},\min\{dA/(1-c),F^{\ast }\}),$
  we have that $P^H$ and $Y^H$ are strictly increasing, while $P^L$
  and $Y^L$ are strictly decreasing.
\end{proposition}


The comparative static results of Proposition \ref{scbeta} deepen the
understating of the effect of the beliefs on the steady states. In Figure
\ref{fig:scbetab} we report the evolution of $Y^L,\,Y^{\ast}$ and
$Y^{H}$ on increasing $\beta$ and $b,$ setting
$c=d=h=0.5,A=F=10,\omega=1,$ and $b=0.5$ (Figure \ref{fig:scbetab} (A)),
 $\beta=1$ (Figure \ref{fig:scbetab} (B)). The values of $Y^{L}$
and $Y^{H}$ are obtained numerically solving an implicit
equation similar to \eqref{im}, but in terms of $Y,$ using \eqref{yp}. Indeed, $S^{\ast}$ is neither affected by the intensity of
choice nor by the bias. Conversely, not only the existence of $S^L$
and $S^H$ is affected by $\beta$ and $b,$ but also their values, and in particular their departure from the intermediate
  steady state $S^{\ast}.$ This means that the stronger are the beliefs (in terms of polarization or intensity of choice),
	the stronger get the optimistic/pessimistic connotation of the possible final
  outcomes. In agreement with \eqref{bounds}, the maximum
  possible deviation of $S^{L}$
and $S^{H}$ from $S^{\ast}$ is proportional to the belief
  polarization and the bounds for $P$ and $Y$ can actually be approached as
  $\beta\rightarrow +\infty$, that is agents immediately switch to the
  best performing behavior as if ``perfectly rational''.
\begin{figure}[t]
\begin{center}
  \begin{minipage}{0.4\textwidth}
    \centering
    \includegraphics[width=\textwidth,trim=1cm 0cm 2cm 0cm, clip=true]{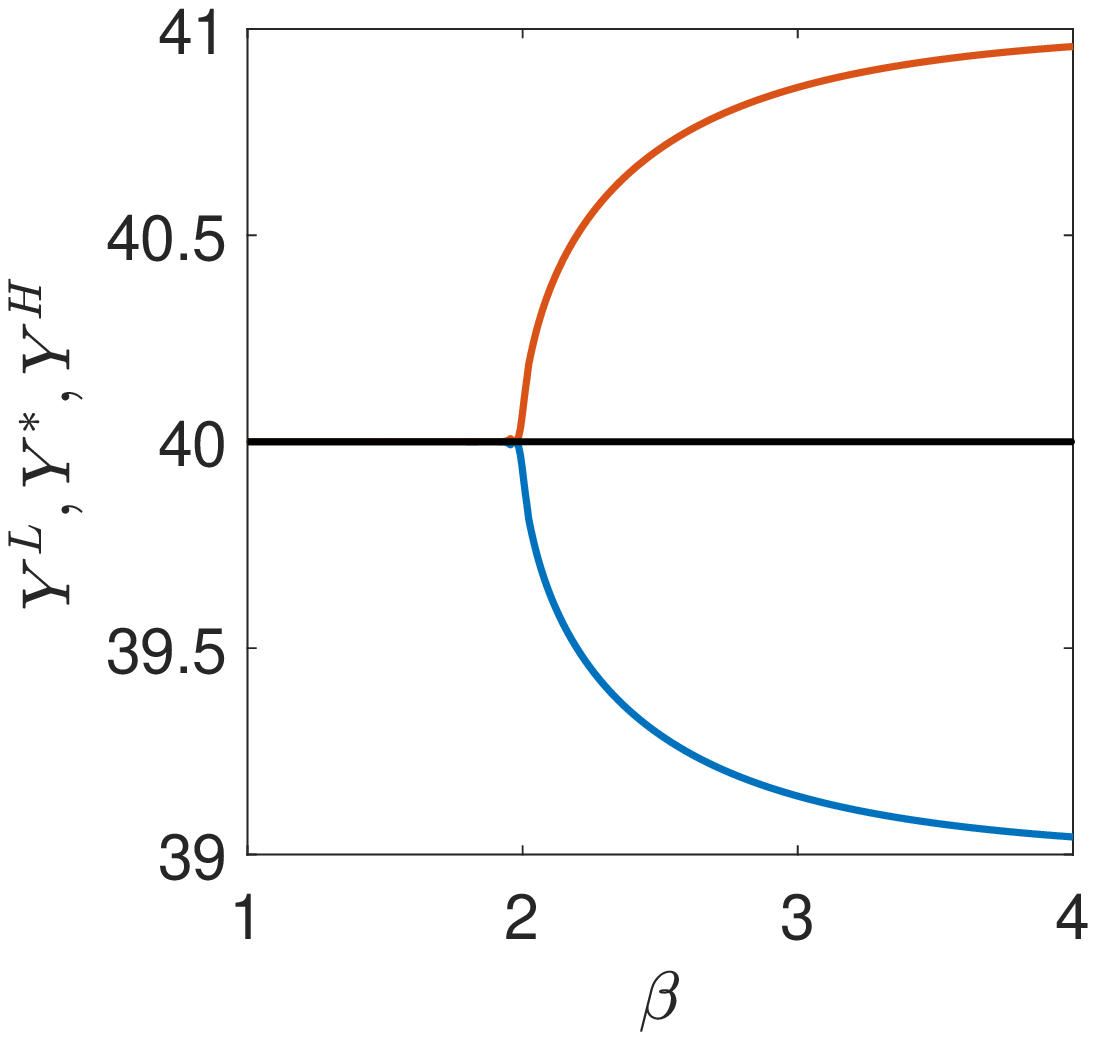}\\
    (A)
  \end{minipage}
  \begin{minipage}{0.4\textwidth}
    \centering
    \includegraphics[width=\textwidth,trim=1cm 0cm 2cm 0cm, clip=true]{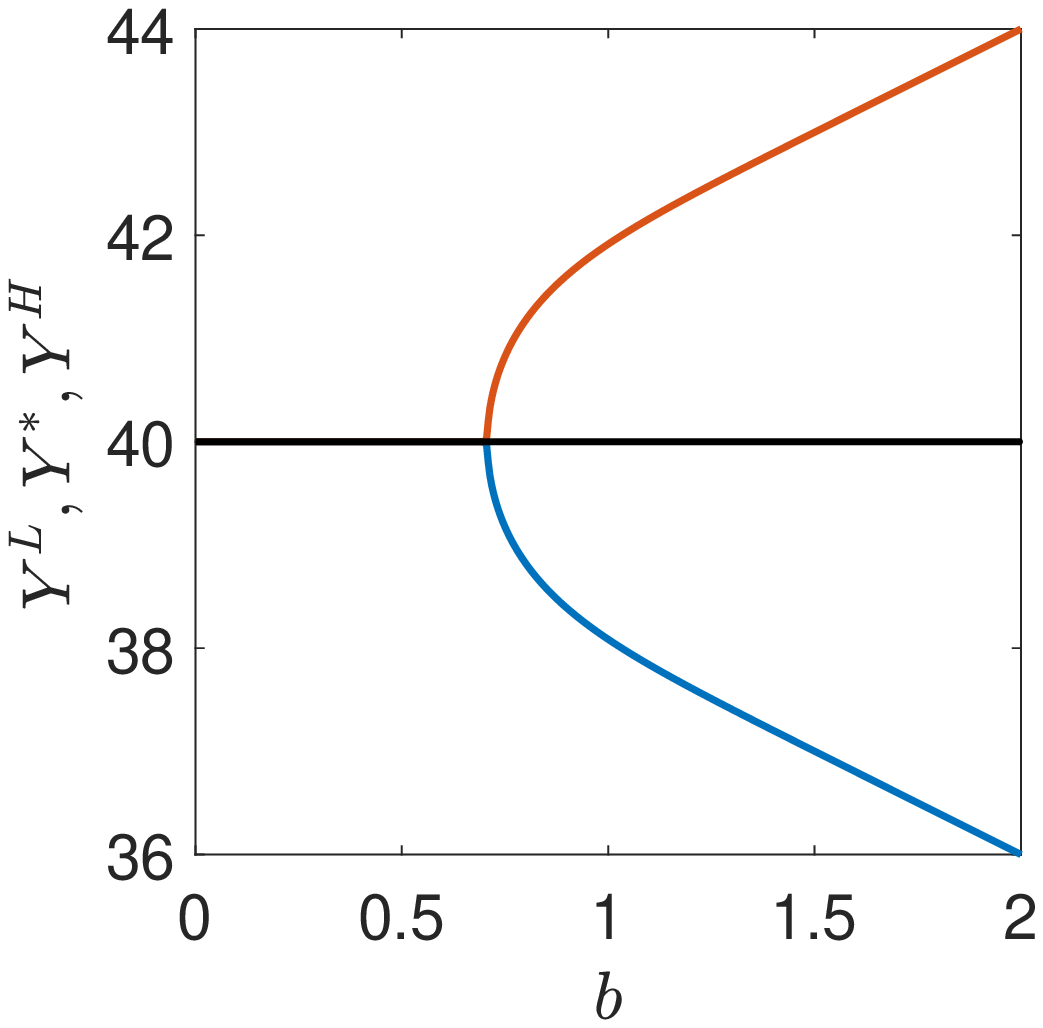}
    (B)
  \end{minipage}
\end{center}
\caption{Steady state national incomes $Y^{L}$ (blue), $Y^{\ast}$ (black) and $Y^{H}$
  (red) on varying $\beta$ (plot (A)) and $b$ (plot
  (B)).} \label{fig:scbetab}
\end{figure}

The next proposition studies the effects of $\omega$ on the steady
states.

\begin{proposition}\label{prop3}
  Under assumption \eqref{a1}, on increasing $\omega\in[0,1],$
  $P^{\ast}$
  can be either increasing, decreasing or there exists
  ${\omega}_{P^{\ast}}\in(0,1)$ such that $P^{\ast}$ decreases on
  $[0,{\omega}_{P^{\ast}})$ and increases on $({\omega}_{P^{\ast}},1].$
  On increasing $\omega\in[0,1],$
  $Y^{\ast}$
  can be either increasing or there exists
  ${\omega}_{Y^{\ast}}\in(0,1)$ such that $Y^{\ast}$ increases on
  $[0,{\omega}_{Y^{\ast}})$ and decreases on $({\omega}_{Y^{\ast}},1].$
  \\
  Let ${1}/{\sqrt{2\beta}}<b<\min\{dA/(1-c),F^{\ast }\}$ 
  and assume that
  \begin{equation}\label{ine}
    P^{L}(\omega)\ne P^{\ast} - \frac{1-c}{1 - c - {\omega}^2 d h} \sqrt{b^2 - \frac{1}{2
        \beta}},\quad\forall\omega\in (0,1).
  \end{equation}	
  Then, on increasing $\omega\in(0,1),$
  we have that, for $j\in\{L,H\},$
  $P^j$ can be either increasing, decreasing or there exists
  ${\omega}_{P^{j}}\in(0,1)$ such that $P^j$ decreases on
  $(0,{\omega}_{P^{j}})$ and increases on $({\omega}_{P^{j}},1).$
  On increasing $\omega\in(0,1),$ for $j\in\{L,H\},$
  $Y^{j}$
  can be either increasing or there exists
  ${\omega}_{Y^{j}}\in(0,1)$ such that $Y^{j}$ increases on
  $(0,{\omega}_{Y^{j}})$ and decreases on $({\omega}_{Y^{j}},1).$
\end{proposition}

\begin{figure}[t!]
\begin{center}
  \begin{minipage}{0.32\textwidth}
    \centering
    \includegraphics[width=\textwidth,trim=1.7cm 0cm 2.6cm 0cm, clip=true]{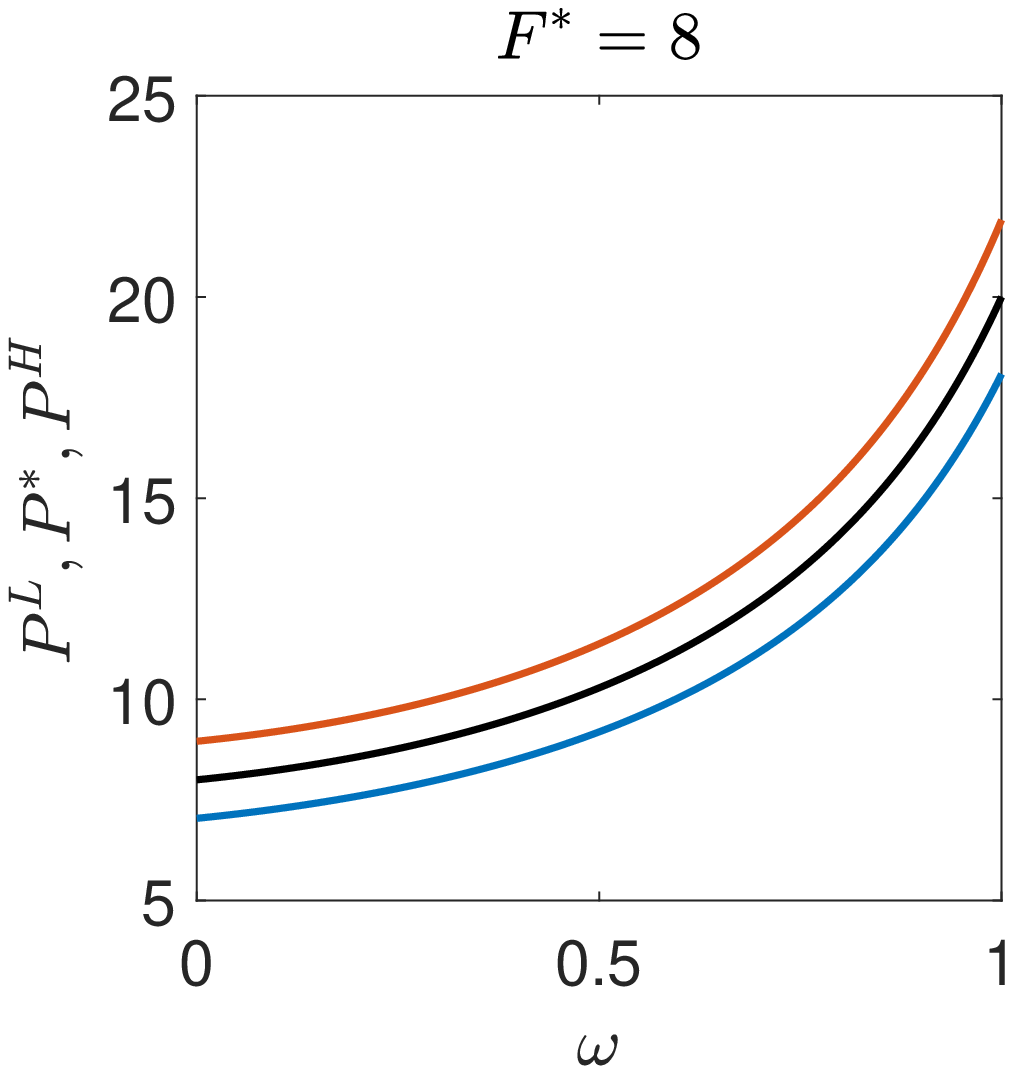}\\
    (A)
  \end{minipage}
  \begin{minipage}{0.32\textwidth}
    \centering
    \includegraphics[width=\textwidth,trim=1.7cm 0cm 2.6cm 0cm, clip=true]{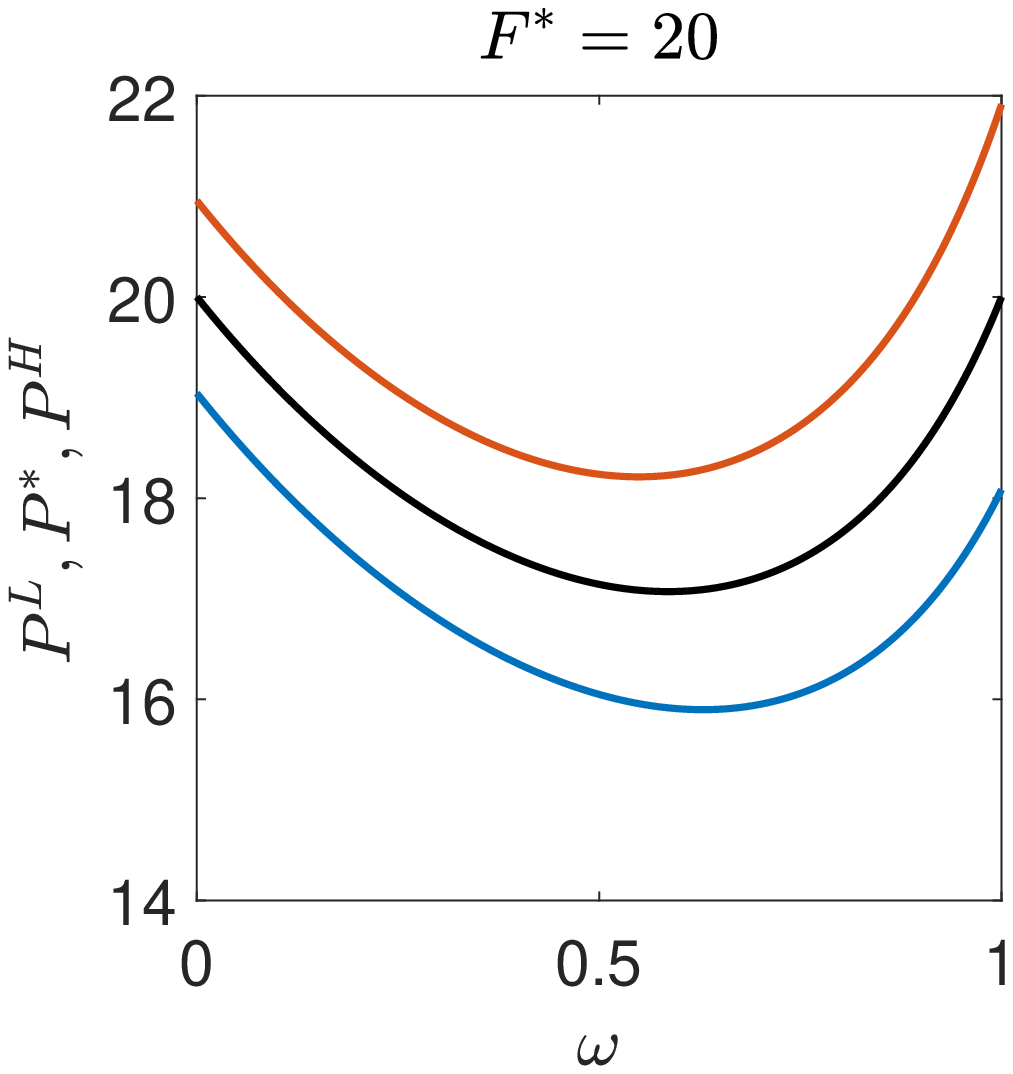}
    (B)
  \end{minipage}
  \begin{minipage}{0.32\textwidth}
    \centering
    \includegraphics[width=\textwidth,trim=1.7cm 0cm 2.6cm 0cm,
    clip=true]{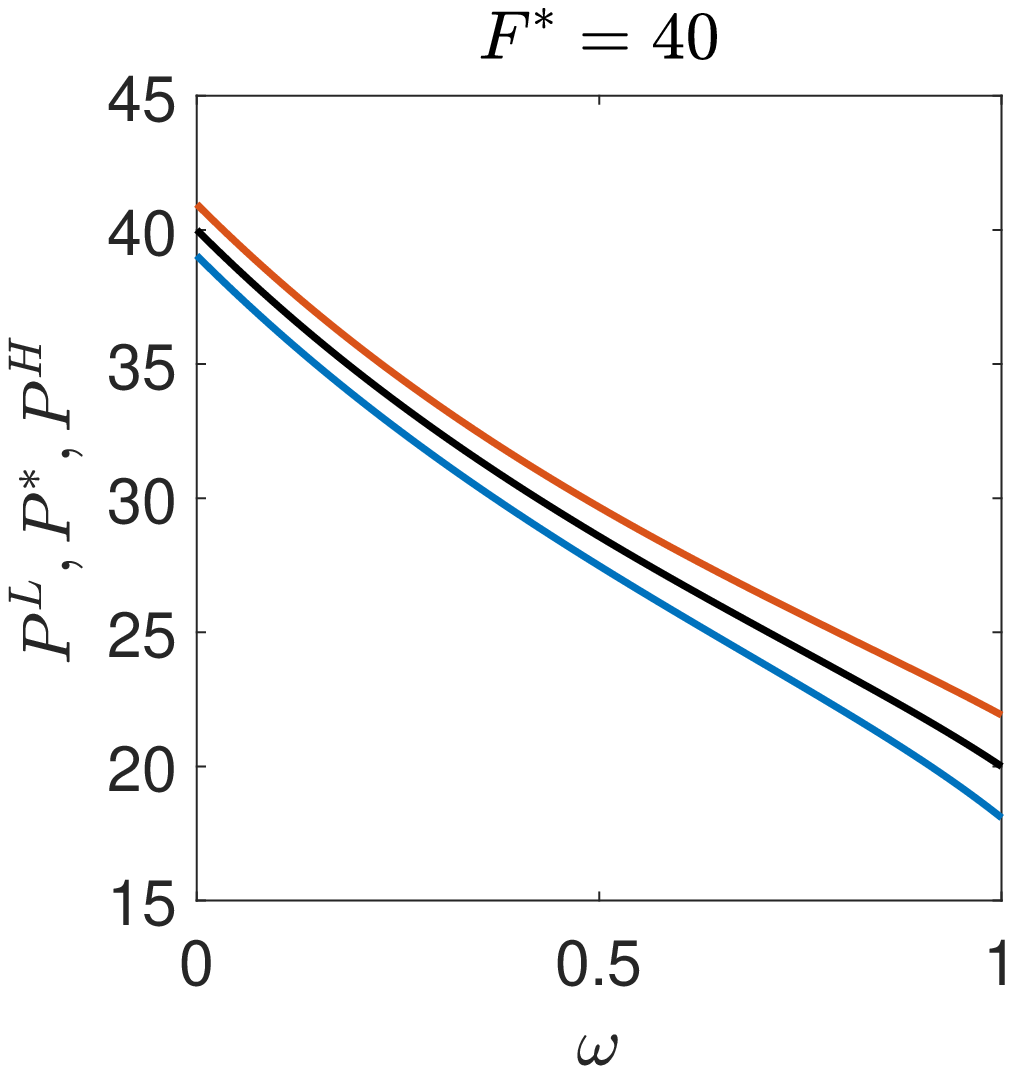}
    (C)
  \end{minipage}
\end{center}
\caption{Steady states $P^{L}$ (blue), $P^{\ast}$ (black) and $P^{H}$
  (red) on varying $\omega.$} \label{fig:scw}
\end{figure}

The previous proposition shows that the effect produced by the degree of interaction between
real and stock markets on the steady state values can be quite ambiguous. Since the
  value of $Y^{\ast}$ is the same as that resulting in \cite{cnp17} in
  the case of exogenous government expenditure, the unbiased steady
  national income is either increasing or concave unimodal, reaching a greater value in the case of completely
  integrated markets than in that of isolated ones (see
  \cite{cnp17}). Additionally, both optimistic and pessimistic steady
  national income can be either increasing or concave unimodal as well.

  Concerning prices, in Figure \ref{fig:scw} we report the evolution with respect to $\omega$
  of the steady state values $P^L,\,P^{\ast}$ and $P^H$ setting
  $A=10,c=d=h=0.5$ and $b=\beta=1,$ for different values of
  $F^{\ast}.$ We remark that $P^L$ and $P^H$ are computed numerically
  solving equation \eqref{im} through which they are implicitly defined. On
  increasing $\omega$ we may have situations in which the more the
  markets are integrated, the more price increases (Figure
  \ref{fig:scw} (A)), as well as opposite situations in which,
  increasing the market integration, price is penalized (Figure
  \ref{fig:scw} (C)), with intermediate scenarios in which a partial
  integration provides the lowest price (Figure \ref{fig:scw} (B)). In
  this last case, we stress a further element of ambiguity: if we
  compare the steady state values obtained for the same pair of
  stock and real markets when independent ($\omega=0$) and completely
  integrated ($\omega=1$), it is not possible to conclude which
    one has the largest prices. From Figure \ref{fig:scw} (B), it
  is evident that we can simultaneously have all the possible
  orderings between prices, as
  $P^L(0)>P^L(1),P^{\ast}(0)\approx P^{\ast}(1)$ and $P^H(0)<P^H(1).$\\
In the rest of the section, we focus on the local stability of the steady states. The
first group of results (see Propositions \ref{stabcond}--\ref{th:stabw}) deals with $S^{\ast}.$ We start providing
analytical conditions under which $S^{\ast}$ is locally asymptotically
stable.

\begin{proposition}\label{stabcond}
  Under assumption \eqref{a1}, System \eqref{map} is locally
  asymptotically stable at $S^{\ast}$ provided
  that
  \begin{subequations}\label{eq:sc}
    \begin{empheq}[left=\empheqlbrace]{align}
      &2b^2\beta-1<0\label{eq:sc1}\\
      &\left( 2-E\right) \left( 1+c+2\gamma\right) -\omega ^{2}dhE>0 \label{eq:sc2} \\
      &1-c+cE+\omega ^{2}dhE+\gamma c-E\gamma-E\gamma c+E^{2}\gamma -
      E^{2}\gamma^{2}+E\gamma^{2}-\gamma>0 \label{eq:sc3} \\
      &2\gamma+c-cE-\gamma E-\omega ^{2}dhE<3 \label{eq:sc4} %
    \end{empheq}%
  \end{subequations}
where $E=\mu\sigma(1-2b^2\beta).$
\end{proposition}

Before making explicit from \eqref{eq:sc} the roles of the interaction
  degree, of the intensity of choice and of the bias on the stability of $S^{\ast },$
  we note that condition \eqref{eq:sc1} actually coincides with \eqref{cond}, whose
  violation leads to the emergence of two more steady states. As it will become evident
  also from the simulations, when $S^{\ast}$ loses stability because
  of \eqref{eq:sc1}, a pitchfork bifurcation occurs and two stable
  steady
  states $S^L$ and $S^H$ emerge.\\
  We now turn our attention to the roles of $\beta$ and $b$. We
    highlight that a neutral effect on a steady state is realized when it
    is locally asymptotically stable/unstable independently of the
    parameter values; a stabilizing/destabilizing effect occurs when
    the steady state is locally asymptotically stable only above/below
    a given threshold value of the considered parameter and unstable
    otherwise; a mixed effect arises when the steady state is locally
    asymptotically stable only for intermediate parameter values,
    between two stability thresholds, and unstable otherwise.

\begin{proposition}\label{th:stabbetab}
  Under assumption \eqref{a1}, increasing $b$ or $\beta$ can have a
  mixed, destabilizing or neutral effect on $S^{\ast }$.
\end{proposition}

The previous result points out that the bias and the intensity of
choice have the same qualitative effect on the stability of
$S^{\ast}.$ If the polarization and
  the relevance of the beliefs is sufficiently increased, $S^{\ast}$ necessarily loses
  stability. Indeed, from the proof of Proposition \ref{th:stabbetab} we find that $S^{\ast}$
	can be unstable for any values of $b$ and $\beta,$ but it cannot be always locally asymptotically
	stable. The unbiased steady state
  is then weakened when increasing heterogeneity of the beliefs
  as it becomes surrounded by the
  new potentially attracting polarized steady states and because it becomes unstable, although from a static viewpoint its position is not affected by $b$ and $\beta.$
However, as we will see in Section 5, the role of the
   unbiased steady state, when unstable, may still be significant.


In the following result, we describe the possible stability scenarios
 arising for $S^{\ast}$ when $\omega$ varies.

 \begin{proposition} \label{th:stabw} Under assumption \eqref{a1},
    increasing $\omega$ on $[0,1]$ can have a neutral,
     destabilizing, stabilizing or mixed effect on $S^{\ast}.$
\end{proposition}

The results obtained in Proposition \ref{th:stabw} on the role of $\omega$ are
in agreement with those reported in \cite{naim2}, namely the possible
scenarios are the same, proving once more the ambiguous role of the
interaction degree. It is then not possible to conclude that either
encouraging or dampening the integration between the two sides of the
economy improves the stability of the markets.

To summarize the findings about the stability of $S^{\ast},$ in
Figure \ref{fig:sr} we report three different stability regions in the
$(\beta,\omega)$-plane, obtained setting $c=d=h=0.38,b=0.5$ and
$\mu=1,$ and changing $\sigma$ and $\gamma$ from time to time in order
  to highlight different stability configurations. Parameters'
combinations for which $S^{\ast}$ is locally asymptotically stable are
represented using yellow color, while cyan color is used for the
instability region. The stability thresholds corresponding to
conditions \eqref{eq:sc1}, \eqref{eq:sc2} and \eqref{eq:sc3} are
represented by solid, dashed and dotted black lines, respectively.
\begin{figure}[t]
\begin{center}
  \begin{minipage}{0.32\textwidth}
    \centering
    \includegraphics[width=\textwidth,trim=1.7cm 0cm 2.6cm 0cm, clip=true]{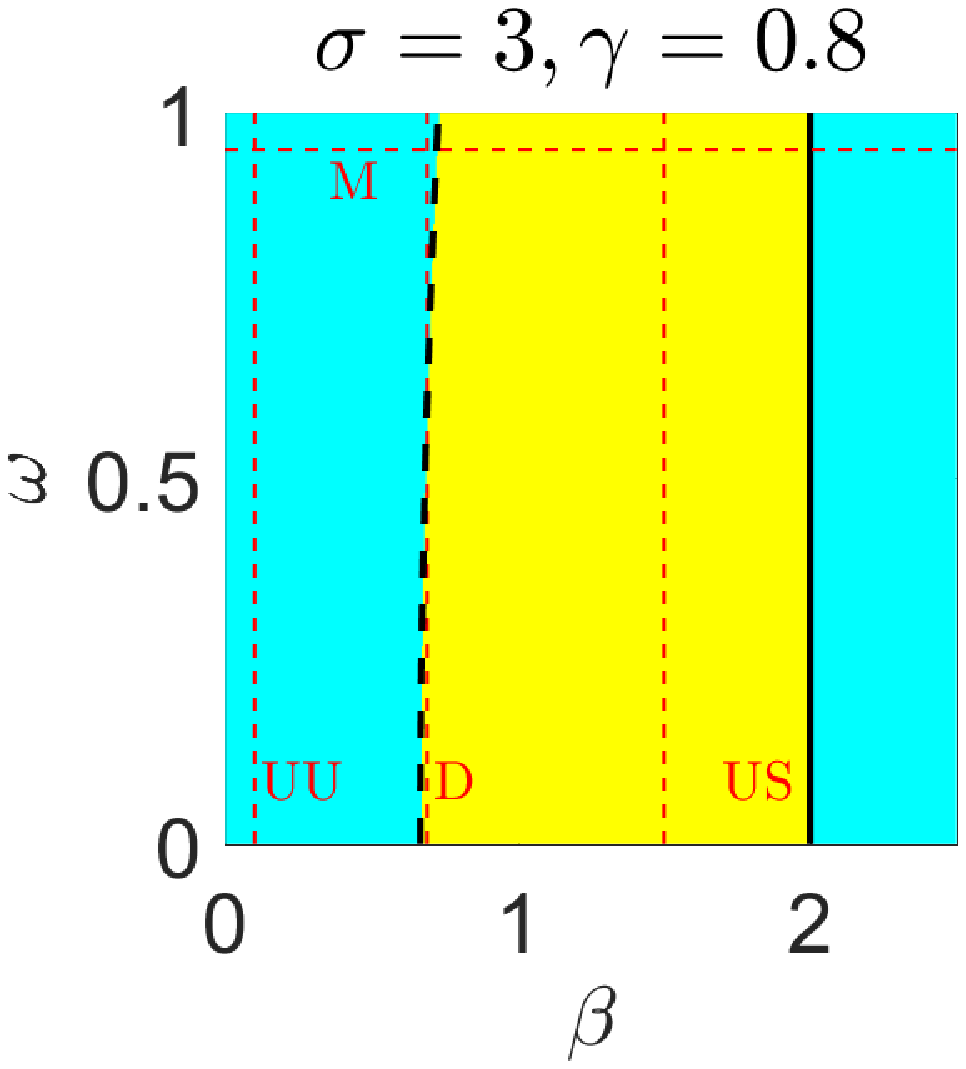}\\
    (A)
  \end{minipage}
  \begin{minipage}{0.32\textwidth}
    \centering
    \includegraphics[width=\textwidth,trim=1.7cm 0cm 2.6cm 0cm, clip=true]{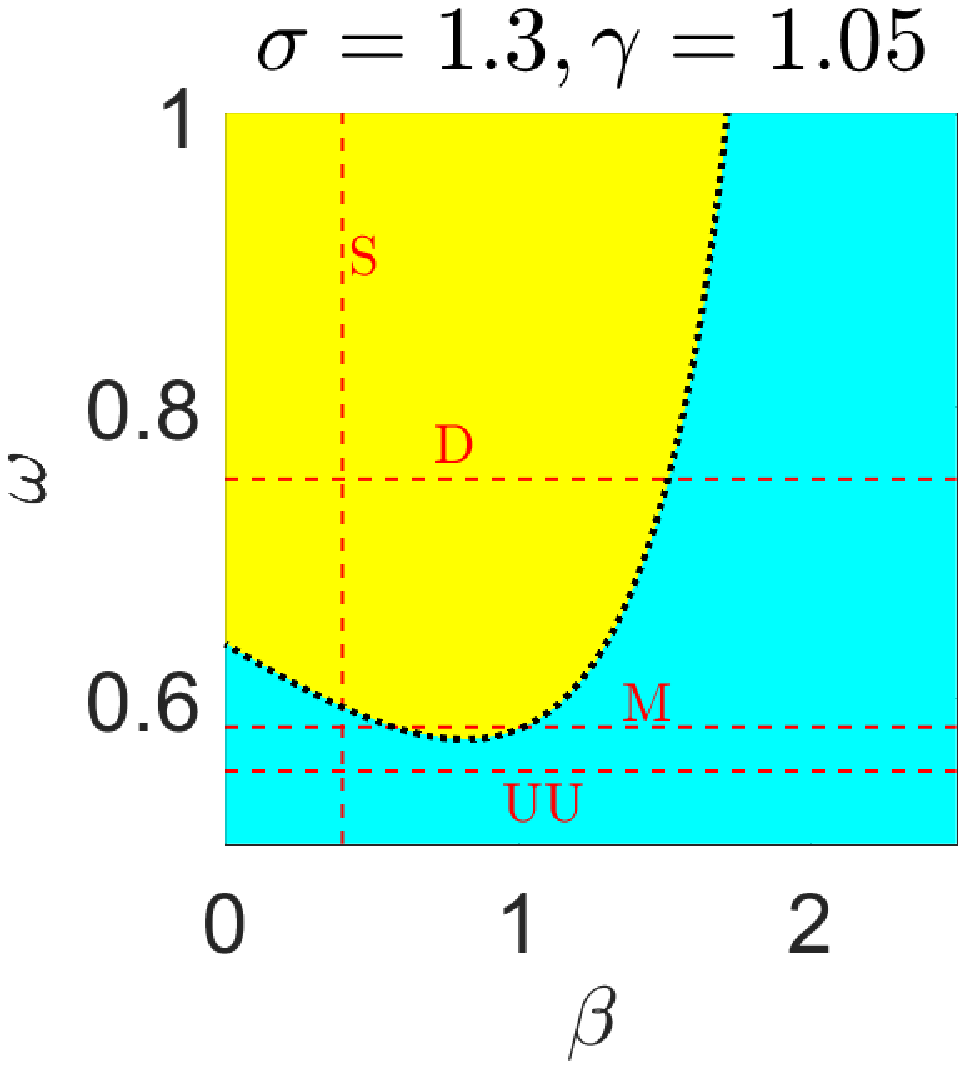}
    (B)
  \end{minipage}
  \begin{minipage}{0.32\textwidth}
    \centering
    \includegraphics[width=\textwidth,trim=1.7cm 0cm 2.6cm 0cm,
    clip=true]{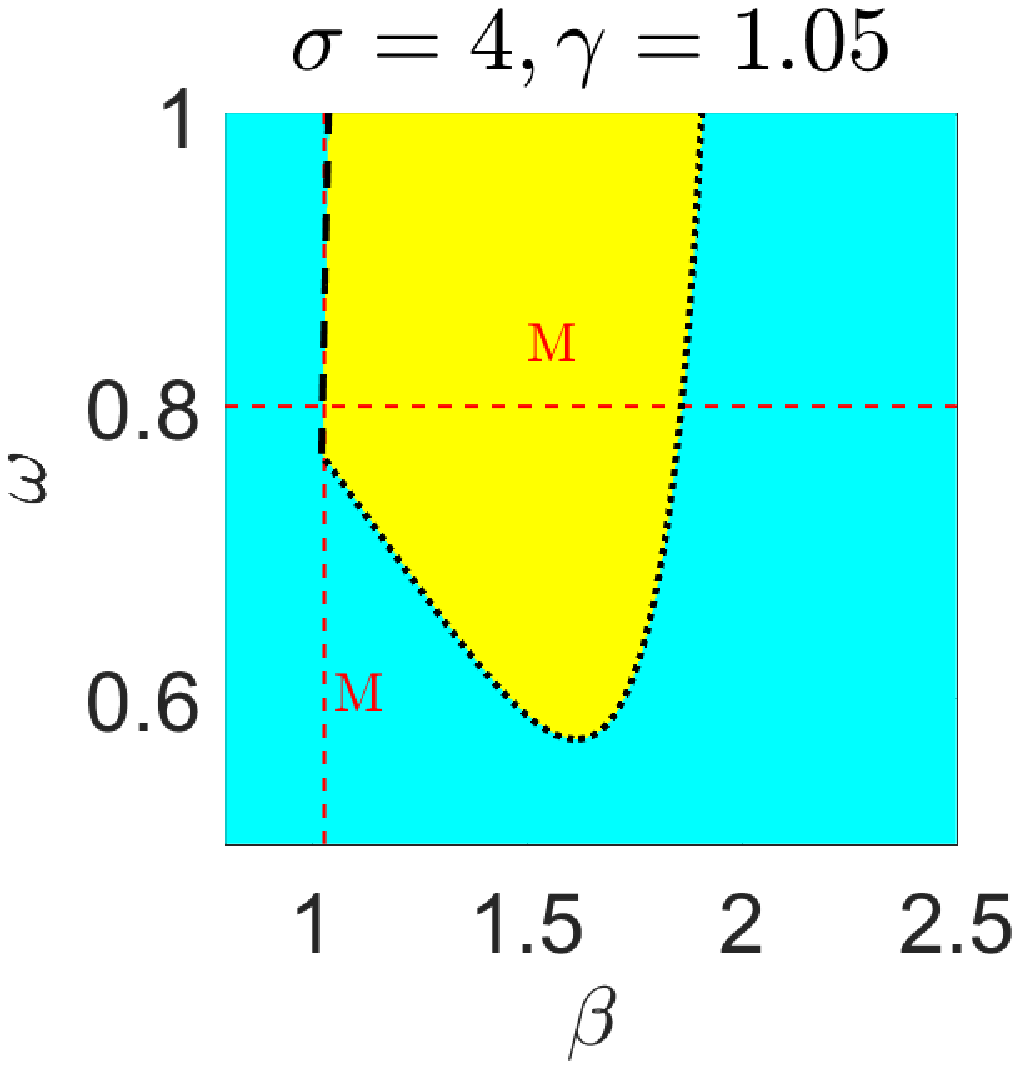}
    (C)
  \end{minipage}
\end{center}
\caption{Stability regions in the $(\beta,\omega)$-plane. Yellow and
  cyan colors respectively identify parameters' combinations for which
  $S^{\ast}$ is locally asymptotically stable and unstable. Solid,
  dashed and dotted black lines represent stability thresholds, while
  dashed red lines identify unconditionally unstable (UU),
  unconditionally stable (US), mixed (M), destabilizing (D) and
  stabilizing (S) scenarios. } \label{fig:sr}
\end{figure}
In particular, the region reported in Figure \ref{fig:sr} (A) is
obtained setting $\sigma=3,\,\gamma=0.8,$ and shows the occurrence,
on varying $\omega$, of unconditionally unstable, destabilizing and
unconditionally stable scenarios for $\beta=0.1$, $\beta=0.69$ and
$\beta=1.5,$ respectively (vertical dashed red lines). On varying
$\beta,$ we always find a mixed scenario (e.g. horizontal dashed red
line $\omega=0.95$). The stabilizing scenario (vertical dashed red
line $\beta=0.4$) on increasing $\omega$ is reported in Figure
\ref{fig:sr} (B), which is obtained setting $\sigma=1.3$ and
$\gamma=1.05.$ On varying $\beta$ we have unconditionally unstable
(e.g. when $\omega=0.55$), mixed (e.g. when $\omega=0.58$) and
destabilizing (e.g.  when $\omega=0.75$) scenarios. Finally, setting
$\sigma=4$ and $\gamma=1.05$ we obtain the stability region reported
in Figure \ref{fig:sr} (C), in which we underline the occurrence of
mixed scenarios both on varying $\omega$ (e.g. for $\beta=1.03$) and
$\beta$ (e.g. for $\omega=0.8$). The previous considerations allow
concluding that each scenario predicted by Propositions \ref{th:stabbetab} and \ref{th:stabw}
is actually possible. We
stress that the mixed scenarios on varying $\beta$ pointed out in each of
Figures \ref{fig:sr} (A)--(C) are obtained crossing different couplings
of thresholds.

In the next proposition we show the possible effects of $\omega$ and
$\beta$ on the local asymptotic stability of $S^H$ and
$S^L,$ provided that they exist and are positive.
\begin{proposition}\label{th:stabhl}
  Let assumptions \eqref{a1} and \eqref{a2} hold
  true.  If, keeping $\omega$ fixed and
  increasing $\beta,$ for $S^{\ast}$ we have\\
  $\bullet$ a destabilizing scenario, then on increasing
  $\beta\in(1/(2b^2),+\infty)$ we can either have a stabilizing or an
  unconditionally stable scenario on $S^H$
  and $S^L;$\\
  $\bullet$ a mixed scenario, then on increasing
  $\beta\in(1/(2b^2),+\infty)$ we can either have a destabilizing or a
  mixed scenario on
  $S^H$ and $S^L;$\\
  $\bullet$ an unconditionally unstable scenario, then $S^H$
  and $S^L$ are unconditionally unstable, too.\\
  Moreover, on increasing $\omega,$ there is a one to one
    correspondence between the scenarios for $S^{\ast}$ for
    $\beta<1/(2b^2)$ and those for $S^{H}$ and $S^L$ for
    $\beta>1/(2b^2).$
\end{proposition}

The previous proposition shows that the stability of $S^L$ and $S^H$
is strongly connected to that of $S^{\ast}.$ From the mathematical
point of view, we actually have that the stability behavior of $S^L$
and $S^H$ is a ``specular reflection'' of that of $S^{\ast}$ with
respect to $\beta=1/(2b^2).$ Decreasing $\beta$ has the same effect on
the stability of $S^{\ast}$ of increasing it on the stability of $S^L$
and $S^H.$ Thanks to Proposition \ref{th:stabhl}, we can highlight a
substantial difference with the scenarios of the economic setting
reported in \cite{naim2}: in the present case, for increasing values
of the intensity of choice, when $S^{\ast}$ loses stability (both in
the destabilizing and in the mixed scenarios) we can still have
convergence toward a steady state. Such framework occurs when $S^{\ast}$ undergoes
a pitchfork bifurcation and the new steady states $S^{L}$ and $S^{H}$
are locally stable.
This means that, in such scenario, the final
outcome is always a steady state, which, even if it changes, may be
locally stable. The most extreme situation is
realized when a destabilizing scenario for $S^{\ast}$ is followed by
an unconditionally stable one for $S^L$ and $S^H.$ Accordingly, convergence
toward a steady state may occur for \textit{any} value of $\beta.$



\section{Numerical simulations}
In this section we present the results of several numerical
investigations, collected pursuing a double aim. Firstly, we wish to
complement the analysis of Section \ref{sec:an}. In particular, we
want to understand what kinds of dynamics occur when
$S^{\ast}$ becomes unstable, as well as to investigate the possible
scenarios arising with the emergence of $S^L$ and $S^H.$ Secondly, we
want to study the qualitative properties of the time series when
exogenous non-deterministic effects are taken into account, that is,
considering a stochastically perturbed version of the model, in
agreement with \cite{cnp17, franke,naim2}.  In continuity with Section
\ref{sec:an}, we focus on the effect of the degree of
interaction $\omega$ and of the intensity of choice $\beta.$\\

In order to perform simulations, we need to specify the analytical
expressions of $g_I$ and $g_P.$ In what follows, we use the same
sigmoid function considered in \cite{CN16}, i.e.,
\begin{equation}\label{sigmoid}
  g_X (z) =\begin{cases}
    a_1^X\tanh\left(\dfrac{z}{a_1^{X}}\right) & \text{ if } z\geq 0,\\
    a_2^X\tanh\left(\dfrac{z}{a_2^{X}}\right) & \text{ if } z< 0,
  \end{cases}
\end{equation}
where $X\in\{I,P\}$ and $a_1^X,a_2^X$ are positive parameters, with
$a_1^X$ as upper bound when $z\to +\infty$ and $-a_2^X$ as lower bound
when $z\to -\infty$. A straightforward check shows that function $g_X$
belongs to $\mathcal{C}^2(\mathbb{R})$ and satisfies assumptions
\eqref{eq:sp}. In all the simulations reported in the present section
we set
$F=A=15$, $c=d=h=0.38$, $b=0.5$ and $\mu=1.$\\
Concerning the upper and lower bounds $a_i^X,\,i\in\{1,2\},$ of
functions $g_I$ and $g_P,$ we start by recalling that the steady state
values of $P$ and $Y$ significantly change on varying $\omega.$ Since
such bounds represent the maximum possible positive or negative
variation of $P_t$ and $Y_t,$ it is reasonable to assume that they are
qualitatively connected to the magnitude of prices and national
incomes. In fact, it makes economic sense to think that investments
and prices can not grow indefinitely but that, instead, their
variation is linked to the pace of economic activity, reflected in the
level of prices and national income, which are in turn connected to
the amount of resources existing in a certain country. To encompass
such fact into the simulations, we let $a_i^X$ proportionally change
depending on the value of the steady state $S^{\ast},$ setting
  $a_1^P=2P^{\ast}/F,\,a_2^P=4P^{\ast}/F,\,a_1^I=3Y^{\ast}(1-c)/A$ and
  $a_2^I=6Y^{\ast}(1-c)/A.$ In this way, for $\omega=0$ we have
$a_1^P=2,\,a_2^P=4,\,a_1^I=3,\,a_2^I=6,$ while they increase/decrease,
on varying $\omega,$ proportionally to $P^{\ast}$ and $Y^{\ast}$ (see
Proposition \ref{prop3}). We recall that we checked the robustness of
the simulations by varying all the parameters in suitable ranges,
obtaining results that are qualitatively comparable with those
reported in the next subsections. We also point out that the reported two-dimensional
  bifurcation diagrams are obtained setting the same initial datum for
  each parameters' coupling. Conversely, in order to highlight
  coexistence, one-dimensional bifurcation diagrams are all obtained
  ``following the attractor''. This means that, having a strictly
  increasing or decreasing sequence of parameters $a_i$,
  $i=1,\dots,m,$ we set the initial datum for the first simulation
  equal to $a_1,$ while the initial datum for the simulation with
  $a_{i+1}$ is set suitably close to the solution of the simulation
  obtained with $a_i$.
\subsection{Deterministic simulations}
\begin{figure}[t]
\begin{center}
  \begin{minipage}{0.4\textwidth}
    \centering
    \includegraphics[width=\textwidth,trim=0.7cm 0cm 1.9cm 0cm, clip=true]{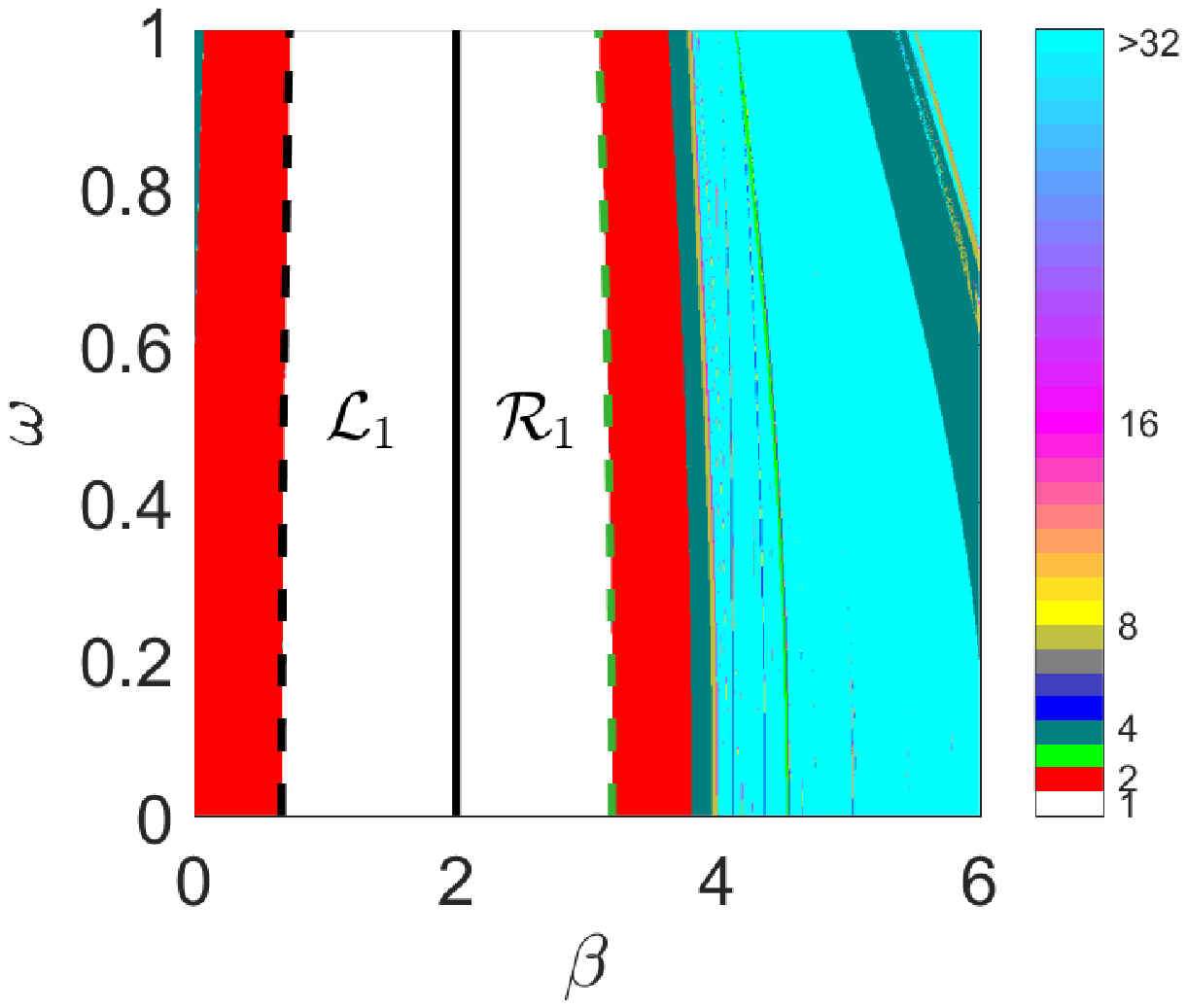}\\
    (A)
  \end{minipage}
  \begin{minipage}{0.4\textwidth}
    \centering
    \includegraphics[width=\textwidth,trim=0.8cm 0cm 2cm 0cm, clip=true]{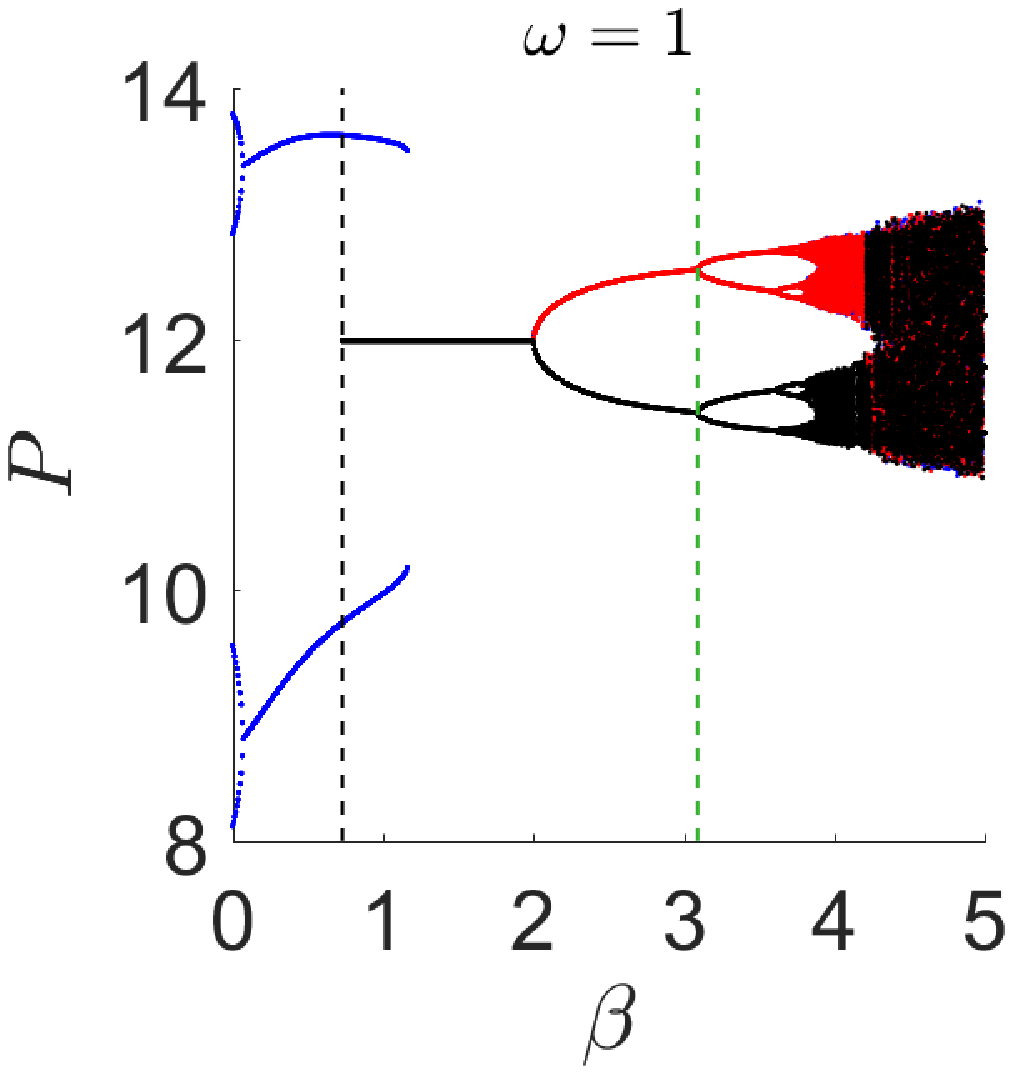}\\
    (B)
  \end{minipage}
\end{center}
\caption{\small Simulations obtained setting $\sigma=3$ and
  $\gamma=0.8.$ (A): Two-parameters bifurcation diagram. White color
  identifies parameters' combinations for which the initial datum
  converges toward a steady state, other colors are used for
  attractors consisting of more than a single point (cyan is used for
  attractors consisting of more than 32 points). Black lines denote
  the stability thresholds of $S^{\ast},$ while the green line denotes
  the stability threshold of $S^{H}.$ (B): Bifurcation diagrams for $P$
  on varying $\beta$ when $\omega=1.$ Different colors correspond to
  different initial conditions.} \label{fig:sim1}
\end{figure}
\begin{figure}[t!]
\begin{center}
\begin{minipage}{0.32\textwidth}
    \centering
    \includegraphics[width=\textwidth,trim=0.8cm 0cm 2.2cm 0cm, clip=true]{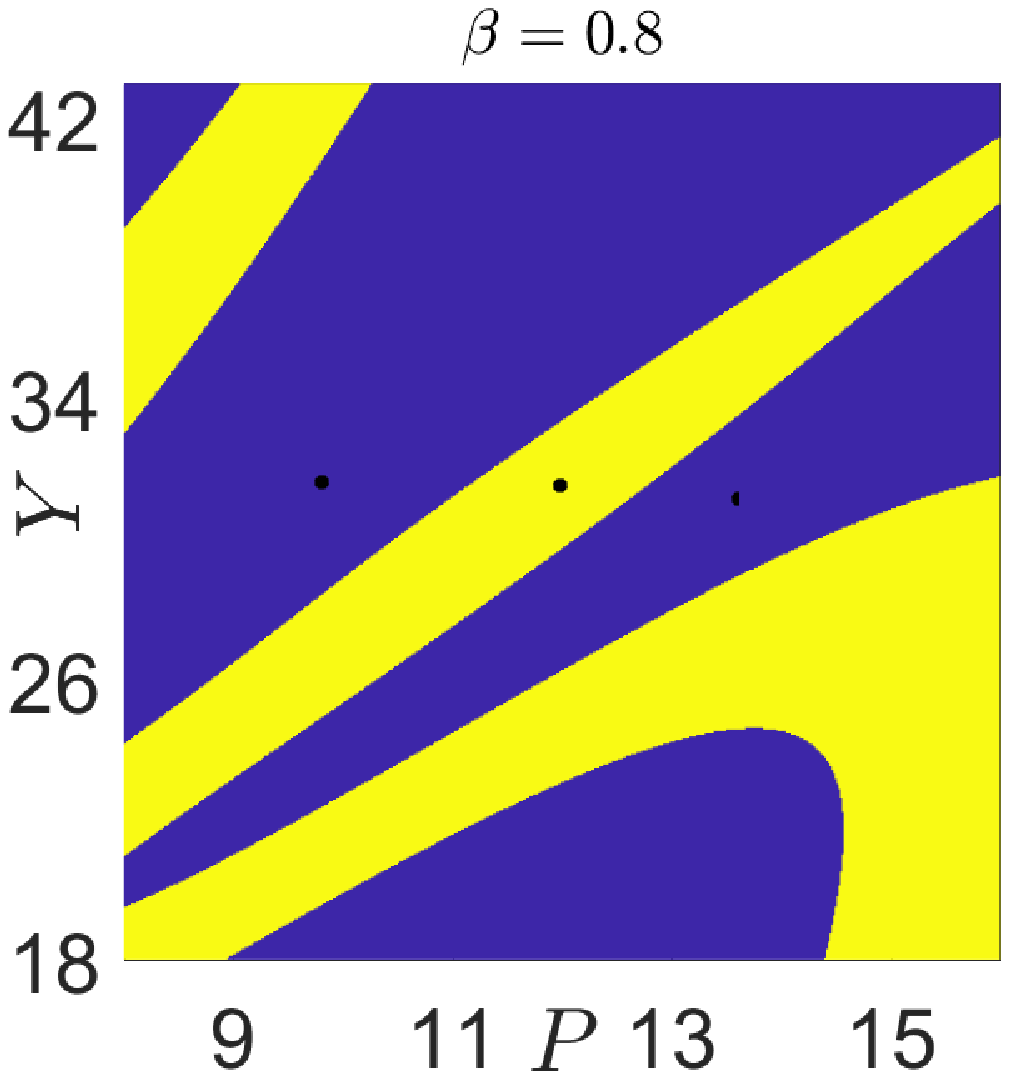}
    (A)
  \end{minipage}
  \begin{minipage}{0.32\textwidth}
    \centering
    \includegraphics[width=\textwidth,trim=1.3cm 0cm 2.2cm 0cm, clip=true]{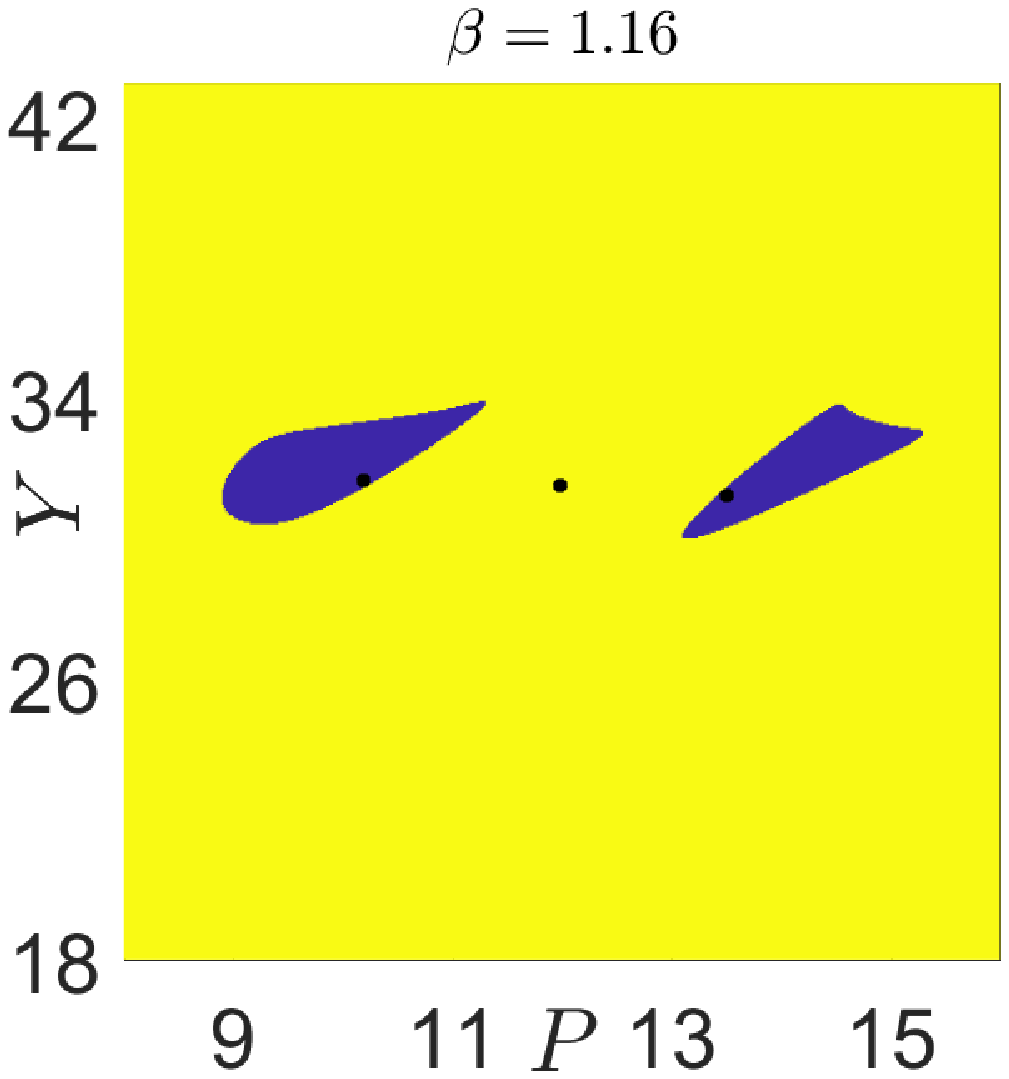}\\
    (B)
  \end{minipage}
  \begin{minipage}{0.32\textwidth}
    \centering
    \includegraphics[width=\textwidth,trim=0.8cm 0cm 2.2cm 0cm, clip=true]{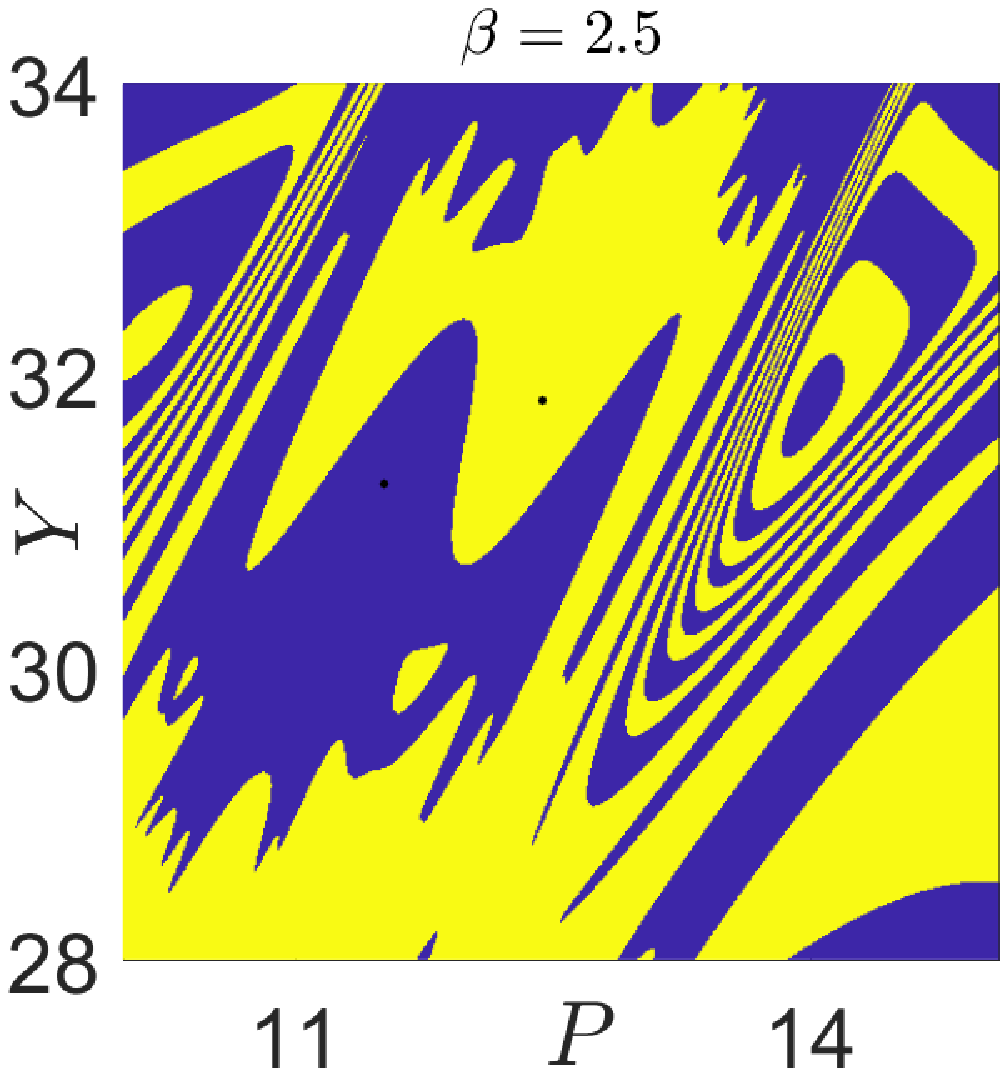}
    (C)
  \end{minipage}\\
  \begin{minipage}{0.32\textwidth}
    \centering
    \includegraphics[width=\textwidth,trim=1.3cm 0cm 2.2cm 0cm, clip=true]{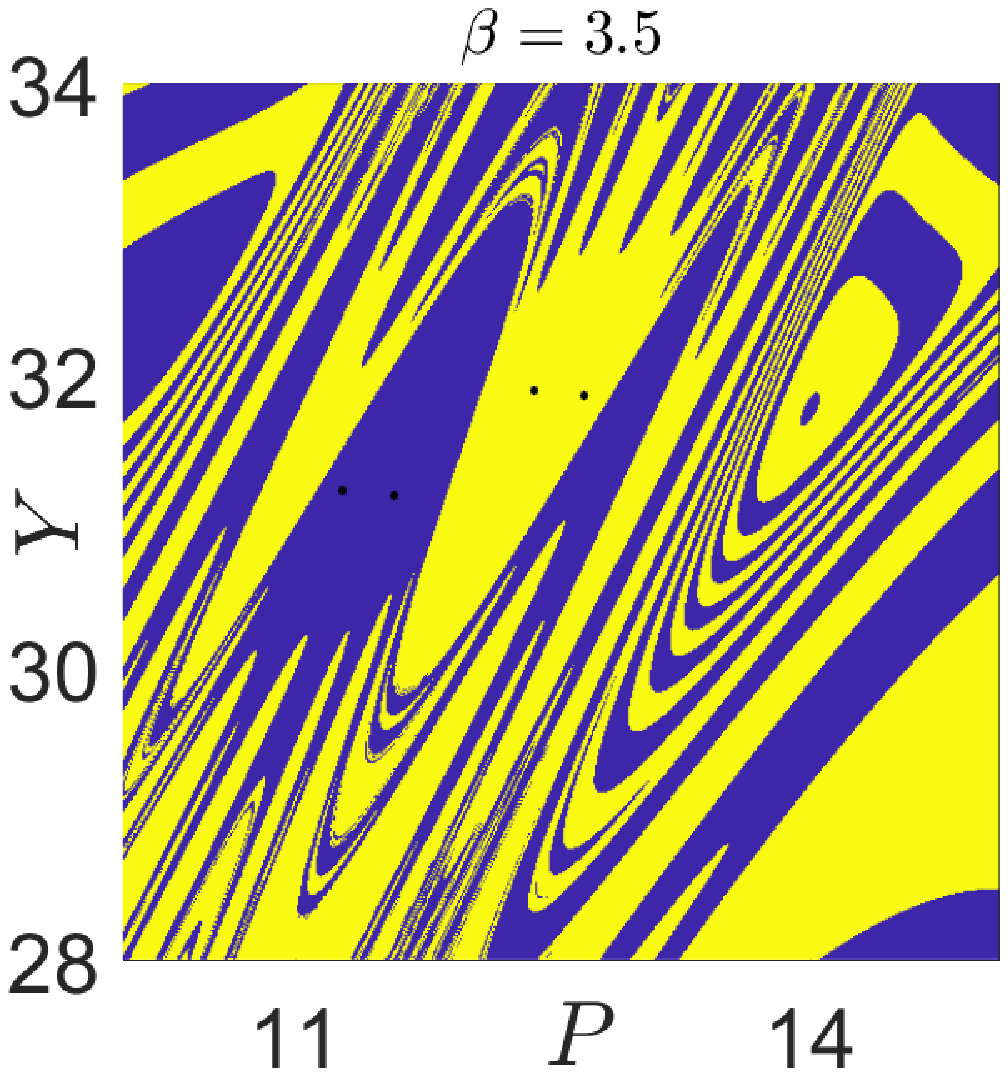}\\
    (D)
  \end{minipage}
  \begin{minipage}{0.32\textwidth}
    \centering
    \includegraphics[width=\textwidth,trim=0.8cm 0cm 2.2cm 0cm, clip=true]{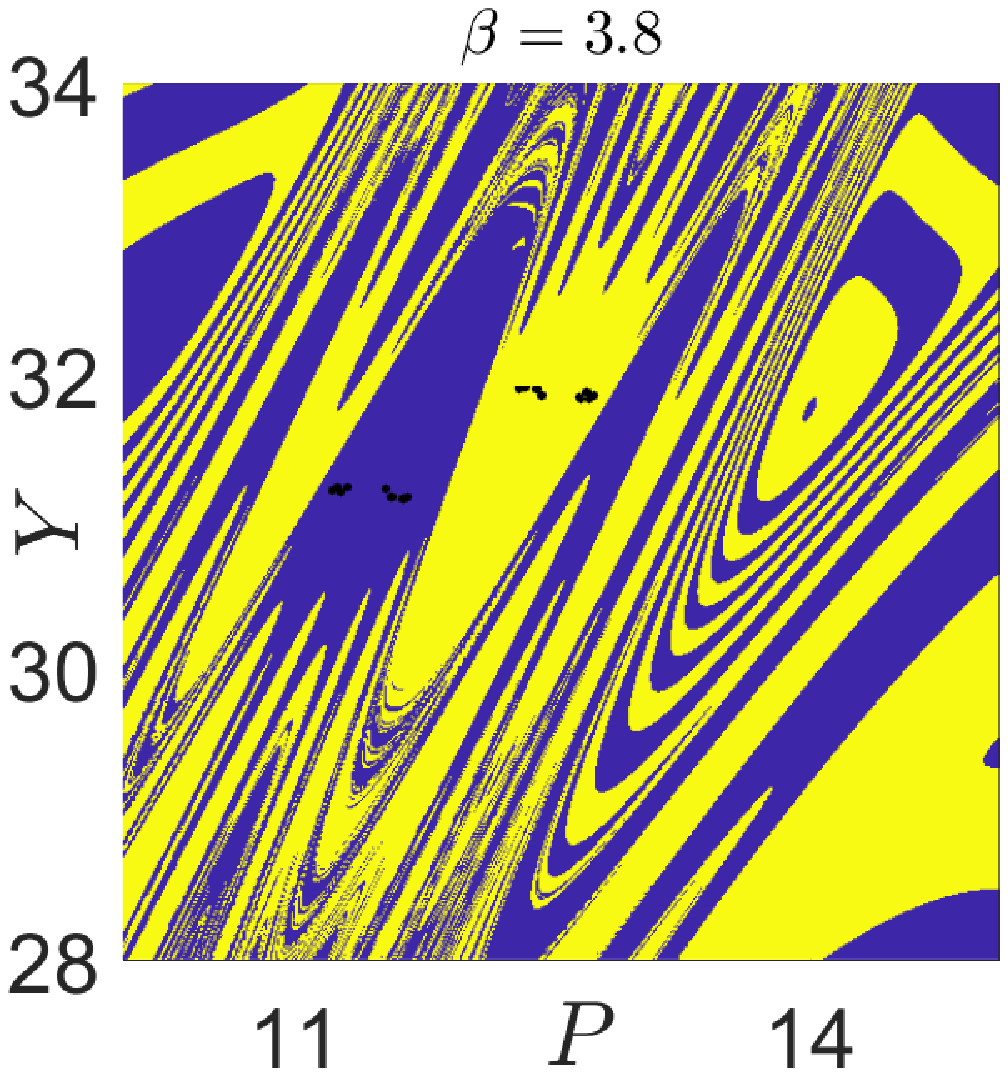}\\
    (E)
  \end{minipage}
  \begin{minipage}{0.32\textwidth}
    \centering
    \includegraphics[width=\textwidth,trim=1.3cm 0cm 2.2cm 0cm, clip=true]{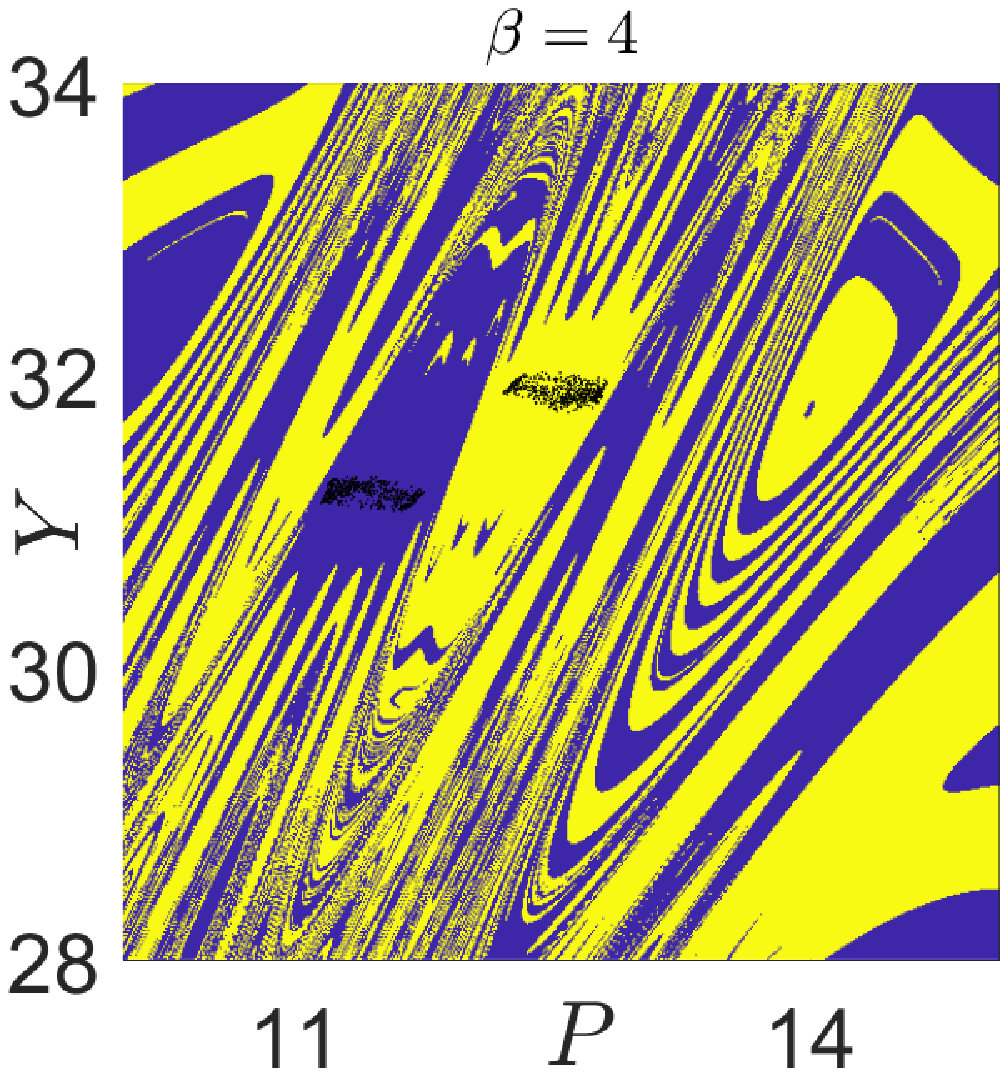}\\
    (F)
  \end{minipage}
\end{center}
\caption{\small Evolution of the basins of attraction on varying the intensity of choice $\beta$. In (A)-(B) the basins of attraction of the steady state $S^{\ast}$ is represented in yellow together with the basin of the coexisting 2-cycle, in violet. Panel (C) refers to the basins of the biased steady states $S^L$ and $S^H$ whose birth is due to the occurrence of the pitchfork bifurcation which makes the unbiased steady state $S^{\ast}$ unstable. In panel (D) the basins of the two 2-cycles arising at the destabilization of $S^L$ and $S^H$ are depicted, which then turn into chaotic attractors due to a cascade of further period-doubling bifurcations for increasing values of $\beta$, as panels (E)-(F) show.}
\label{fig:bas1}
\end{figure}
We shall now explain the functioning of the model by considering
different sets of simulations, which provide a portrait of the
possible dynamics that may occur in our model economy, emphasizing how
the relevant parameters affect the evolution of the economic
variables. In particular, we aim to account for the richness of the
possibly complex dynamical behaviors arising in unstable regimes,
which however keep evident track of the simple distinctive
characterization of the
beliefs, in terms of optimistic/pessimistic bias.\\
In the first set of simulations we set $\sigma=3$ and $\gamma=0.8,$
and the resulting stability region for $S^{\ast}$ is that reported in
Figure \ref{fig:sr} (A).  In Figure \ref{fig:sim1} (A) we show the
two-dimensional bifurcation diagram in the $(\beta,\omega)$-plane,
where each point is represented using a different color depending on
the number of points of the attractor that was approached starting by
the initial datum
$(Y_0,P_0,Z_0)=(Y^{\ast}+0.001,P^{\ast}+0.001,Z^{\ast}+0.001)$ after
$T=10000$ iterations. In particular, white color denotes parameter
combinations for which convergence is toward a one-point attractor,
i.e., a steady state, red color is used for period-2 cycles, cyan
color for attractors consisting of more than 32 points (namely, high
period cycles, chaotic attractors or closed invariant curves), and the
remaining colors for cycles with periods between 3 and 32. Black solid
and dashed lines respectively represent the stability thresholds
corresponding to \eqref{eq:sc1} and \eqref{eq:sc2}, delimiting the
region within which (see also the yellow set in Figure \ref{fig:sr}
(A)) the steady state $S^{\ast}$ is stable and dynamics converge
toward it (white region $\mathcal{L}_1$). If we start from any
$(\beta,\omega)\in\mathcal{L}_1$ and we increase $\beta,$ we exit the
stability region $\mathcal{L}_1$ as $S^{\ast}$ crosses the black solid
line, but in the white region $\mathcal{R}_1$ dynamics again converge
toward a steady state. In fact, according to Propositions \ref{ss} and
\ref{stabcond}, such black line, whose equation is given by $\beta=2$
since $b=0.5$, represents both the bound for the stability condition
\eqref{eq:sc1} for $S^{\ast}$ and the threshold to the left/right of
which we have a unique/three coexisting steady states. As a
consequence, when the solid line $\beta=2$ is crossed on increasing
$\beta$, the steady state $S^{\ast}$ becomes unstable by means of a
pitchfork bifurcation and the two stable steady states $S^L$ and $S^H$
arise. We stress that, with the present choice of the initial datum,
dynamics converge toward $S^H$ for all
$(\beta,\omega)\in\mathcal{R}_1.$

Region $\mathcal{R}_1$ is the specular of $\mathcal{L}_1$ with respect
to $\beta=2.$ In particular, the green dashed line is obtained starting
from the black dashed threshold line and numerically solving the implicit
equations used in the proof of Proposition \ref{th:stabhl}. This,
in agreement with Proposition \ref{th:stabhl}, confirms that region
$\mathcal{R}_1$ is the region of local asymptotic stability for
$S^{H}.$

As remarked in the comments to Figure \ref{fig:sr} (A), on increasing
$\beta$ we have mixed scenarios for $S^{\ast}$ for any $\omega.$ From
Proposition \ref{th:stabhl}, since the right stability threshold for
$S^{\ast}$ is described by the ``pitchfork'' condition \eqref{eq:sc1},
this is the only situation in which a mixed scenario for $S^{\ast}$
``translates'' into a destabilizing scenario for $S^{H}.$ As described
after Proposition \ref{th:stabhl}, this is simply due to the fact
that, crossing $\beta=2$ from right to left, $S^H$ disappears.  In
fact, starting from any $(\beta,\omega)\in\mathcal{R}_1,$ for
increasing values of $\beta$ we cross the green stability threshold,
while decreasing $\beta$ we leave $\mathcal{R}_1$ entering
$\mathcal{L}_1.$ To be more precise about the possible dynamics
occurring when steady states lose stability, we can see that when we
leave both $\mathcal{L}_1$ and $\mathcal{R}_1$ crossing a dashed
stability threshold, we enter one of the red regions, meaning that
dynamics, previously converging to a steady state, now converge to a
period-2 cycle. \\
In Figure \ref{fig:sim1} (B) we report three bifurcation diagrams
obtained for $\omega=1$\footnote{We checked that for
  $\omega<1$ we obtain bifurcation diagrams that are qualitatively the
  same of that reported in Figure \ref{fig:sim1} (B).}. The blue
bifurcation diagram is obtained setting
$(Y_0,P_0,Z_0)=(Y^{\ast}+0.001,P^{\ast}+0.001,Z^{\ast}+0.001)$ and on
varying $\beta$ between $0$ and $1.17,$ while the red and the black
bifurcation diagrams are obtained respectively setting
$(Y_0,P_0,Z_0)=(Y^{\ast}+0.001,P^{\ast}+0.001,Z^{\ast}+0.001)$ and
$(Y_0,P_0,Z_0)=(Y^{\ast}-0.001,P^{\ast}-0.001,Z^{\ast}-0.001),$ and on
varying $\beta$ between $0.73$ and $5.$

We can see that as long as $\beta< 0.728$ (black dashed line), but close to such value,
dynamics converge toward a period-2 cycle
which 
attracts almost all the trajectories\footnote{More precisely, when
  $(\beta,\omega)$ belongs to the dashed lines in Figure
  \ref{fig:sim1} (A) a flip bifurcation occurs. In fact, for the
  present parameters' configuration, on such dashed line, condition
  \eqref{eq:sc2} becomes an equality, while the remaining stability
  conditions in Proposition \ref{stabcond} are satisfied. Condition
  \eqref{eq:sc2} is violated when $p(-1)=0,$ being $p$ the
  characteristic polynomial of the Jacobian matrix $J^{\ast}$
  evaluated at $S^{\ast}.$ Since the largest eigenvalue modulus of
  $J^{\ast}$ is equal to $-1,$ a flip bifurcation occurs.}. Such
period-2 cycle coexists with the stable steady state $S^{\ast}$ on an
interval of values of $\beta,$ and then the former disappears when its
basin of attraction shrinks and is absorbed by that of the steady
state. This coexistence is also reported in Figures \ref{fig:bas1}
(A)-(B) where the basin of attraction\footnote{We remark that Figure
  \ref{fig:bas1} reports the slice on plane $Y=Z$ of the
  three-dimensional basins of attraction for some increasing
  values of $\beta,$ together with the projection on such plane of the attractors toward which the initial data converged.} of $S^{\ast}$ is represented in yellow while the basin of the 2-cycle is depicted in violet. The two panels show that, as $\beta$ grows, the basin of the 2-cycle reduces in size and for further increases of the intensity of choice parameter there is a contact between the periodic points and their basin boundary leading to the disappearance of the 2-cycle.\\
From the red and the black bifurcation diagrams of Figure
\ref{fig:sim1} (B), it is evident the pitchfork bifurcation occurring
at $\beta=2$ with the red and the black lines springing from
$S^{\ast}$ that respectively represent $S^H$ and $S^L$ (see Figure
\ref{fig:bas1} (C), which reports the basins of attraction of the
steady states born via the pitchfork bifurcation). These are the
biased steady states, whose birth via pitchfork bifurcation can be
read as the effect of agents' decision when the intensity of choice is
high enough. In this case almost all agents adopt the optimistic or
pessimistic predictor, which is the best performing predictor in terms
of forecast error, leading the dynamics to converge to an equilibrium
different from $S^{\ast}.$ We stress that, according to Proposition
\ref{scbeta}, all components of $S^H$ and $S^L$ are respectively
increasing and decreasing with respect to $\beta.$ If $\beta$ is
further increased, they both lose stability through a flip bifurcation
and then a cascade of period-doubling bifurcations leads to chaos. In
fact, from Figure \ref{fig:sim1} (D) on we observe the basins of
attraction of the couple of 2-cycles born when $S^L$ and $S^H$ lose
stability. It is also worth remarking that Proposition \ref{th:stabhl}
guarantees that the first period-doubling bifurcation occurs for the
same $\beta$ value both for $S^L$ and $S^H,$ and Figure \ref{fig:sim1}
(B) suggests that the subsequent period-doubling bifurcations are
simultaneous, too. Still increasing $\beta$ has the effect of making
these 2-cycles unstable, leading to the emergence of chaotic
attractors (see Figures \ref{fig:sim1} (E)-(F)) associated with
erratic dynamics in the course of price and national income. In this
scenario, on increasing $\beta,$ the intermediate unbiased steady
state loses its relevance when becomes unstable, as almost all
trajectories converge toward the optimistic or the pessimistic
attractor.

We now focus on the role of the degree of interaction $\omega$,
considering vertical sections of the bifurcation diagram in Figure
\ref{fig:sim1} (A).
The destabilizing scenario for both $S^{\ast}$ and
$S^H$ only leads to the emergence of a period-2 cycle (e.g. for
$\beta=0.7$). It is worth noticing that the possible effects on
$S^H$ of increasing the interaction degree when
$\beta>2$ are again specular with respect to those on
$S^{\ast}$ when $\beta<2.$ In fact, we can
identify neutrally stable (e.g. for $\beta=2.5$) or unstable
(e.g. for $\beta=4$) scenarios, as well as destabilizing scenarios
(e.g. for $\beta=3.2$).

The second family of simulations we consider is obtained setting
$\sigma=1.3$ and $\gamma=1.05.$ The corresponding stability region for
$S^{\ast}$ is depicted in Figure \ref{fig:sr} (B) and the
corresponding two-dimensional bifurcation diagram is reported in
Figure \ref{fig:sim2} (A), in which the dotted black line represents
the stability threshold \eqref{eq:sc3}. Differently from the first set
of simulations, condition $\beta=2$ (represented using blue color)
acts no more as a stability threshold, but it still shapes regions in
which we have either one or three steady states. The initial datum is
always set equal to
$(Y_0,P_0,Z_0)=(Y^{\ast}+0.001,P^{\ast}+0.001,Z^{\ast}+0.001)$ and in
region $\mathcal{R}_2$ the convergence is toward $S^H,$ while in
region $\mathcal{L}_2$ convergence is toward
$S^{\ast}.$\\
In this case, we can have either convergence toward a steady state (in
regions denoted by $\mathcal{L}_2$ and $\mathcal{R}_2$) or toward an
attractor consisting of more than $32$ points (cyan regions). Both
couple of regions are specular with respect to the vertical blue line
$\beta=2,$ in agreement with Proposition \ref{th:stabhl}, and,
depending on the value of the interaction degree, unconditionally
unstable, mixed and destabilizing scenarios for
$S^{\ast}$ respectively correspond to
unconditionally unstable, mixed and stabilizing scenarios for
$S^H.$ This means that the emergence of new steady
states makes possible to have convergence toward a stable steady state
as $\beta\rightarrow+\infty.$

If we start from $(\beta,\omega)\in\mathcal{L}_2$ and we increase (or,
depending on $\omega,$ also if we decrease) $\beta$, the crossing of
the black dotted line makes $S^{\ast}$ unstable and immediately
trajectories converge toward an attractor consisting of more than 32
points (represented by cyan color). The same happens in relation to
$S^H$ if we start from $(\beta,\omega)\in\mathcal{R}_2$ and we
decrease (or, depending on $\omega,$ also if we increase) $\beta,$
crossing the green dotted line (which, like in Figure \ref{fig:sim1}
(A), is obtained starting from the black dotted threshold and
numerically solving the implicit equations used to prove Proposition
\ref{th:stabhl}). This stability loss for $S^{\ast}$ and $S^H$ is
associated with the occurrence of Neimark-Sacker bifurcations, as it
looks evident evident also from Figure \ref{fig:sim2} (B), where we
report three bifurcation diagrams on varying $\beta$ for
$\omega=0.575,$ which highlight the presence of a mixed scenario for
both $S^{\ast}$ and $S^H$. We stress that we checked that
qualitatively similar results are obtained for all the values of
$\omega$ for which a mixed scenario arises.
\begin{figure}[t]
\begin{center}
  \begin{minipage}{0.4\textwidth}
    \centering
    \includegraphics[width=\textwidth,trim=0.8cm 0cm 2.2cm 0cm, clip=true]{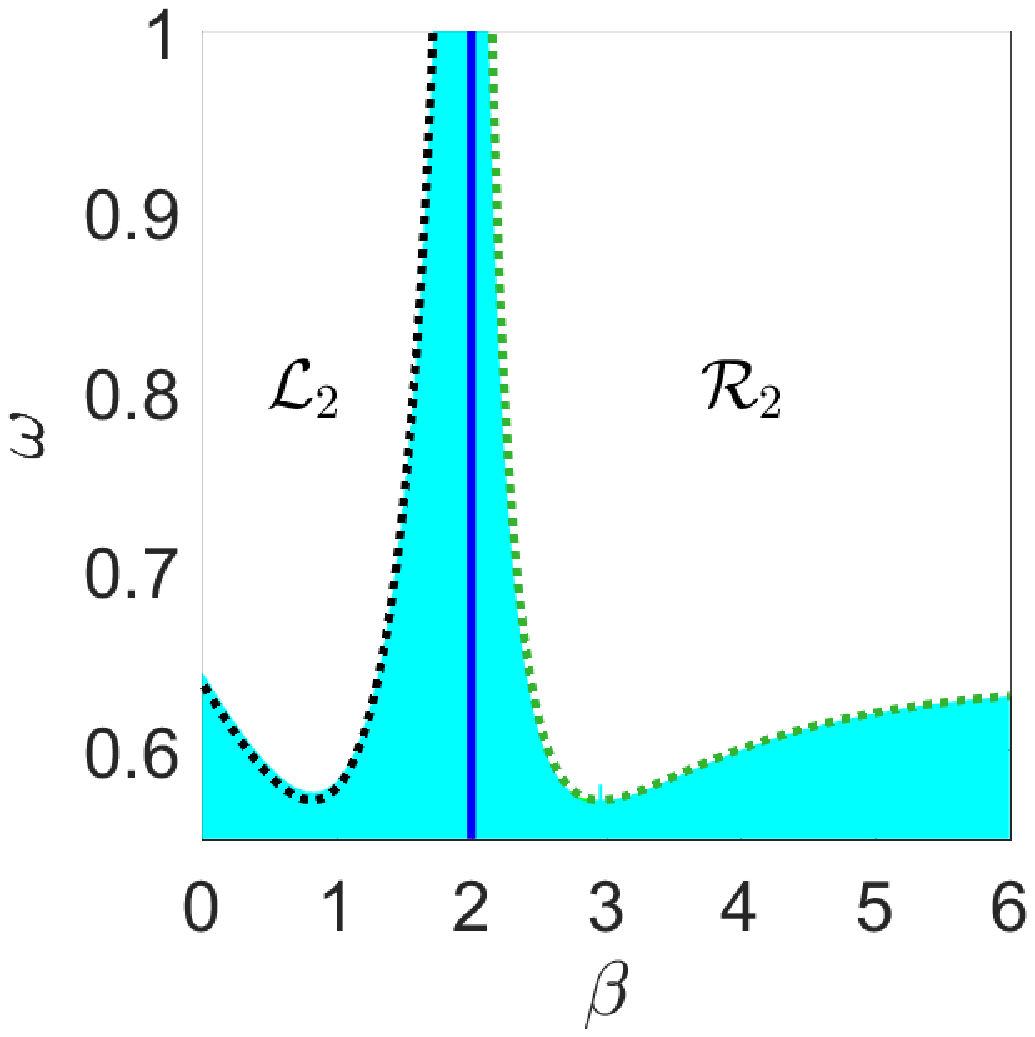}\\
    (A)
  \end{minipage}
  \begin{minipage}{0.4\textwidth}
    \centering
    \includegraphics[width=\textwidth,trim=0.8cm 0cm 2.2cm 0cm, clip=true]{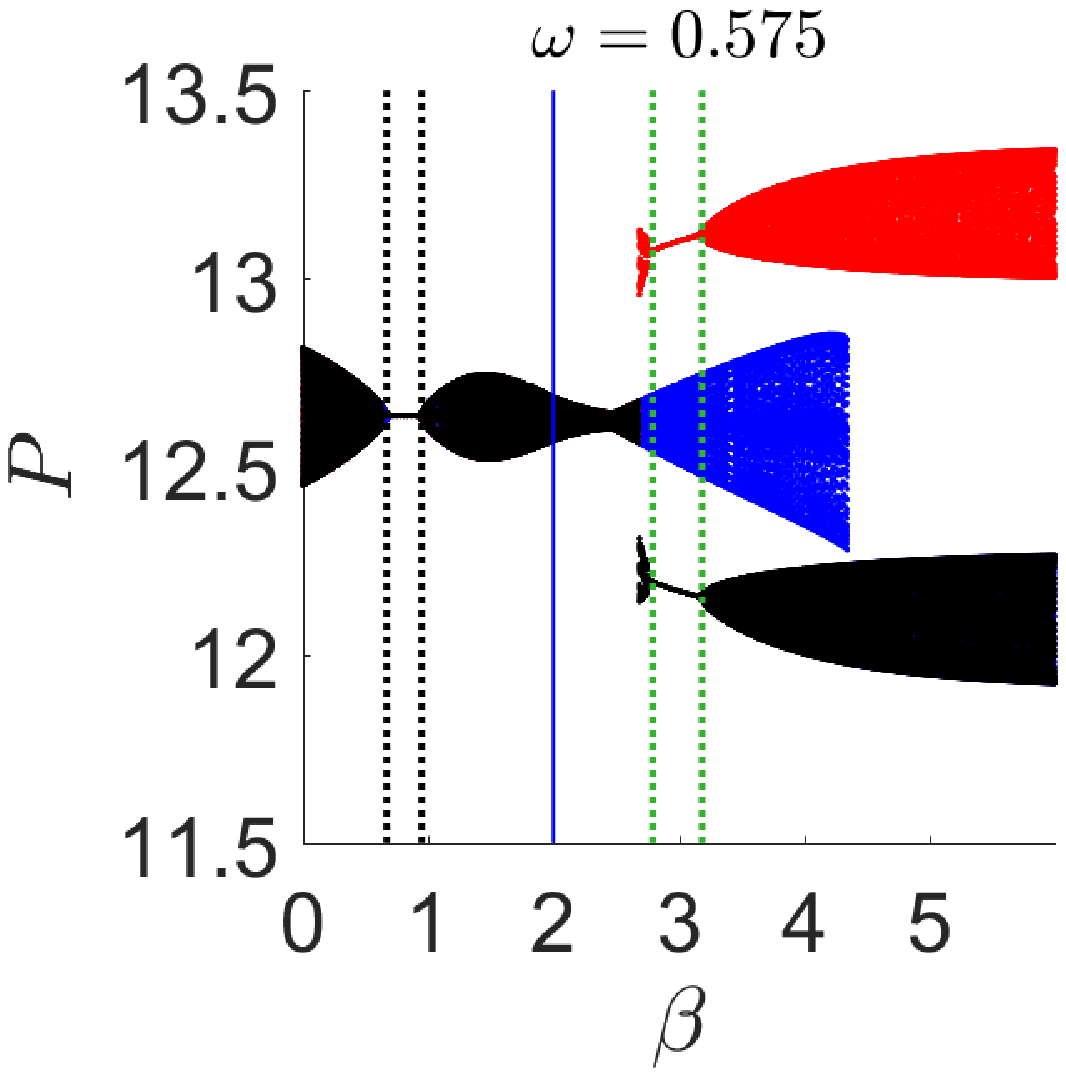}\\
    (B)
  \end{minipage}
\end{center}
\caption{\small Simulations obtained setting $\sigma=1.3$ and
  $\gamma=1.05.$ (A): Two-parameters bifurcation diagram. White color
  is used to identify parameters' combinations for which the initial
  datum converges toward a steady state, cyan is used for attractors
  consisting of more than 32 points. The black dotted line represents
  the stability threshold of $S^{\ast}$ while the green dotted line
  refers to the stability threshold of $S^H.$ (B): Bifurcation
  diagrams for $P$ on varying $\beta$ when $\omega=0.575.$ Different
  colors correspond to different initial conditions.} \label{fig:sim2}
\end{figure}
The blue bifurcation diagram, obtained setting
$(Y_0,P_0,Z_0)=(Y^{\ast}+0.001,P^{\ast}+0.001,Z^{\ast}+0.001)$ and increasing $\beta$ from $0$ to $6,$ lies below the red one,
obtained setting
$(Y_0,P_0,Z_0)=(Y^{\ast}+0.001,P^{\ast}+0.001,Z^{\ast}+0.001)$ and decreasing $\beta$ from $6$ to $0,$ which in turns lies above the
black one, obtained setting
$(Y_0,P_0,Z_0)=(Y^{\ast}-0.001,P^{\ast}-0.001,Z^{\ast}-0.001)$ and decreasing again $\beta$ from $6$ to $0$.\\
We numerically checked that the closed invariant curve,
arising from the stability loss of $S^{\ast}$,
attracts almost any trajectory also for intensity of choice values
$\beta>2$ suitably close to $2$, when steady states
$S^H$ and $S^L$ already exist but are
unstable (black bifurcation diagram, just to the right of the blue
dotted line). Such closed invariant curve coexists, as $\beta$
increases, first with another couple of closed invariant curves
(blue, red and black bifurcation diagrams just to the left of the
first green dotted line), then with $S^H$
and $S^L$ (blue, red and black bifurcation diagrams
between the two green dotted lines) and finally again with another
couple of closed invariant curves (blue, red and black bifurcation
diagrams just to the right of the second green dotted line).
\\Similar considerations hold also when the destabilizing scenario
for $S^{\ast}$ is followed by the stabilizing one
for $S^H$ and $S^L.$
Once more, the common aspect is the coexistence of attractors
   encompassing optimistic, pessimistic or unbiased levels in the
   economic variables, with the addition of the magenta attractor in
   which each piece embodies large or small levels too.
On varying $\omega,$ from vertical sections of Figure \ref{fig:sim2}
(A) we infer that we can have either stabilizing scenarios for both
$S^{\ast}$ (e.g. for $\beta=1$) and
$S^H$ (e.g. for $\beta=3$) or unconditionally unstable
scenarios for both $S^{\ast}$ and
$S^H,$ when $\beta$ is sufficiently close
to $2.$

\begin{figure}[t!]
\begin{center}
  \begin{minipage}{0.4\textwidth}
    \centering
    \includegraphics[width=\textwidth,trim=0.7cm 0cm 1.9cm 0cm, clip=true]{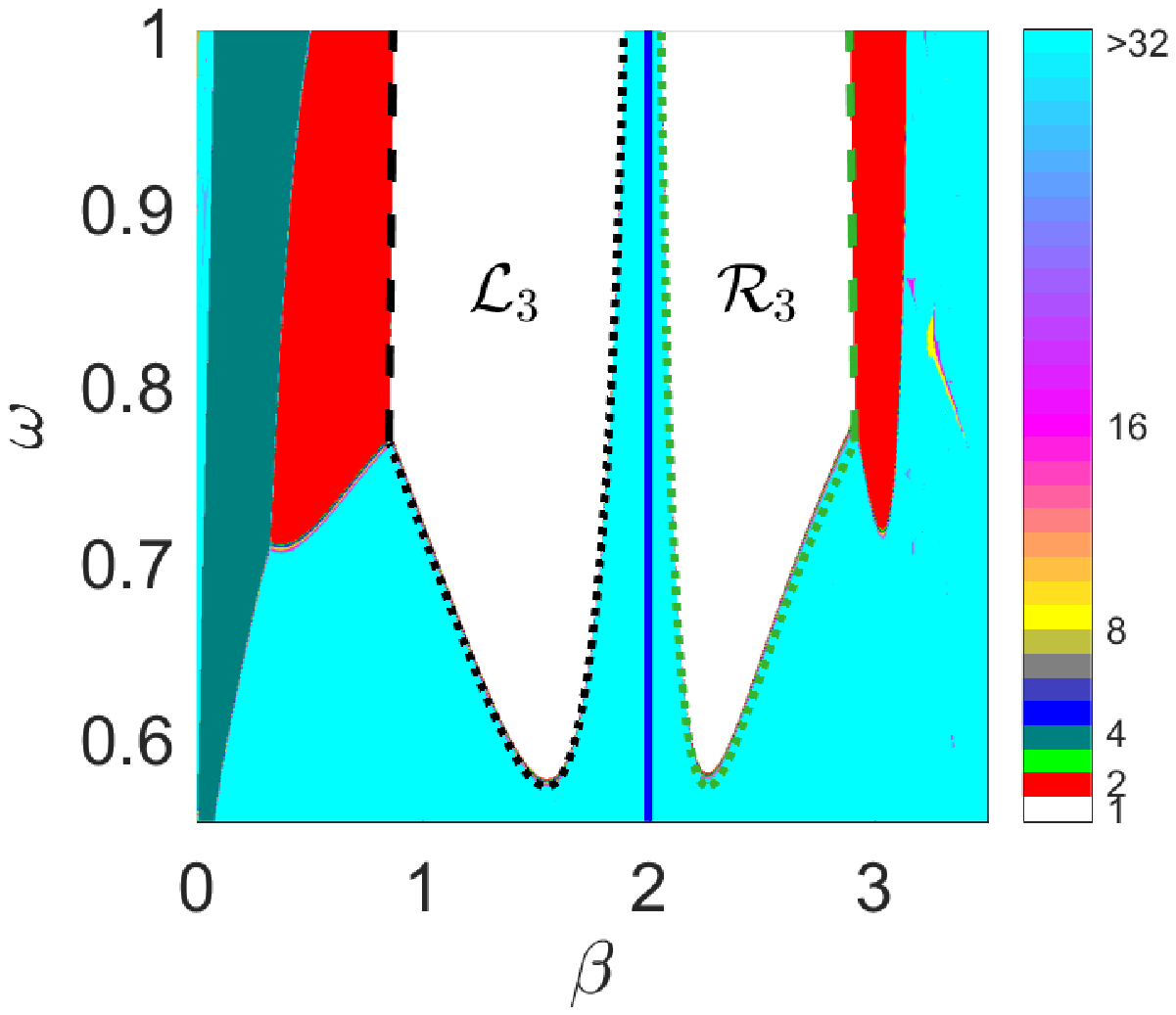}\\
    (A)
  \end{minipage}
  \begin{minipage}{0.4\textwidth}
    \centering
    \includegraphics[width=\textwidth,trim=0.8cm 0cm 2cm 0cm, clip=true]{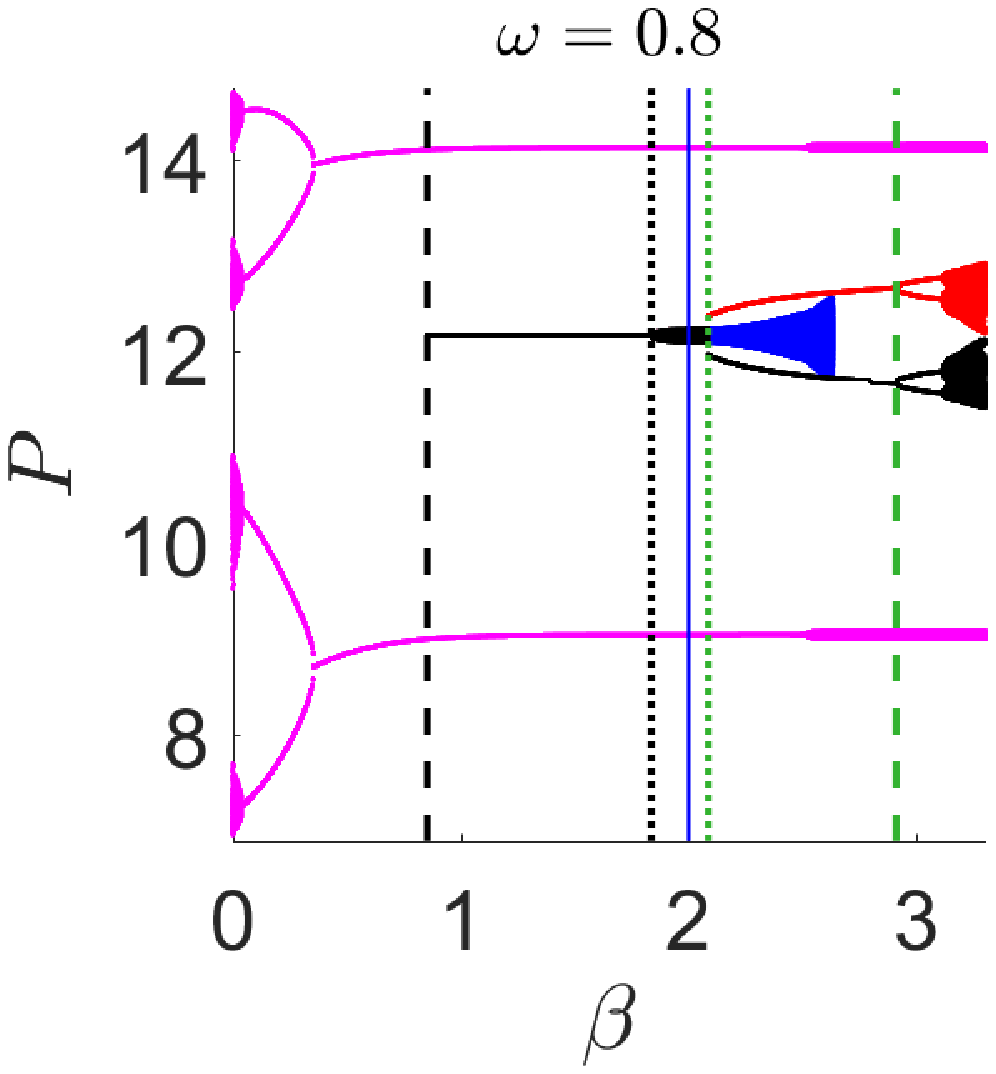}\\
    (B)
  \end{minipage}
  \begin{minipage}{0.4\textwidth}
    \centering
    \includegraphics[width=\textwidth,trim=0.8cm 0cm 2cm 0cm, clip=true]{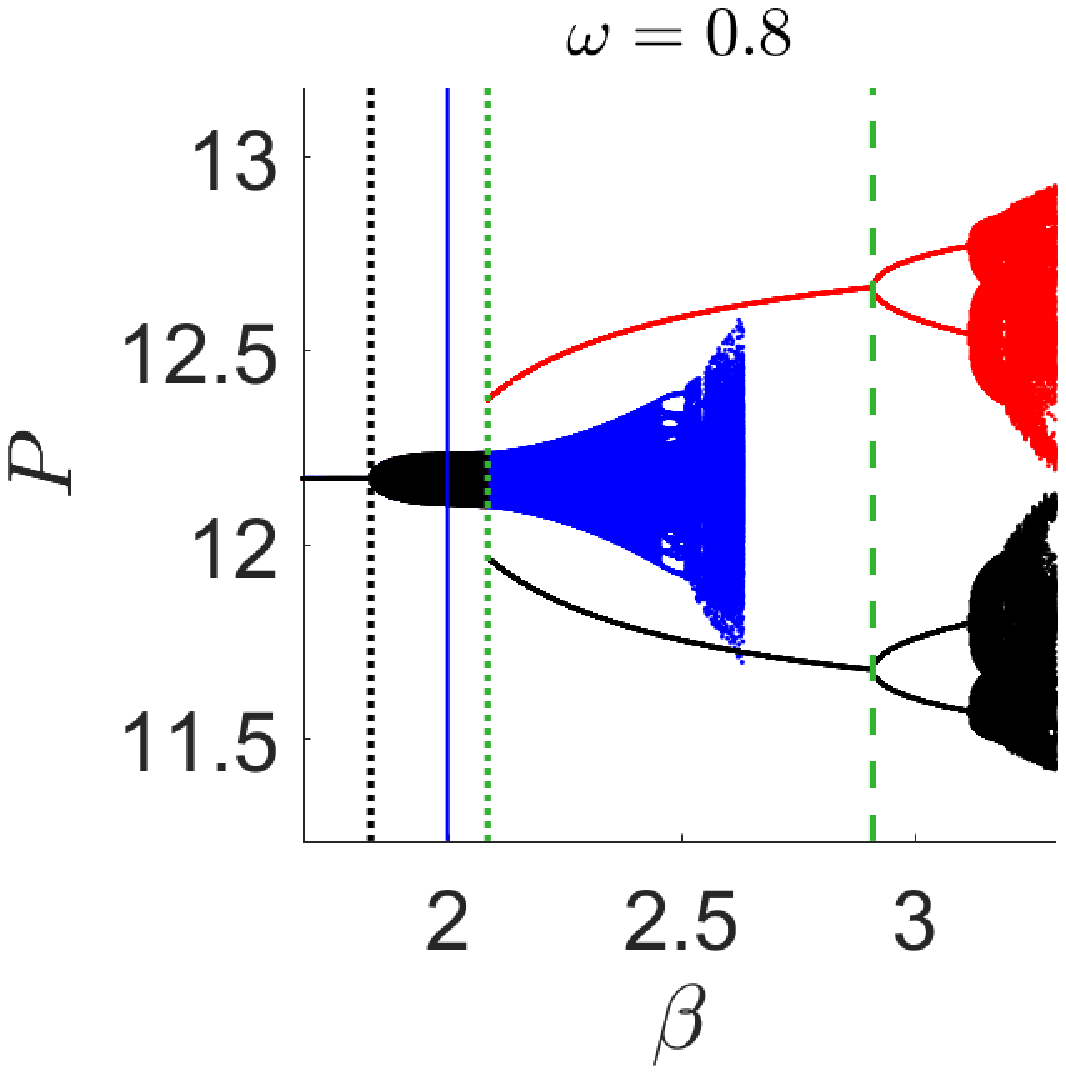}\\
    (B)
  \end{minipage}
\end{center}
\caption{\small Simulations obtained setting $\sigma=4$ and
  $\gamma=1.05.$ (A): Two-parameters bifurcation diagram. White color
  is used to identify parameters' combinations for which the initial
  datum converges toward a steady state, other colors are used for
  attractors consisting of more than a single point (cyan is used for
  attractors consisting of more than 32 points). Black lines represent
  stability thresholds of $S^{\ast},$ green lines stability thresholds
  of $S^H.$ (B): Bifurcation diagrams for $P$ on varying $\beta$ when
  $\omega=0.8.$ Different colors correspond to different initial
  conditions. (C): blow up of plot (B).} \label{fig:sim3}
\end{figure}

The last family of simulations, whose two-dimensional bifurcation
diagram is reported in Figure \ref{fig:sim3} (A), is obtained setting
$\sigma=4$ and $\gamma=1.05,$ and the corresponding stability region
is that reported in Figure \ref{fig:sr} (C). The initial datum is the
same as in the first two sets of simulations and, when dynamics
converge for $\beta>2,$ we have convergence toward $S^H.$ All the
previous considerations about the specular stability/instability
regions and stability thresholds for $S^{\ast}$ and $S^H$ with respect
to $\beta=2$ still hold. On varying $\beta$, we can have either
unconditionally unstable or mixed scenarios for both $S^{\ast}$ and
$S^H$ but, differently from the first two families of simulations, the
convergence toward the steady states may be replaced, when the
aforementioned steady states become unstable, by convergence toward
different kinds of attractors for different values of the interaction
degree parameter $\omega$. This is perceivable looking at the
stability thresholds for $S^{\ast}$ reported in Figure \ref{fig:sim3}
(A). The black dotted curve corresponds to condition \eqref{eq:sc3},
crossing which we enter the cyan region, in which attractors consist
of more then 32 points (actually, it is a closed invariant curve and
the stability loss of $S^{\ast}$ occurs through a Neimark-Sacker
bifurcation), while the black dashed line corresponds to condition
\eqref{eq:sc2}, crossing which we enter the red region, in which the
attractor is a period-2 cycle. The same holds for the green stability
thresholds for $S^H,$ too, which also in this case are obtained from
the black thresholds and numerically solving the implicit equations
used in the proof of Proposition \ref{th:stabhl}.

In the four bifurcation diagrams reported in Figures \ref{fig:sim3}
(B)-(C) we study the occurring dynamics when $\omega=0.8.$ The blue
bifurcation diagram is obtained setting
$(Y_0,P_0,Z_0)=(Y^{\ast}+0.001,P^{\ast}+0.001,Z^{\ast}+0.001)$ and
increasing $\beta$ from $1.5$ to $2.65;$ the red one is obtained
setting $(Y_0,P_0,Z_0)=(Y^{\ast}+0.001,P^{\ast}+0.001,Z^{\ast}+0.001)$
and decreasing $\beta$ from $3.3$ to $0.855;$ the black one is
obtained setting
$(Y_0,P_0,Z_0)=(Y^{\ast}-0.001,P^{\ast}-0.001,Z^{\ast}-0.001)$ and
decreasing $\beta$ from $3.3$ to $0.855;$ the magenta one is obtained
setting $(Y_0,P_0,Z_0)=(Y^{\ast}+0.001,P^{\ast}+0.001,Z^{\ast}+0.001)$
and increasing $\beta$ from $0$ to $3.3.$ As $S^{\ast}$ turns unstable
on decreasing $\beta$ (when the black dashed line is crossed), the
steady state is replaced by a 2-cycle, similarly to what happens in
the first family of simulations. However, as $S^{\ast}$ becomes
unstable on increasing $\beta$ (when there is a crossing of the black
dotted line), the steady state is replaced by a closed-invariant
curve. The attractor arising from the stability loss of $S^{\ast}$
then coexists with the stable steady states $S^L$ and $S^H.$ Further increasing $\beta,$ when such two steady states
  become unstable (green dashed line), a flip bifurcation
  occurs\footnote{Conversely, if $\beta$ is decreased, we
  numerically checked that when $S^L$ and $S^H$ become unstable a closed
  invariant curve emerges, which coexists with the blue attractor and
  quickly disappears.}. The subsequent
period-2 cycles lose stability through a Neimark-Sacker bifurcation.
\\The magenta attractor exists for the whole range of values of $\beta$
we consider. If we decrease $\beta$ from $0.854$ to $0,$ such
attractor firstly evolves toward a period-4 cycle and then a
Neimark-Sacker bifurcation occurs. If we increase $\beta$ starting
from $0.854,$ each of the two points composing the attractor is
replaced by a closed invariant curve. As $\beta$ increases, we have
that the attractor originated by the loss of stability of $S^{\ast}$
(blue bifurcation diagram) coexists with steady states $S^L$ and
$S^H,$ with a closed invariant curve, and with a period-2 cycle
(magenta bifurcation diagram), giving rise to a quite articulated
multistability situation made by qualitatively different attractors.
Once more, the common aspect is the coexistence of attractors
   encompassing optimistic, pessimistic or unbiased levels in the
   economic variables, with the addition of the magenta attractor in
   which each piece embodies large or small levels, too.

The previous considerations confirm and enrich the analytical results
about local stability. 
From a qualitative point of view, we can not speak about a unique type of mixed scenario, as those arising from the
simulations above are significantly different. The first relevant difference
can be highlighted in correspondence of the loss of stability of
$S^{\ast}$ for increasing values of $\beta$, namely when the
rightmost stability threshold is violated. In the example reported in
Figure \ref{fig:sim1}, the stable steady state
$S^{\ast}$ loses stability through a pitchfork
bifurcation, and hence the new attractor has the same complexity as
the previous one, even if it has a different interpretation from the
economic point of view. In the example reported in Figures
\ref{fig:sim2} and \ref{fig:sim3}, dynamics converging to
$S^{\ast}$ are replaced by quasi-periodic dynamics,
that is, by an attractor of different complexity, even if arisen from
$S^{\ast}.$ The second difference can be
highlighted in correspondence to the loss of stability of
$S^{\ast}$ as $\beta$ decreases, when the leftmost
stability threshold is violated. In this case, the unique steady state
is again $S^{\ast},$ but different dynamics are
possible, as the stability can be lost through a ``simple'' period-2 cycle
or through more complex quasi-periodic dynamics. We stress that, for
the same parameter values describing the stock and the real economy,
different dynamics can arise on changing the degree of interaction. 
However, even in
   the presence of a very articulated range of complex situations, all the
   dynamic scenarios can be still traced back to the static setting
   outlined in Proposition \ref{ss}, and hence easily read in terms of
   the effects of beliefs.

\subsection{Stochastic simulations}
The deterministic analysis has revealed that, starting from a situation in which the fundamental steady state is locally stable, an increase in the parameter $\beta$ generates the onset of endogenous fluctuations as well as the rise of multistability phenomena. The goal of this part is to show that a stochastic version of our model is able to generate varied and realistic dynamics, as observed in real financial markets, such as bubbles and crashes for stock prices and fat tails and excess volatility in the distributions of returns. In the presence of exogenous noise (on investors' demands) and, accordingly, of fundamental shocks to the price dynamics, periods of high volatility in the price course may alternate with periods in which prices do not depart too much from the fundamental value. Such behavior may arise when the parameter setting is located near the pitchfork bifurcation boundary and exogenous noise can occasionally spark long-lasting endogenous fluctuations around the new steady states.\\
Therefore the model is modified along the following lines. The demand placed by the two groups of agents is set as
\begin{equation}
  D_t^X=\mu(F_t^X-P_t)+\varepsilon_t^X,\, X\in\{O,P\},
\label{dem_shock}
\end{equation}
where $\varepsilon_t^X$ are sequences of independent, identically
distributed random normal variables, with zero mean and variance
$s_X^2$. These random disturbances may reflect errors in investors'
decision making processes, heterogeneity within the same group, or they may
capture the idea that it is quite difficult for investors to determine
the fundamentals (in fact fundamental values may change over time due
to real shocks), as already argued by \cite{keynes}.\\
Plugging the
perturbed demands into the price equation \eqref{mm} and acting as in
Section 2, we end up with
\begin{equation}
 P_{t}=P_{t-1}+\sigma g_P\left(\mu \left( \left( 1-\omega \right) F^{\ast }+\omega dY_{t-1}-P_{t-1}+b\left( \frac{2}{1+e^{-4b\beta \left( P_{t-1}-\left(  1-\omega\right) F^{\ast }-\omega dY_{t-1}\right) }}-1\right) \right)+\varepsilon_t\right),
\label{price_shock}
\end{equation}
where $\varepsilon_t$ is is a sequence of independent, identically
distributed random normal variables, with zero mean and variance
$s^2=(\varepsilon_t^O+\varepsilon_t^P)^2/(2\mu)^2$.\\
In analyzing the stochastic version of our model, we aim at showing
that several stylized facts of financial markets may be retrieved
that, in turn, can possibly affect the behavior of the real
sector. Thanks to their apparent explanatory power, models that have
been being developed with interacting agents and interconnected
sectors are increasingly used as tools for economic policy
recommendations (see e.g. \cite{dieci,franke,
  schmitt}).\\
In what follows, we consider the price returns, defined as
\[
R_{t}=\log (P_{t})-\log (P_{t-1}),
\]
focusing, at first, on the deviation from normality of their distribution. In all the simulations reported in this subsection we employ the same parameter setting adopted for Figure \ref{fig:sim1}.
\begin{figure}[t!]
  \centering
  \includegraphics[width=0.5\textwidth]{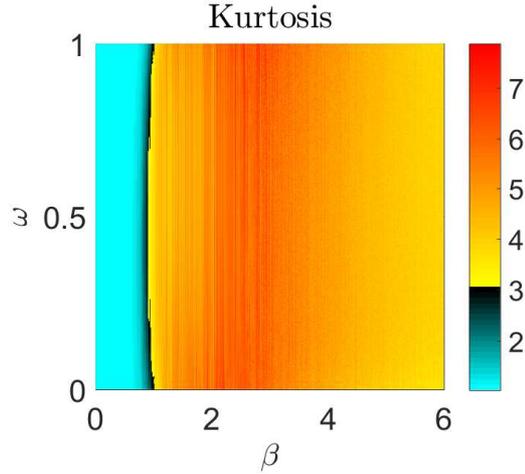}
  \caption{\small Kurtosis of returns distribution in the $(\beta, \omega)$-plane. Different colors are associated with different kurtosis values, according to the colors bar. The distributions of returns exhibits the fat-tail property (with more probability mass in the tails than warranted by a normal
distribution with identical mean and standard deviation) with the highest kurtosis levels obtained for values of the intensity of choice $\beta$ roughly between $2$ and $4$. 
}
\label{kurt}
\end{figure}
Figure \ref{kurt} displays the values of the kurtosis in the distributions of the log-returns as the parameters $\beta$ and $\omega$ jointly vary. We observe that, independently of the degree of interaction between the two markets, as the intensity of choice increases from 2 to 4, the kurtosis increases as well, reaching values that reveal a non-normal distribution of the log-returns. In fact, kurtosis measures how fat the tails of the distributions are compared to the ones of a normal distribution. In this respect, to observe fat tails it is essential to have trajectories in which prices move frequently far from their average and, in turn, this is related to the behavior of the intensity of choice $\beta$. For example, looking at Figure \ref{fig:sim1}, we observe that, as $\beta$ increases, coexisting attractors appear, due to the occurrence of the pitchfork bifurcation, and this reflects into a large kurtosis, e.g. for $2<\beta<4$ roughly. As $\beta$ increases further, the bifurcation diagram of Figure \ref{fig:sim1} shows that the dynamics turn into chaotic and this translates into slightly reduced values for the kurtosis in the log-returns of the stochastic version of the model. Such results are confirmed also by the top and bottom left panels of the Figure \ref{stoch_sim}, obtained for $\omega=0$ and $\omega=1,$ respectively. As the degree of market interaction grows, the switching between the optimistic and pessimistic strategies towards the best performing rule is able to increase the kurtosis in the distributions of the log-returns.\\ Another peculiar stylized fact of returns' times series is the volatility clustering, namely, the occurrence of several consecutive periods characterized by high volatility alternated with others characterized by low volatility. This is graphically evident from the central panels of Figure \ref{stoch_sim}, in which a typical example of time series of returns, obtained for the parameter setting used for Figure \ref{fig:sim1} and $s=0.15\,P^{\ast}/F$, is reported. Finally, volatility clustering is highlighted by the typical strongly positive, slowly decreasing autocorrelation coefficients of absolute returns, reported in the right top and bottom panels of Figure \ref{stoch_sim}.\\
Overall, we may come up with the conclusion that the stochastic
version of our model with interacting real and financial sectors is
able to replicate key empirical regularities of actual stock markets.
In particular, the functioning of the stochastic version of the model
descends from the functioning of its deterministic counterpart. In the
deterministic setup, endogenous dynamics arise when a model parameter
crosses the Neimark-Sacker bifurcation boundary, as well as coexisting
attractors appear when the pitchfork bifurcation takes place. In the
stochastic version, when the two non-fundamental steady states appear
and are still locally stable, the interplay of nonlinear elements and
random shocks leads to realistic dynamics.

\begin{figure}[t!]
  \begin{center}
    \begin{minipage}{0.32\textwidth}
      \centering
      \includegraphics[width=\textwidth,trim=1.5cm 0 1.5cm 0]{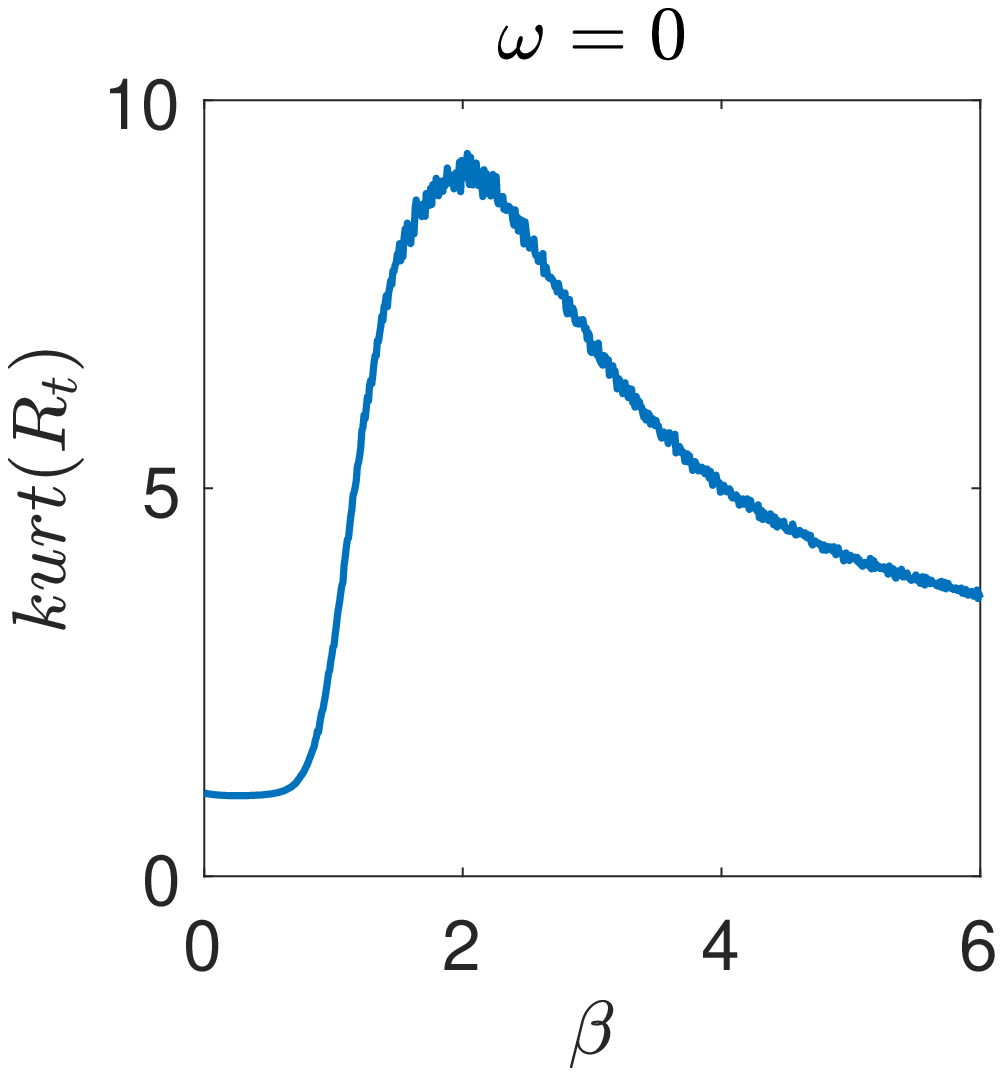}\\
      (A)
    \end{minipage}
    \begin{minipage}{0.32\textwidth}
      \centering
      \includegraphics[width=\textwidth,trim=1.5cm 0 1.5cm 0]{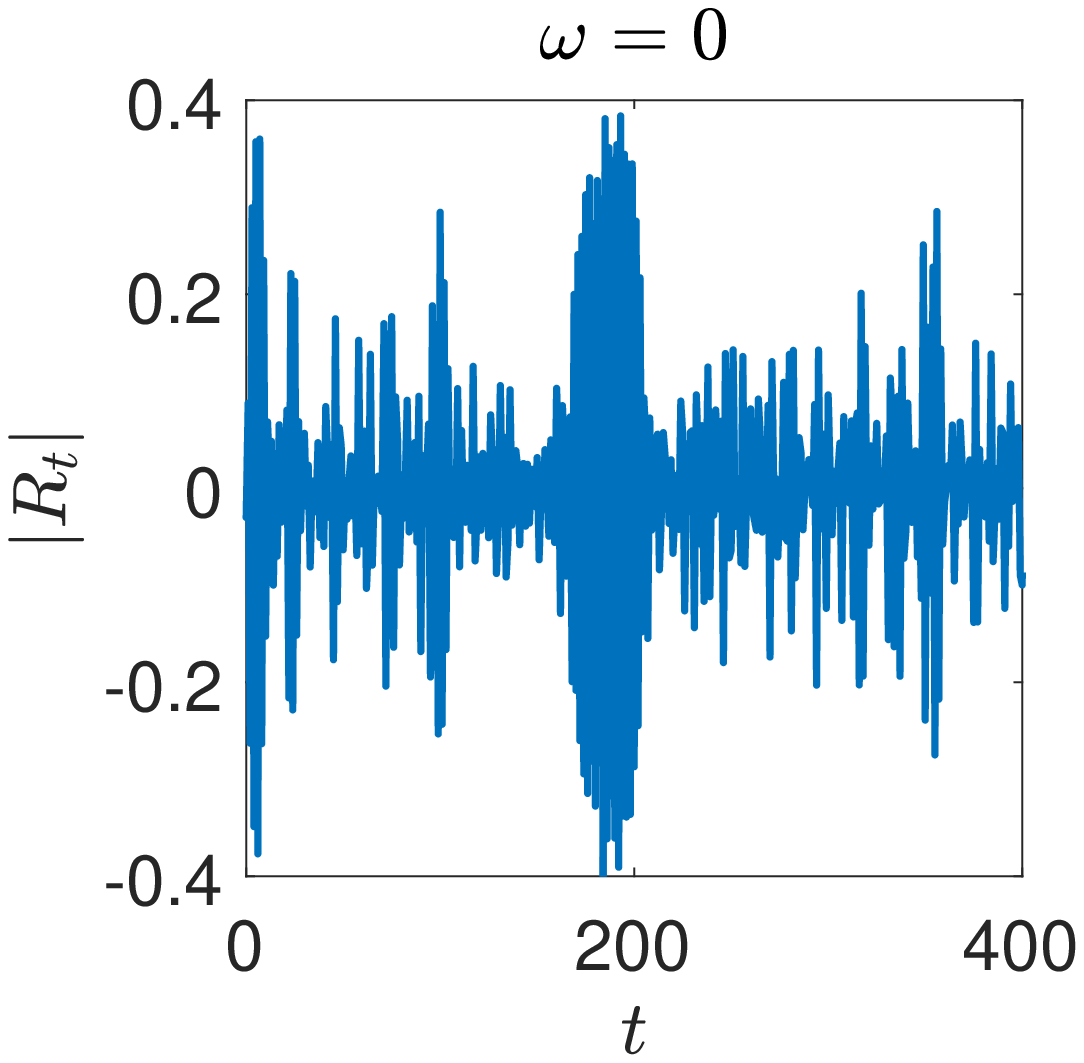}\\
      (B)
    \end{minipage}
    \begin{minipage}{0.32\textwidth}
      \centering
      \includegraphics[width=\textwidth,trim=1.5cm 0 1.5cm 0]{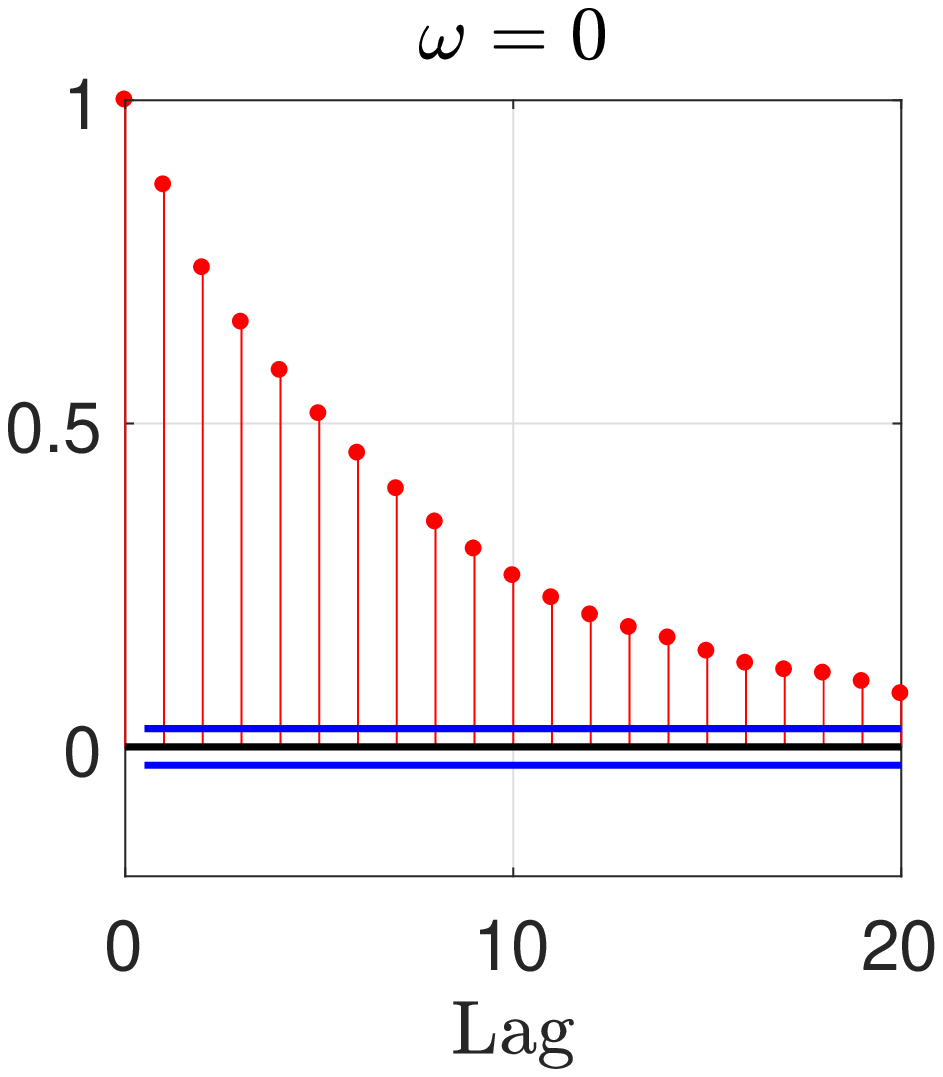}\\
      (C)
    \end{minipage}
    \begin{minipage}{0.32\textwidth}
      \centering
      \includegraphics[width=\textwidth,trim=1.5cm 0 1.5cm 0]{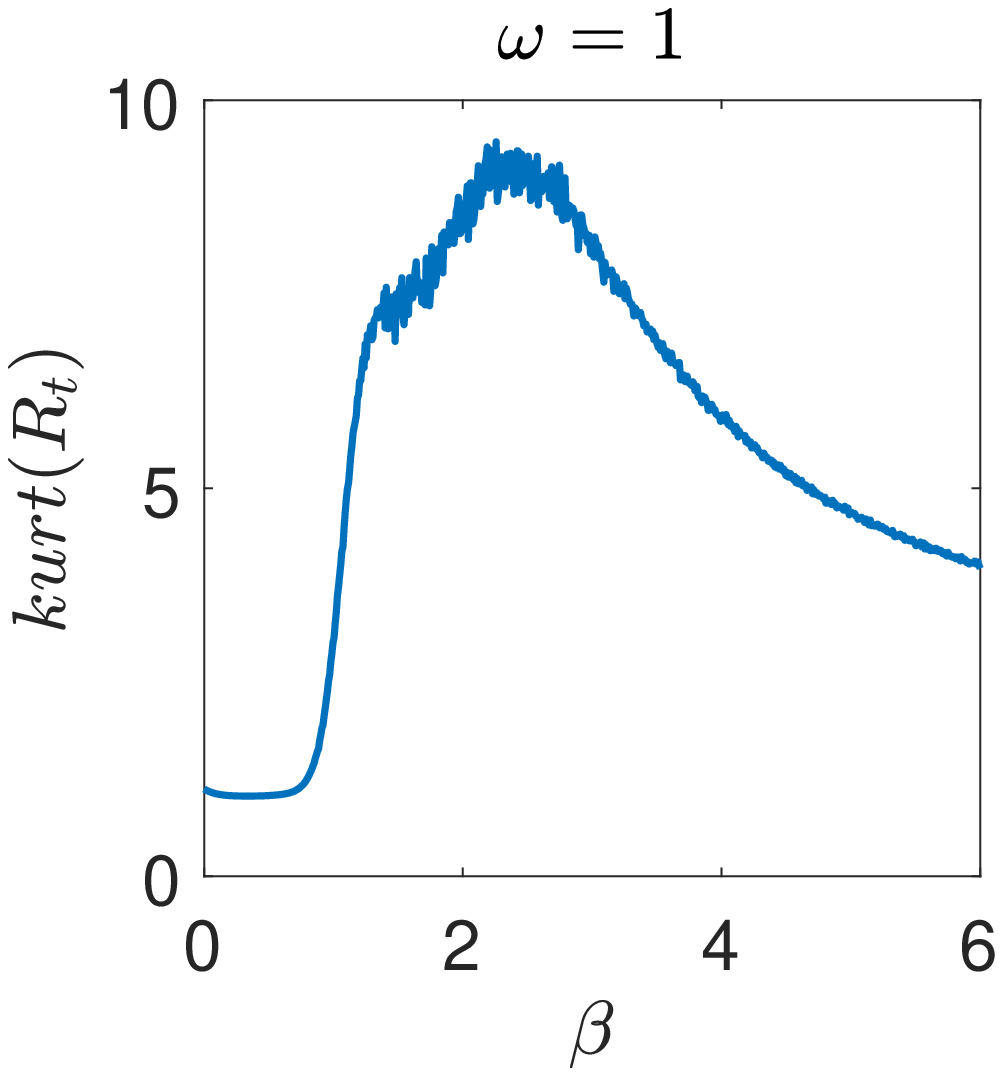}\\
      (D)
    \end{minipage}
    \begin{minipage}{0.32\textwidth}
      \centering
      \includegraphics[width=\textwidth,trim=1.5cm 0 1.5cm 0]{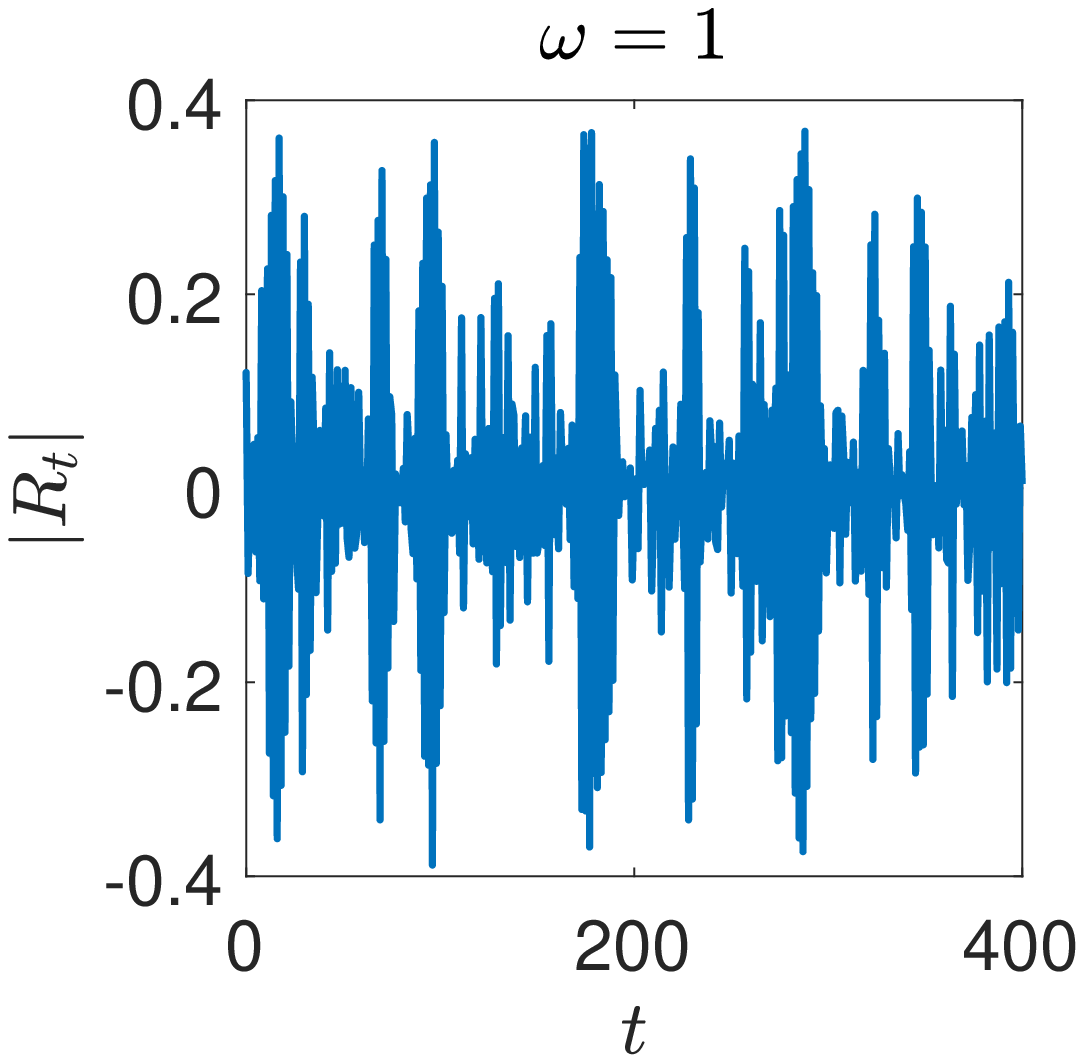}\\
      (E)
    \end{minipage}
    \begin{minipage}{0.32\textwidth}
      \centering
      \includegraphics[width=\textwidth,trim=1.5cm 0 1.5cm 0]{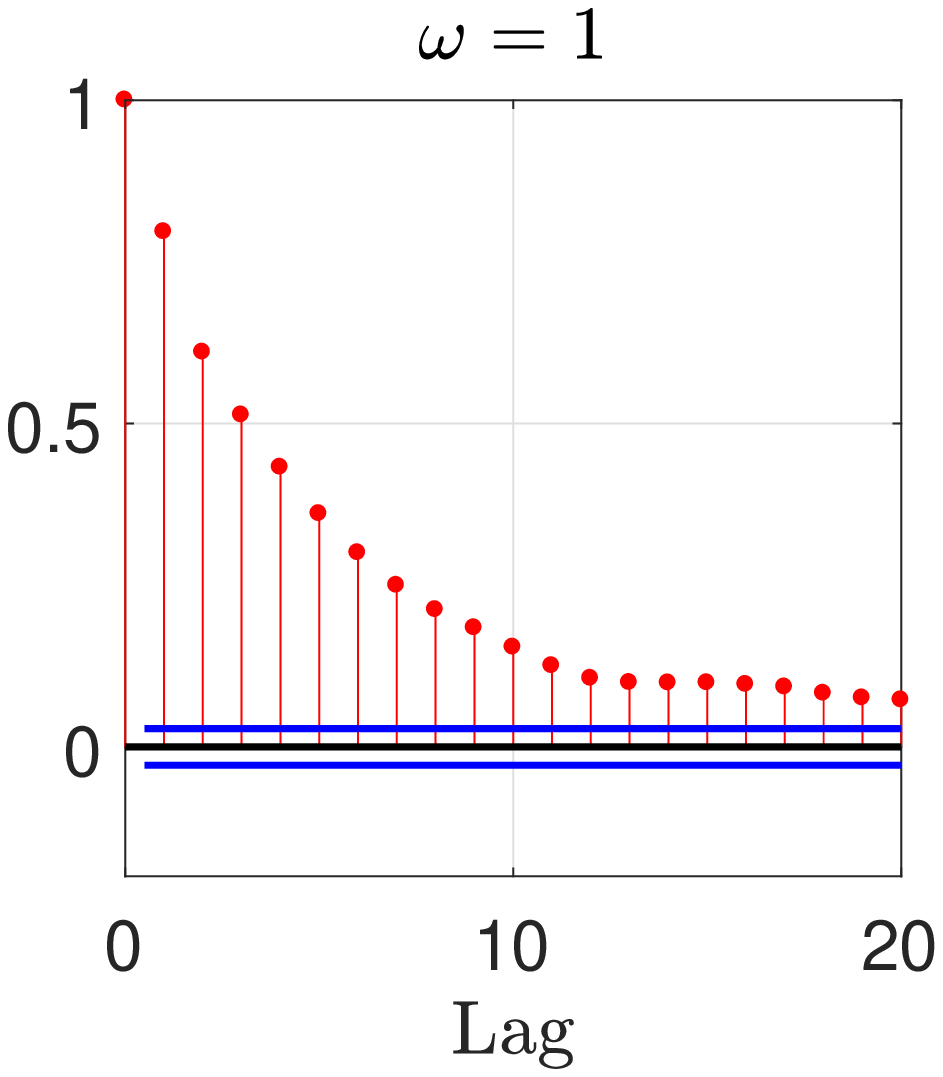}\\
      (F)
    \end{minipage}
  \end{center}
  \caption{\small The left column represents the evolution
of kurtosis when $\omega=0$ (A) and $\omega=1$ (D), the central column
plots the time series of absolute
returns exhibiting phases of high volatility alternating
with phases of low volatility either for $\omega=0$ (B) or
$\omega=1$ (E),
the right column shows the autocorrelation function for absolute
returns for $\omega=0$ (C) and $\omega=1$ (F) obtained for $\beta=6$, revealing
evidence for volatility clustering.
}
\label{stoch_sim}
\end{figure}

\section{Conclusions}
The proposed model shows that beliefs can influence the possible final
outcomes of an economy not only in terms of increasing the complexity
in the possible endogenous dynamics of the economic variables, but
also in a more constitutive way. From the static point of view,
strongly heterogeneous expectations can alter the set of steady
states, originating strongly polarized multiple steady states that in
turn reflect the optimism/pessimism of the beliefs. This can then
dynamically evolve into complex attractors still characterized by
optimism/pessimism of the beliefs, provided that such bias is strong
and relevant enough.  To sum up, it is certainly evident the role
played by the occurrence of the pitchfork bifurcation, with the
optimistic/pessimistic steady states common to the various sets of
simulations. But this is not the sole effect. From a broader
perspective, the occurrence of such biased steady states coexisting
with the unbiased one is the result of the repeated strategies'
evaluation made by the agents, which may lead the economy to converge
to high, intermediate or low, steady or endogenously fluctuating,
levels of national income (and prices as well), with the consequent
richness of dynamical behaviors related to all equilibria and their
eventual destabilization.\\ The investigation, performed with a mix of
analytical, numerical and empirical tools, has a twofold aim: besides
understanding the influence of the beliefs on the course of the
economy in terms of possible final outcomes and stability of the
desired equilibrium level of national income and prices, we are
interested in the role of market integration and in its possible
beneficial effect on the overall economy in terms of increasing level
of national income.  We showed that an increase of market interaction
may be either beneficial, as it may increase the national income of
the economy, or it may generate a contraction in its level, as no real
benefits to the society deriving from the activities based on an
increased integration may be appraised. The potential beneficial
effect of the degree of interaction seems to emerge also from the
local stability analysis, with the found stabilizing effect of the
interaction degree for suitable values of the intensity of choice
parameter, which measures agents' reactivity toward the best
performing strategy. This interplay between interaction, imitation and
beliefs is crucial also in view of understanding and interpreting the
stylized facts that our model is able to reproduce when buffeted with
stochastic noise, like the recurrent boom-bust phenomena observed in
the real financial markets together with several stylized facts
regarding stock price returns, such as positive autocorrelation,
volatility clustering and non-normal distribution characterized by a
high kurtosis and fat
tails.\\
The proposed model belongs to a class of models that, we believe, may
be useful in order to interpret and understand the recent developments
in the increasingly interconnected economy, which constitutes a true
challenge for regulatory measures which seek to appease abrupt market
fluctuations. We hope that our work will stimulate further
investigations in such direction, especially as concerns the deepening
of the understanding of the role of the agents' beliefs. The financial
crisis at the end of the last decade has not only made evident that
our comprehension of the dynamics of real and financial markets is
still incomplete, but it has also revealed how important is to
intensify the research activity in this field.

%
%
%

\newpage
\appendix
\section{Proofs of Propositions}

\begin{proof}[Proof of Proposition \ref{ss}]
  In order to find the steady states, we set $P_{t+1}=P_t=P,\,Y_{t+1}=Y_t=Y$
  and $Z_{t+1}=Z_t=Z$ in \eqref{map}. From \eqref{mapZ} we have $Z=Y.$
  Inserting it in \eqref{mapY} and \eqref{mapP} since, recalling
  \eqref{eq:spa}, we have that $g_I(0)=0$ and that $g_P(z)=0$ is
  solved by $z=0$ only, we obtain
  \begin{subequations}
    \begin{empheq}[left=\empheqlbrace]{align}
      &Y=A+cY+\omega hP, \label{eqY}\\
      &\left( 1-\omega \right) F^{\ast }+\omega dY-P+b\left(\frac{2}{%
          1+e^{-4b\beta \left( P-\left( 1-\omega \right) F^{\ast }-\omega dY\right) }}%
        -1\right) =0. \label{eqP}
    \end{empheq}%
    \label{eq}
  \end{subequations}

  The steady state existing independently of $\beta$ and $b$ is
  found setting $P=\left( 1-\omega \right) F^{\ast }+\omega dY,$ which
  immediately provides $(Y,P,Z)=(Y^{\ast}, P^{\ast}, Z^{\ast})=S^{\ast},$
whose components, by \eqref{a1}, are positive independently of
$\omega.$

In order to verify the existence of additional steady states, we
notice that, from \eqref{eqY}, national income $Y$ can be rewritten as
a function of the asset price $P$ as
\begin{equation}\label{yp}
Y=\dfrac{A+\omega hP}{1-c}.
\end{equation}
Inserting such expression in \eqref{eqP}, we find the following
equation in $P,$ whose solutions are the steady state values of
price:
\begin{equation}\label{im}
  \left( 1-\frac{%
      \omega ^{2}dh}{1-c}\right) P-\left( 1-\omega \right) F^{\ast }-\frac{%
    \omega dA}{1-c}+b-\frac{2b}{1+e^{-4b\beta \left( \left( 1-\frac{%
            \omega ^{2}dh}{1-c}\right) P-\left( 1-\omega \right) F^{\ast }-\frac{%
          \omega dA}{1-c}\right)} }=0.
\end{equation}
\noindent
Defining
\begin{equation}\label{f1}
  f_{1}(P)\equiv \left( 1-\frac{%
      \omega ^{2}dh}{1-c}\right) P-\left( 1-\omega \right) F^{\ast }-\frac{%
    \omega dA}{1-c}+b
\end{equation}
and
\begin{equation}\label{f2}
  f_{2}(P)\equiv \frac{2b}{1+e^{-4b\beta \left( \left( 1-\frac{%
            \omega ^{2}dh}{1-c}\right) P-\left( 1-\omega \right) F^{\ast }-\frac{%
          \omega dA}{1-c}\right)} }\,,
\end{equation}
we observe that $f_{1}$ is a straight line that crosses the horizontal
axis of the Cartesian coordinate plane at the point
$P_{\ell}=\frac{\left( 1-c\right) \left( 1-\omega \right) F^{\ast
  }+\omega dA-\left( 1-c\right) b}{1-c-\omega ^{2}dh}$
and the vertical axis at
$f_1(0)=b-\left( 1-\omega \right) F^{\ast }-\frac{\omega dA}{1-c}.$ We
stress that such points can be either positive or negative and
that the function $f_{1}$ is increasing in $P$. In fact, thanks to the
positivity condition for the steady state $P^{\ast },$ we have
$1-\frac{\omega ^{2}dh}{1-c} >0.$\\
Simple computations allow to conclude that $f_{2}$ is an increasing
function in $P$ and it has an inflection point at $P=P^{\ast
}$, 
it is convex for $
P\in \left( -\infty,\,P^{\ast }\right) $ and concave for $P\in \left(
  P^{\ast },\,+\infty \right).$ Moreover,
$f_{2}(P)$ crosses the vertical axis at the point $f_2(0)=\frac{%
  2b}{1+e^{4b\beta \left( \left( 1-\omega \right) F^{\ast
      }+\frac{\omega dA}{%
        1-c}\right)} }$.  We always find $f_{1}\left(P^{\ast }\right)
=f_{2}\left(P^{\ast }\right)
$ and, according to the values assumed by the model parameters,
$f_1$ and $f_2$ can also meet at some points $%
P^L$ and $P^H,$ with $P^L\in \left( P_{\ell},\,P^{\ast }\right)
$ and $P^H\in \left( P^{\ast },\,P_{h}\right),$ where
$P_{h}$
is the intersection point between the horizontal line $g(P)=2b$
and $f_{1}(P).$
Recalling the expression of $P^{\ast},$
it is easy to see that $P_{\ell}$
and $P_h$
respectively correspond to the left and right bounds for $P$
in \eqref{bounds}. In order to determine the bounds on $Y,$
  it is possible to proceed as done for $P,$
  obtaining an equation like \eqref{im} from \eqref{eq}, this
  time in terms of $Y.$\\
More precisely, since $f_{1}$
is an increasing straight line and $f_{2}$
is an increasing function in $P$,
with an inflection point at $P=P^{\ast
},$ being convex for $ P\in \left( -\infty,\,P^{\ast
  }\right)$ and concave for $P\in \left( P^{\ast },\,+\infty
\right),$ then if $f_1'(P^{\ast }) \ge f_2'(P^{\ast
})$ we have exactly one solution to \eqref{im}, that is, $P=P^{\ast
},$ while if $f_1'(P^{\ast })<f_2'(P^{\ast
})$ we have exactly three distinct solutions, $P^L,\,P^{\ast
}$ and $P^H,$ to \eqref{im}. We notice that $f_1'(P^{\ast })<f_2'(P^{\ast })$ corresponds to
\[
 1-\frac{%
\omega ^{2}dh}{1-c} < 2b^2\beta\left( 1-\frac{%
\omega ^{2}dh}{1-c}\right),
\]
which, recalling \eqref{a1}, reduces to \eqref{cond}.


To guarantee the economic meaningfulness of $P^H$
and $P^L,$
we require their positivity. A sufficient condition to ensure this
consists in setting $P_{\ell}>0,$
which is verified for each $\omega\in[0,1]$
when $b(1-c)<\min\{dA,F^{\ast
}(1-c)\}.$ Since $f_2(P)$ is strictly positive and
$f_1(P)$ is strictly increasing, their intersection
is realized for $P>0.$

Finally, it is easy to show that $f_1$
and $f_2$
are symmetric w.r.t. $P=P^{\ast
},$ i.e., that $f_i(P^{\ast }+\varepsilon)-f_i(P^{\ast })=f_i(P^{\ast
})-f_i(P^{\ast }-\varepsilon),$ for every
$\varepsilon>0$
and $i\in\{1,2\}.$
For $f_1$
that is true because it is a straight line, while for $f_2$
a direct computation shows that $f_2(P^{\ast
}+\varepsilon)-f_2(P^{\ast })=f_2(P^{\ast })-f_2(P^{\ast
}-\varepsilon)=b\bigl(e^{4b\beta\varepsilon ( 1-\frac{%
    \omega ^{2}dh}{1-c})}-1\bigr)/\bigl(e^{4b\beta\varepsilon (
  1-\frac{%
    \omega ^{2}dh}{1-c})}+1\bigr).$

We then have that, when \eqref{cond} holds, \eqref{im} is solved not
only by $P^{\ast },$ but also by $P^L$ and $P^H,$ symmetric values
w.r.t. $P^{\ast }.$ From \eqref{yp} and since in the steady states we
have $Z=Y,$ it follows that national income positively depends on the
asset price and thus to $P^L$ and $P^H$ correspond $Y^L$ and $Y^H,$
with $Y^L<Y^{\ast }<Y^H,$ as well as $Z^L$ and $Z^H,$ with
$Z^L<Z^{\ast }<Z^H$.
Moreover, if $P^H$
and $P^L$
are positive, then, from \eqref{yp} 
, $Y^H,Y^L,Z^H$
and $Z^L$
are positive, too.
Finally, by the linearity of \eqref{yp} and by
the symmetry of $P^L$ and $P^H$ w.r.t. $P^{\ast },$ it follows that
$Y^L$ and $Y^H$ are symmetric w.r.t. $Y^{\ast },$ and, since in the
steady states we have $Z=Y,$ then $Z^L$ and $Z^H$ are symmetric
w.r.t. $Z^{\ast }.$ This concludes the proof.
\end{proof}

\begin{proof}[Proof of Proposition \ref{scbeta}]
  We start noting that, from \eqref{yp}, $Y$ positively depends on
  $P$, so it is sufficient to study the monotonicity for $P$ on
  varying either $\beta$ or $b$. Setting
  \begin{equation}\label{defs}
    W=\left( 1-\frac{\omega ^{2}dh}{1-c}\right) P-\left( 1-\omega \right) F^{\ast }-\frac{\omega dA}{1-c},
  \end{equation}
  we may rewrite \eqref{im} as
  \begin{equation}\label{imp}
    W+b =\frac{2b}{1+e^{-4b\beta W}},
\end{equation}
from which it follows that
\begin{equation}\label{s}
e^{- 4 b \beta W} = \frac{b - W}{b + W}.
\end{equation}
Moreover, it holds that
\begin{equation}\label{s2}
  W = \frac{1 - c - \omega^2 dh}{1 - c}(P - P^{\ast}).
\end{equation}

For $W \in (- b, b)$ on both sides of equation \eqref{s} there are positive,
strictly decreasing functions of $W$ and the rhs of \eqref{s}
is also independent of $\beta.$ In what follows, we make
reference to some explanatory plots reported in Figure
\ref{fig:scproof}.

Firstly we consider $P^H (\beta).$ From (\ref{bounds}) we have that
$W^H (\beta),$ defined by (\ref{s2}), belongs to $(0, b).$ Setting
$W \in (0, b),$ the lhs of (\ref{s}) is decreasing with respect to
$\beta$, so that if $\beta_1 < \beta_2$ we find
$e^{- 4 b \beta_1 W} > e^{- 4 b \beta_2 W}.$ An illustrative example
is reported in Figure \ref{fig:scproof} (A), from which it is evident
that the last inequality guarantees that the unique solution
$W^H (\beta)$ of \eqref{s} on $(0, b)$ increases with $\beta$, and
hence, recalling (\ref{s2}), $P^H (\beta)$ increases, too. We proceed
in a similar way for $P^L$. From (\ref{bounds}), we have that
$W^L (\beta),$ defined by (\ref{s2}), belongs to $(- b, 0) .$ Setting
$W \in (- b, 0)$, the lhs of \ (\ref{s}) is increasing with respect to
$\beta,$ so that if $\beta_1 < \beta_2$ we find
$e^{- 4 b \beta_1 W} < e^{- 4 b \beta_2 W}$. An illustrative example
is reported in Figure \ref{fig:scproof} (B), from which it is evident
that the last inequality guarantees that the unique solution $W^L (\beta)$
of \eqref{s} on $(- b, 0)$ decreases with $\beta$, and hence,
recalling (\ref{s2}), $P^L (\beta)$ decreases, too. The two limits of
$P^j,\,j\in\{H,L\},$ as $\beta\rightarrow+\infty$ can be easily
obtained by noting that the lhs of \eqref{s} pointwise approaches $0$
for each $W\in(0,b),$ so that $W^H(\beta)$ tends to $b,$ and the lhs of \eqref{s}
pointwise approaches $+\infty$ for each $W\in(-b,0),$ so that $W^L(\beta)$ tends to $-b.$
Using \eqref{s2} allows concluding. For the limits of $Y^j,\,j\in\{H,L\},$
it is sufficient to compute the limit as $\beta\rightarrow+\infty$ of
the rhs of \eqref{yp}, in which we replace $P$ with either $P^H$ or
$P^L$ and we use the just computed limits.

To study the monotonicity of $P^L (b)$ and $P^H (b)$ with respect to $b$, we
rewrite \eqref{s} as
\begin{equation}
  e^{- 4 \beta W} = \left( \frac{b - W}{b + W} \right)^{1 / b}, \label{s3}
\end{equation}
which is well defined on $W \in (- b, b)$ and shares the solutions with
(\ref{s2}). The lhs is represented by a positive, strictly decreasing function
of $W.$ Let us consider $b_1 < b_2$ and the corresponding steady states $P^H
(b_1)$ and $P^H (b_2)$. From (\ref{bounds}) we have that $W^H (b_1)$ and $W^H
(b_2),$ defined by (\ref{s2}), respectively belong to $(0, b_1)$ and $(0, b_2)
.$ For $W \in (0, b_1)$, noting that $\frac{b_1 - W}{b_1 + W} < \frac{b_2 -
W}{b_2 + W}$, we find $\left( \frac{b_1 - W}{b_1 + W} \right)^{1 / b_1} < \left( \frac{b_2
  - W}{b_2 + W} \right)^{1 / b_2} .$
An illustrative example is reported in Figure \ref{fig:scproof} (C),
from which it is evident that the last inequality allows concluding
that if $W^H (b_2) \in (0, b_1) $ it holds that
$W^H (b_1) < W^H (b_2)$ (indeed, also if
$W^H (b_2) \in [b_1, b_2),$ since we have
$W^H (b_1) \in (0, b_1)$). This implies that $W^H (b)$ increases with
$b$, and hence, recalling (\ref{s2}), $P^H (b)$ increases, too.

Now we turn our attention to $P^L (b_1)$ and $P^L (b_2),$ still
assuming $b_1 < b_2$. From (\ref{bounds}) we have that $W^L (b_1)$ and
$W^L (b_2),$ defined by (\ref{s2}), respectively belong to
$(- b_1, 0)$ and $(- b_2, 0) .$ For $W \in (- b_1, 0)$, noting that
$\frac{b_1 - W}{b_1 + W} > \frac{b_2 - W}{b_2 + W}$, we find
$\left( \frac{b_1 - W}{b_1 + W} \right)^{1 / b_1} > \left( \frac{b_2 -
    W}{b_2 + W} \right)^{1 / b_2}.$
An illustrative example is reported in Figure \ref{fig:scproof} (D),
from which it is evident that the last inequality allows concluding
that if $W^L (b_2) \in (- b_1, 0)$ it holds that $W^L (b_1) > W^L (b_2)$
(indeed, also if $W^L (b_2) \in (- b_2, - b_1]$, since we have
$W^L (b_1) \in (- b_1, 0)$). This implies that $W^L (b)$ decreases
with $b$, and hence, recalling (\ref{s2}), $P^L (b)$ decreases, too.
\begin{figure}[t!]
\begin{center}
  \begin{minipage}{0.24\textwidth}
    \centering
    \includegraphics[width=\textwidth,trim=1.7cm 0cm 2.3cm 0cm, clip=true]{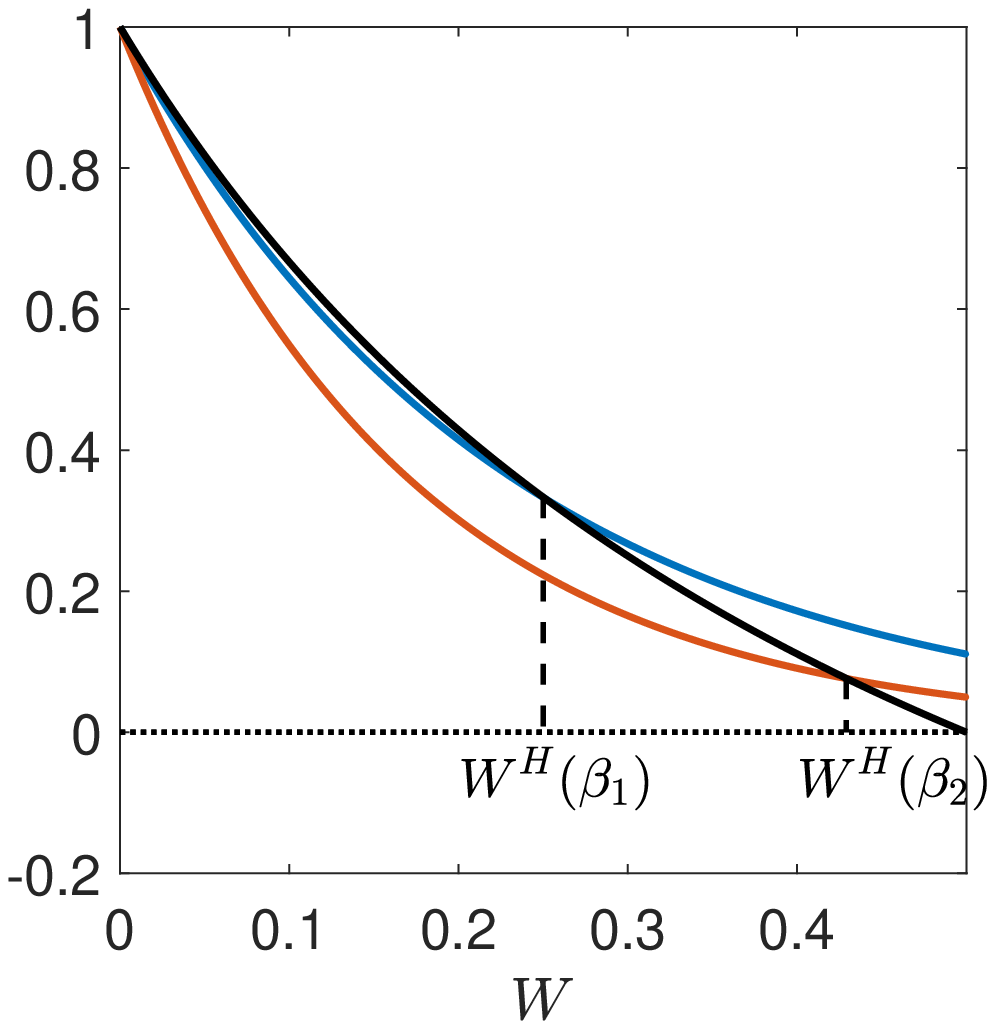}\\
    (A)
  \end{minipage}
  \begin{minipage}{0.24\textwidth}
    \centering
    \includegraphics[width=\textwidth,trim=1.7cm 0cm 2.3cm 0cm, clip=true]{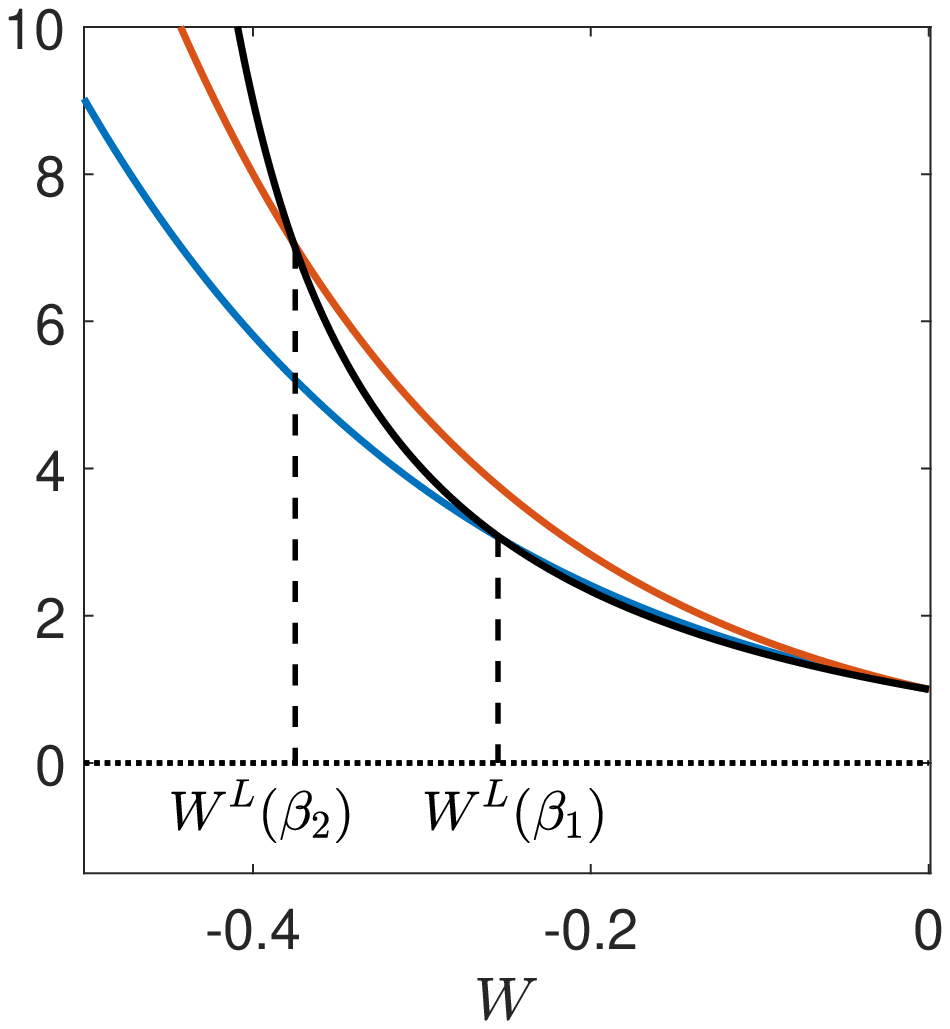}
    (B)
  \end{minipage}
  \begin{minipage}{0.24\textwidth}
    \centering
    \includegraphics[width=\textwidth,trim=1.7cm 0cm 2.3cm 0cm,
    clip=true]{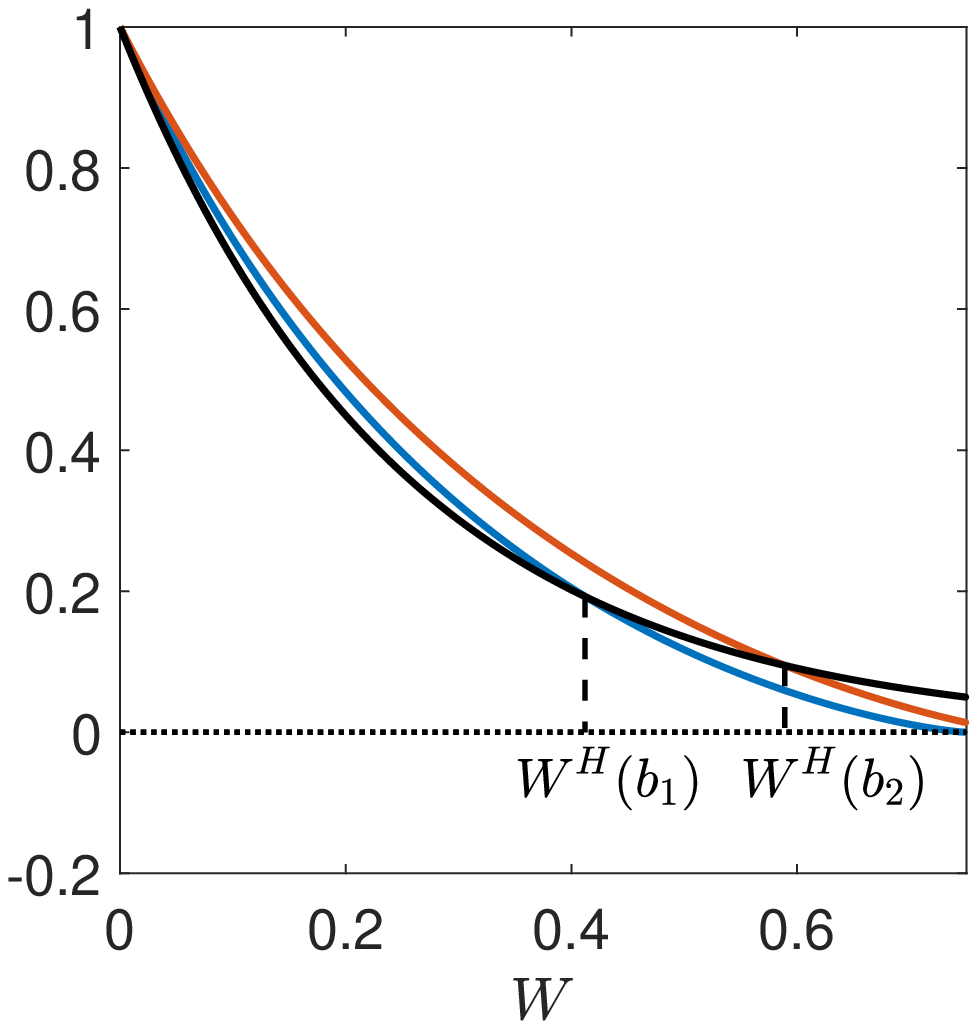}
    (C)
  \end{minipage}
    \begin{minipage}{0.24\textwidth}
    \centering
    \includegraphics[width=\textwidth,trim=1.7cm 0cm 2.3cm 0cm,
    clip=true]{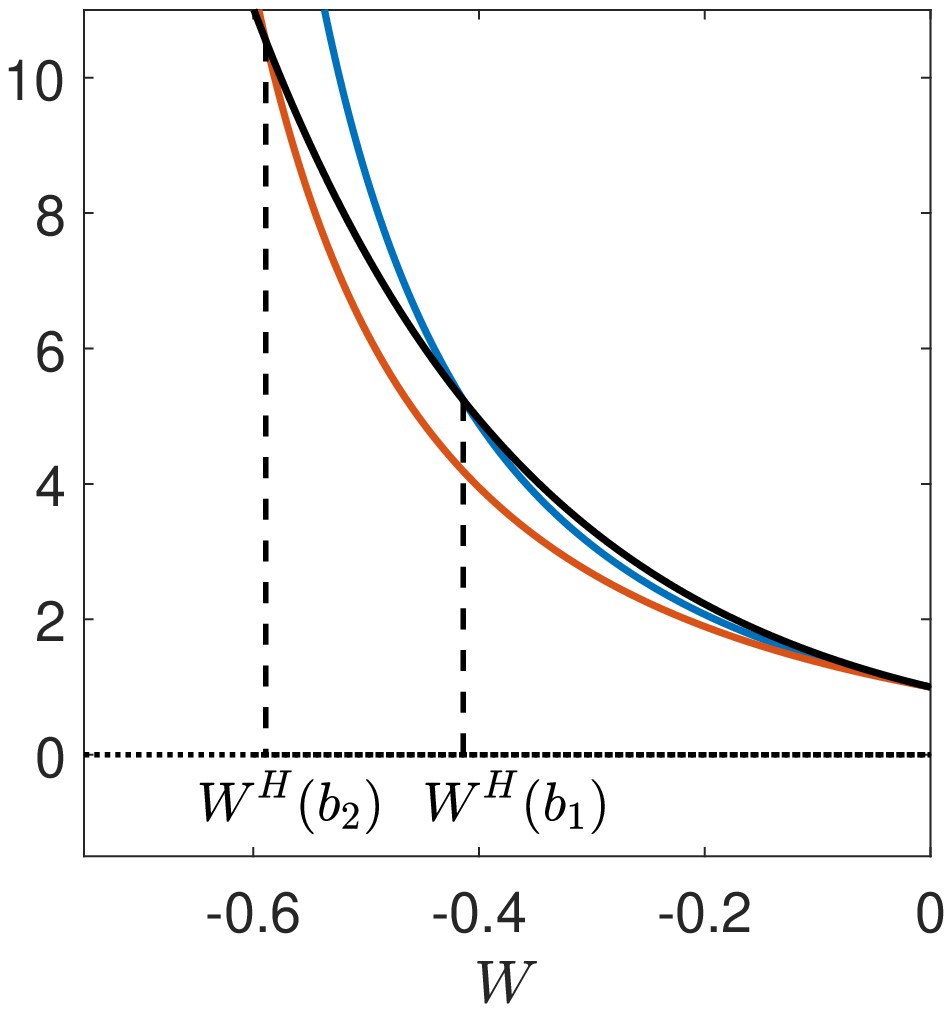}
    (D)
  \end{minipage}
\end{center}
\caption{(A), (B): Intersections between the function representing the
  rhs of \eqref{s} (black) and two functions representing the lhs of
  \eqref{s} respectively for $\beta=\beta_1$ (blue) and
  $\beta=\beta_2>\beta_1$ (red), when $W\in(0,b)$ (plot (A)) and
  $W\in(-b,0)$ (plot (B)). (C),(D): Intersections between the function representing the
  lhs of \eqref{s3} (black) and two functions representing the rhs of
  \eqref{s3} respectively for $b=b_1$ (blue) and
  $b=b_2>b_1$ (red), when $W\in(0,b_1)$ (plot (C)) and
  $W\in(-b_1,0)$ (plot (D)).} \label{fig:scproof}
\end{figure}
\end{proof}

\begin{proof}[Proof of Proposition \ref{prop3}]
Let us start focusing on the effects of $\omega$ on $P^{\ast}.$
A direct computation shows that
$$\frac{\partial P^{\ast}}{\partial\omega}=\frac{dh(dA-F^{\ast}(1-c))\omega^2+2dh F^{\ast}(1-c)\omega+(1-c)(dA-F^{\ast}(1-c))}{(1-c-\omega^2 dh)^2}.$$
Such derivative is non-negative for each $\omega\in[0,1]$ when
$dA-F^{\ast}(1-c)\ge 0$ and in this case $P^{\ast}$ increases with
$\omega,$ for every $\omega.$ When instead $dA-F^{\ast}(1-c)<0,$ the
numerator of ${\partial P^{\ast}}/{\partial\omega}$ may become
negative. More precisely, ${\partial P^{\ast}}/{\partial\omega}>0$ is
equivalent to
$dh(F^{\ast}(1-c)-dA)\omega^2-2dh
F^{\ast}(1-c)\omega+(1-c)(F^{\ast}(1-c)-dA)<0.$ When the discriminant
of
$dh(F^{\ast}(1-c)-dA)\omega^2-2dh
F^{\ast}(1-c)\omega+(1-c)(F^{\ast}(1-c)-dA)=0$ is non-positive, we
have that $P^{\ast}$ decreases with $\omega,$ for each
$\omega\in[0,1].$ When instead the discriminant is positive, it is
easy to show that there are two positive solutions $\omega_{1,2}$ to
such equation, and the larger between them always exceeds $1.$ Hence,
if the smallest solution $\omega_1$
is larger than 1, then $P^{\ast}$ decreases with $\omega,$ for each $\omega\in[0,1].$ If instead $\omega_1$ lies in $(0,1),$ then $P^{\ast}$ decreases with $\omega$ on $[0,\omega_1)$ and increases with $\omega$ on $(\omega_1,1].$ Setting ${\omega}_{P^{\ast}}=\omega_1,$ we obtain the desired result on $P^{\ast}.$\\
As concerns $Y^{\ast},$ a direct computation shows that
$$\frac{\partial Y^{\ast}}{\partial\omega}=\frac{h(\omega^2 dh F^{\ast}+2\omega(dA-F^{\ast}(1-c))+F^{\ast}(1-c))}{(1-c-\omega^2 dh)^2}.$$
Such derivative is positive for each $\omega\in[0,1]$ when
$dA-F^{\ast}(1-c)\ge 0$ and in this case $Y^{\ast}$ increases with
$\omega,$ for every $\omega.$ When instead $dA-F^{\ast}(1-c)<0,$ the
numerator of ${\partial Y^{\ast}}/{\partial\omega}$ may become
negative. More precisely, when the discriminant of
$\omega^2 dh F^{\ast}+2\omega(dA-F^{\ast}(1-c))+F^{\ast}(1-c)=0$ is
non-positive, we have that $Y^{\ast}$ increases with $\omega,$ for
each $\omega\in[0,1].$ When instead the discriminant is positive, it
is easy to show that there are two positive solutions $\omega_{3,4}$
to such equation, and the largest between them always
exceeds $1.$ Hence, if the smallest solution
$\omega_3$ is larger than 1, then $Y^{\ast}$ increases with $\omega,$
for each $\omega\in[0,1].$ If instead $\omega_3$ lies in $(0,1),$ then
$Y^{\ast}$ increases with $\omega$ on $[0,\omega_3)$ and decreases
with $\omega$ on $(\omega_3,1].$ Setting
${\omega}_{Y^{\ast}}=\omega_3,$ we obtain the desired result on
$Y^{\ast}.$

Let us now assume ${1}/{\sqrt{2\beta}}<b<\min\{dA/(1-c),F^{\ast }\},$
namely \eqref{cond} is violated and \eqref{a2} holds true, so that the two
additional steady states $0<P^L(\omega)<P^H(\omega)$ exist
for every $\omega\in[0,1].$ In order to analyze their behavior with
respect to increasing values of $\omega,$ we use the implicit function
theorem. Let
$f : (0,+\infty)\times(0, 1) \rightarrow
\mathbb{R},\,(P,\omega)\mapsto f(P,\omega),$ be
defined by the lhs of \eqref{im}
and let $P: (0, 1) \rightarrow (0,+\infty),\,\omega\mapsto P(\omega),$
be the function which associates to $\omega\in (0, 1)$ one of the
solutions implicitly defined by $f (P, \omega) = 0$ for
$b > 1 / \sqrt{2 \beta}$. Recalling the definition of $W$ in
\eqref{defs}, we then have
\[ \frac{\partial f}{\partial P} = \frac{(1 - c - dh \omega^2) (1 + 2
    (1 - 4 b^2 \beta) e^{- 4 b \beta W} + e^{- 8 b \beta W})}{(1 - c)
    (e^{- 4 b \beta W} + 1)^2} \] and
\[ \frac{\partial f}{\partial \omega} = \frac{(- Ad + F - Fc - 2 Pdh
    \omega) (1 + 2 (1 - 4 b^2 \beta) e^{- 4 b \beta W} + e^{- 8 b
      \beta W})}{(1 - c) (e^{- 4 b \beta W} + 1)^2}. \]

We can apply the implicit function theorem to investigate how the
steady state values for $P$ depend on $\omega$ provided that
\begin{equation}\label{neq}
  1 + 2 (1 - 4 b^2 \beta) e^{- 4 b \beta W} + e^{- 8 b \beta W} \neq 0.
\end{equation}
Setting $Z=e^{- 4 b \beta W} = (b - W)/(b + W),$ condition \eqref{neq}
can be rewritten as $1 + 2 (1 - 4 b^2 \beta) Z + Z^2 \neq 0,$ or
equivalently
\[ \frac{4 b^2 (2 \beta W^2 - 2 \beta b^2 + 1)}{(W + b)^2} \neq 0, \]
which holds true if and only if
$2 \beta W^2 - 2 \beta b^2 + 1 \neq 0.$ Recalling \eqref{s2},
the last inequality can be rewritten as
\begin{equation}\label{ne}
  P \neq P^{\ast} \pm \frac{1-c}{1 - c - {\omega}^2 d h} \sqrt{b^2 - \frac{1}{2
      \beta}}.
\end{equation}
By the symmetry of $P^{L}$ and $P^H$ with respect to
$P^{\ast},$ \eqref{ne} is satisfied if \eqref{ine} holds true
for each $\omega\in (0,1).$ 
Under such assumption we have
that, since $P (\omega)$ is well-defined for
$\omega \in (0, 1),$ then, by the implicit function theorem,
it is continuously differentiable for each $\omega \in (0, 1)$
and it holds that
\begin{equation}\label{po}
  P^{\prime} (\omega) = - \frac{\frac{\partial f}{\partial \omega}}{\frac{\partial
      f}{\partial P}} = \frac{Ad - F (1 - c) + 2 P (\omega) dh \omega}{1 - c - dh
    \omega^2}.
\end{equation}
From now on, to fix ideas, we focus on $P^L (\omega):$ analogous
arguments hold for $P^H
(\omega),$ as well. \\
If $Ad- F (1 - c) > 0$, then $\left(P^L\right)^{\prime}
(\omega) > 0$ on $(0, 1)$ and $P_L (\omega)$ is strictly increasing
for $\omega\in (0,1)$.\\
Conversely, if $Ad- F (1 - c) < 0$, then $\left(P^L\right)^{\prime}(0) < 0:$
if $\left(P^L\right)^{\prime} (\omega) < 0$ for each $\omega \in (0, 1)$, then
$P_L (\omega)$ is strictly decreasing on $(0, 1)$, otherwise there
exists ${\omega}_{P^L} \in (0, 1)$ such that
\begin{equation}\label{pl}
P^L ({\omega}_{P^L}) = \frac{F (1 - c) - Ad}{2 dh {\omega}_{P^L}}.
\end{equation}
If such ${\omega}_{P^L} \in (0, 1)$ exists, $P^L$ is decreasing on $(0,
{\omega}_{P^L})$ and $\left(P^L\right)^{\prime} ({\omega}_{P^L}) = 0$. Since
\[ \left(P^L\right)^{\prime\prime} (\omega) = \frac{2 dh \omega (Ad -F (1 - c) + 2 dh \omega P_L
   (\omega))}{(1 - c - dh \omega^2)^2} + \frac{2 dh \omega P_L^{\prime} (\omega) + 2
   dh P_L (\omega)}{1 - c - dh \omega^2}, \]
we have
\[ \left(P^L\right)^{\prime\prime} ({\omega}_{P^L}) = \frac{2 dh
  P^L({\omega}_{P^L})}{1 - c - dh {\omega}_{P^L}^2} > 0 \]
and thus ${\omega}_{P^L}$ is a minimum point. As
$P^{\prime}(\omega)$ positively depends on $P(\omega)\omega,$ if
$\omega$ increases and $P(\omega)$ is increasing, by the continuity of
$P^{\prime}(\omega)$ there cannot exist
$\widehat{\omega}_L\in({\omega}_{P^L},1)$ with
$P^{\prime}(\widehat{\omega}_L)\le 0.$ 
This allows us to conclude that $P^L (\omega)$ is strictly
increasing on $({\omega}_{P^L}, 1).$
We remark that
  ${\omega}_{P^L},\, {\omega}_{P^H}$ and ${\omega}_{P^{\ast}}$ do not
  necessarily exist for the same parameter values. Indeed, from
  \eqref{pl} it follows that ${\omega}_{P^j}$ is inversely
  proportional to $P^j ({\omega}_{P^j}),$ for $j\in\{L,H,\ast\}.$
  Since for each $\omega\in [0,1]$ we have
  $P^L (\omega)<P^{\ast} (\omega)<P^H (\omega),$ it holds that
  ${\omega}_{P^L}>{\omega}_{P^{\ast}}>{\omega}_{P^H}.$ Hence, the
  first to (possibly) enter from above the interval $(0,1)$ is
  ${\omega}_{P^H},$ (possibly) followed by ${\omega}_{P^{\ast}},$
  (possibly) followed by ${\omega}_{P^{L}}.$

As concerns $Y,$  to fix ideas, we will just focus on $Y^{L},$ as analogous arguments hold for $Y^H,$ as well.\\
Recalling \eqref{yp} and \eqref{po}, we have that
\begin{equation}\label{yld}
\left(Y^L\right)'(\omega)=\frac{h}{1-c}\left(P^L(\omega)+\omega \left(P^L\right)'(\omega)\right)=
\frac{h\left(P^L(\omega)dh\omega^2+\omega(dA-F^{\ast}(1 - c))+P^L(\omega)(1-c)\right)}{(1-c)(1 - c - dh \omega^2)}.
\end{equation}
Such derivative is positive for each $\omega\in(0,1)$ when
$dA-F^{\ast}(1-c)\ge 0$ and in this case $Y^L(\omega)$ increases with
$\omega,$ for every $\omega.$ When instead $dA-F^{\ast}(1-c)<0,$ the
numerator of $\left(Y^L\right)'(\omega)$ may become negative. More
precisely, $\left(Y^L\right)'(0)>0.$ If $\left(Y^L\right)'(\omega)>0$
for every $\omega\in (0,1),$ then $Y^L$ is strictly increasing on
$(0,1).$ Otherwise there exists ${\omega}_{Y^L}\in(0,1)$ such that
\begin{equation}\label{lt}
P^L ({\omega}_{Y^L}) = \frac{{\omega}_{Y^L}(F (1 - c) - Ad)}{1 - c + dh {\omega}_{Y^L}^2}.
\end{equation}
If such ${\omega}_{Y^L} \in (0, 1)$ exists, then $Y^L$ is increasing on $(0,
{\omega}_{Y^L})$ and $\left(Y^L\right)' ({\omega}_{Y^L}) = 0$. Since
\[ \left(Y^L\right)'' ({\omega}_{Y^L}) = \frac{h\left((1 - c + dh
   {\omega}_{Y^L}^2)\left(P^L\right)'(\omega_{Y^L})+2hd\omega P^L(\omega_{Y^L})-F (1 - c) + Ad\right) }{(1-c)(1 - c - dh
   {\omega}_{Y^L}^2)}, \]
	recalling \eqref{po} and \eqref{lt} we find
	\[ \left(Y^L\right)'' ({\omega}_{Y^L})=\frac{2h(Ad-F (1 - c))}{(1 - c + dh {\omega}_{Y^L}^2)(1 - c - dh {\omega}_{Y^L}^2)}<0.\]
	Hence, ${\omega}_{Y^L}$ is a local maximum point. Actually, it is a global maximum point, i.e., $Y^L$ increases on $(0,{\omega}_{Y^L})$ and decreases on $({\omega}_{Y^L},1).$ Indeed, if for some $\omega_0\in({\omega}_{Y^L},1)$ we have $\left(Y^L\right)'(\omega_0)<0,$ then, by \eqref{yld}, it follows that
$$P^L (\omega_0) < \frac{\omega_0(F (1 - c) - Ad)}{1 - c + dh \omega_0^2}.$$
If by contradiction there exists $\omega_2>\omega_0$ where $Y^L$ is
increasing, then $\left(Y^L\right)'(\omega_2)>0.$ By the continuity of
$\left(Y^L\right)'$ there must exists at least an
$\omega_1\in(\omega_0,\omega_2)$ in which $Y'(\omega_1)=0.$ Let us
assume that $\omega_1$ is the smallest value of $\omega$ for which
this happens, so that $\left(Y^L\right)'(\omega)<0$ on
$[\omega_0,\omega_1).$ We notice that, under the maintained assumption
that $F (1 - c) - Ad>0,$ the map
$$\phi:(0,1)\to\mathbb R,\quad\omega\mapsto\frac{\omega(F (1 - c) - Ad)}{1 - c + dh \omega^2}$$
is strictly increasing with $\omega.$ Then, since
${\omega}_{Y^L}<\omega_1,$ we have
\begin{equation}\label{yl1}
  P^L ({\omega}_{Y^L}) = \frac{{\omega}_{Y^L}(F (1 - c) - Ad)}{1 - c + dh {\omega}_{Y^L}^2}<P^L (\omega_1)=\frac{\omega_1(F (1 - c) - Ad)}{1 - c + dh \omega_1^2}.
\end{equation}
Recalling that, when $F (1 - c) - Ad>0,$ $P^L$ can be either
decreasing for every $\omega\in(0,1)$ or there exists
${\omega}_{P^{L}}\in(0,1)$ such that $P^L$ decreases on
$[0,{\omega}_{P^{L}})$ and increases on $({\omega}_{P^{L}},1),$ it
follows that \eqref{yl1} excludes the first possibility.  Hence, $P^L$
increases on $({\omega}_{P^{L}},1)$ and it holds that
${\omega}_{P^{L}}<\omega_1$ as $P^{L}$ is increasing at
$\omega=\omega_1$. Since by \eqref{yp}, when $P^L$ increases,
$Y^L$ increases, too, we have that $Y^L$ starts increasing in a left
neighborhood of $\omega_1,$ against the assumption on $\omega_1.$
Hence, $Y^L$ increases on $(0,{\omega}_{Y^L})$ and decreases on
$({\omega}_{Y^L},1),$ as desired.
\end{proof}

\begin{proof}[Proof of Proposition \ref{stabcond}]
  In order to use the conditions in \cite{farebrother} to prove
  \eqref{eq:sc} we need to compute the Jacobian matrix $J^{\ast }$ for
  the map $G$ in \eqref{map} in correspondence to $S^{\ast}.$ We have
  that the Jacobian matrix of $G$ is
\begin{equation}\label{J}
  J = \left( \begin{array}{ccc}
     \gamma g_I^{\prime}  (Y - Z) + c & h \omega & - \gamma g_I^{\prime}  (Y - Z)\\
     j_{21} & j_{22} & 0\\
     1 & 0 & 0
   \end{array} \right)
\end{equation}
where
\begin{subequations}\label{js}
  \begin{equation}\label{j21}
    \begin{split}
      j_{21} = & \mu \sigma g_P^{\prime}  \left( - \mu \left( P + F (\omega - 1) - b \left(
            \frac{2}{e^{- 4 b \beta (P + F (\omega - 1) - Yd \omega)} + 1} - 1 \right)
          - Yd \omega \right) \right)\\
      &\cdot\left( d \omega - \frac{8 b^2 \beta d \omega
          e^{- 4 b \beta (P + F (\omega - 1) - Yd \omega)}}{(e^{- 4 b \beta (P + F
            (\omega - 1) - Yd \omega)} + 1)^2} \right)
    \end{split}
  \end{equation}
  and
  \begin{equation}\label{j22}
    \begin{split}
      j_{22} = &\mu \sigma \left( \frac{8 b^2 \beta e^{- 4 b \beta (P + F (\omega
            - 1) - Yd \omega)}}{(e^{- 4 b \beta (P + F (\omega - 1) - Yd \omega)} +
          1)^2} - 1 \right)\\
      & \cdot g_P^{\prime}  \left( - \mu \left( P + F (\omega - 1) - b \left(
            \frac{2}{e^{- 4 b \beta (P + F (\omega - 1) - Yd \omega)} + 1} - 1 \right)
          - Yd \omega \right) \right) + 1.
    \end{split}
  \end{equation}
\end{subequations}
Recalling that from \eqref{eq:spb} we have
$g_I^{\prime}(0)=g_P^{\prime}(0)=1,$ $J^{\ast }$ reads as
  \begin{equation}
    J^{\ast }=\left(
      \begin{array}{ccc}
        c+\gamma & \omega h & -\gamma \\
        \omega dE & 1-E & 0 \\
        1 & 0 & 0%
      \end{array}%
    \right) \label{jac}
  \end{equation}
  The Farebrother conditions are the following:
  \begin{enumerate}
  \item[$i)$] $1+C_{1}+C_{2}+C_{3}>0$
  \item[$ii)$] $1-C_{1}+C_{2}-C_{3}>0$
  \item[$iii)$] $1-C_{2}+C_{1}C_{3}-(C_{3})^2>0$
  \item[$iv)$] $3-C_{2}<0$
  \end{enumerate}
  where $C_{i}$, $i\in\{1,2,3\},$ are the coefficients of the characteristic polynomial
  $\lambda^3+C_{1}\lambda^2+C_{2}\lambda+C_{3}=0$.
  \noindent
  From the Jacobian matrix (\ref{jac}), we find that the characteristic equation is given by
  $$\lambda ^{3}+\left( -c-\gamma-1+E\right) \lambda
  ^{2}+\left( 2\gamma+c-cE-\gamma E-\omega ^{2}dhE\right) \lambda
  +\left( -\gamma+E\gamma\right)=0$$
  which immediately provides $C_1,C_2$ and $C_3.$ Replacing their expressions in
  $i)-iv)$ allows concluding the proof after some simple algebraic
  manipulations. We just notice that $i)$ provides
  $E\left( 1-c-\omega ^{2}dh\right) >0$ which, recalling assumption
  \eqref{a1}, can be simplified as $E>0,$ i.e., \eqref{eq:sc1}.
\end{proof}

\begin{proof}[Proof of Proposition \ref{th:stabbetab}]
  We describe in detail only the scenarios arising on varying $\beta,$ since,
recalling that $E = \sigma \mu (1 - 2 b^2 \beta)$, the role of $b$ is
equivalent to that of $\beta$.

Collecting $E$ in \eqref{eq:sc} we find the following system:
\begin{equation}
  \label{eq:condE} \left\{ \begin{array}{l}
    E  > 0\\
    - E (dh \omega^2 + c + 2 \gamma + 1) + 2 c + 4 \gamma + 2
    > 0\\
    \gamma  (1 - \gamma) E^2 + ((c - \gamma) (1 -
    \gamma) + dh \omega^2) E + (1 - c)  (1 - \gamma) > 0\\
    E (dh \omega^2 + c + \gamma) - 2 \gamma - c + 3 > 0
  \end{array} \right.
\end{equation}
We start solving System \eqref{eq:condE} with respect to
$E \in \mathbb{R}$ under assumption \eqref{a1}. Then we will rewrite
the solutions in terms of the positive parameter $\beta$.

Recalling \eqref{a1}, the first inequality in \eqref{eq:condE}
requires $E > 0,$ while the second inequality can be rewritten as
\[ E < \frac{2 c + 4 \gamma + 2}{dh \omega^2 + c + 2 \gamma + 1} =
  E_2 \]
and the fourth inequality is equivalent to
\[ E > \frac{2 \gamma + c - 3}{dh \omega^2 + c + \gamma} =
   E_4. \]
Combining these three conditions we obtain
\begin{equation}
  \max \{ 0, E_4 \} < E < E_2\,. \label{eq:cond124}
\end{equation}
Concerning the third condition in \eqref{eq:condE}, we
distinguish three cases:

a) $\gamma = 1.$ We start noting that we have $E_4 < 0,$ so that condition
\eqref{eq:cond124} reduces to
\begin{equation}
  0 < E < E_2\,. \label{eq:cond124r}
\end{equation}
The third inequality in \eqref{eq:condE} reduces to $Edh \omega^2 > 0,$ i.e. $E > 0,$ so that
System \eqref{eq:condE} is solved for $0<E<E_2,$ i.e.,
\begin{equation}\label{eq:E2}
  \frac{1}{2 b^2} \left( 1 - \frac{E_2}{\mu \sigma} \right) < \beta <
   \frac{1}{2 b^2},
\end{equation}
which means that if $\mu\sigma<E_2$ we are in the destabilizing
scenario, since the left inequality in \eqref{eq:E2} is always
fulfilled, while if $\mu\sigma>E_2$ we obtain the mixed scenario.

b) $\gamma < 1.$ We again have $E_4 < 0,$ so that condition
\eqref{eq:cond124} reduces to \eqref{eq:cond124r}. The lhs
of the third inequality in \eqref{eq:condE} represents a convex
parabola, intersecting the vertical axis at $(1 - c)(1 - \gamma) > 0$.
If the discriminant is negative, the inequality is fulfilled by any
$E$, so that System \eqref{eq:condE} is solved when
\eqref{eq:cond124r} holds, from which we obtain the same possible
scenarios considered in case a).  If the discriminant is positive, in
principle, we may have two situations, i.e.,
$(c - \gamma) (1 - \gamma) + dh \omega^2>0$ and
$(c - \gamma) (1 - \gamma) + dh \omega^2\le 0.$ If
$(c - \gamma) (1 - \gamma) + dh \omega^2,$ representing the slope of
the parabola at its intersection point with the vertical axis, is
positive then the third inequality is solved by
$E < E_3^a \vee E > E_3^b$, where $E^{a, b}_3$ are the (negative)
roots of the lhs of the third inequality in \eqref{eq:condE}. This means that
simultaneously considering $E < E_3^a \vee E > E_3^b$ and
\eqref{eq:cond124r}, System \eqref{eq:condE} is again solved when
\eqref{eq:cond124r} holds true and we obtain the same scenarios as before.

If $(c - \gamma) (1 - \gamma) + dh \omega^2 \le 0$ the discriminant of
the lhs of the third inequality of \eqref{eq:condE} can not
be positive. In fact we have
\[
  \gamma  (1 - \gamma) E^2 + ((c - \gamma) (1 -
   \gamma) + dh \omega^2) E + (1 - c)  (1 - \gamma) >
   (1 - \gamma)\left(\gamma   E^2 + (c - \gamma)  E + (1 - c)\right)
\]
and the rhs is strictly positive. To show this, it is
enough to note that the
discriminant of $\gamma E^2 + (c - \gamma) E + (1 - c)$, given by
\[ c^2 + 2 c \gamma + \gamma ^2 - 4 \gamma = c^2 - \gamma + 2 \gamma
  (c - 1) + \gamma (\gamma - 1), \] is negative since
$c < 1, \gamma < 1$ and, in order to have
$(c - \gamma) (1 - \gamma) + dh \omega^2 \le 0$, we necessarily need
$c \le \gamma$, which, together with $\gamma < 1,$ guarantees
$c^2 < \gamma$.

If the discriminant of the lhs of the third inequality in
\eqref{eq:condE} is negative or null, by the same arguments used above, it follows
that the third inequality in \eqref{eq:condE} is fulfilled by any
value of $E>0$ and we obtain again the scenarios arising in case a).

c) $\gamma > 1.$ In this case the lhs of the third
inequality in \eqref{eq:condE} represents a concave parabola,
intersecting the vertical axis at $(1 - c) (1 - \gamma) < 0$. If the
discriminant is negative, the inequality is never fulfilled, and we
have the unconditionally unstable scenario. Conversely, if the
discriminant is positive, the third inequality is solved by
$E_3^a < E < E_3^b$, which combined with \eqref{eq:cond124} provides
$\max \{ 0, E_4, E_3^a \} < E < \min \{ E_2, E_3^b \},$ that is
qualitatively identical to \eqref{eq:cond124}. This allows concluding
that no other scenarios can arise.
\end{proof}

\begin{proof}[Proof of Proposition \ref{th:stabw}]
    We start making $\omega$ explicit in System \eqref{eq:sc}, which can then be
  rewritten as follows:
\[
\left\{
\begin{array}{l}
E >0 \\
\omega ^{2}<\frac{\left( 2-E\right) \left( 1+c+2\gamma\right) }{dhE} \\
\omega ^{2}>\frac{c-cE+\gamma-1+E\gamma-\gamma c-E\gamma^{2}-E^{2}\gamma+E\gamma c+E^{2}\gamma^2}{dhE} \\
\omega ^{2}>\frac{2\gamma+c-cE-\gamma E-3}{dhE}%
\end{array}%
\right.
\]
Accordingly, the second inequality is never fulfilled if $E\ge 2,$
while it is fulfilled for
\begin{equation}
0\le\omega<\sqrt{\frac{\left( 2-E\right) \left( 1+c+2\gamma\right)
  }{dhE}}\equiv\omega_2\label{w2}
\end{equation}
if $0<E<2.$
As concerns the third and the fourth inequalities, they are always fulfilled if the numerators on the rhs are negative, while they are respectively satisfied for
$$1\ge\omega>\sqrt{\frac{(\gamma - 1)(\gamma E^2-(\gamma-c) E +
      1-c)}{dhE}}\equiv\omega_3$$
and
$$1\ge\omega>\sqrt{\frac{2\gamma+c-cE-\gamma E-3}{dhE}}\equiv\omega_4$$
when both numerators are non-negative.  In particular, we have that
the numerator of $\omega_4$ is non-negative, i.e.,
$2\gamma+c-cE-\gamma E-3\ge 0,$ when $\gamma>(3-c)/2$ and
$0<E\le\frac{2\gamma+c-3}{c+\gamma}.$ On the other hand, it holds that
the numerator of $\omega_3$ is non-negative, i.e.,
\begin{equation}\label{ce}
(\gamma - 1)(\gamma E^2-(\gamma-c) E +
      1-c)\ge 0,
\end{equation}
when $\gamma=1,$ when
$\gamma E^2-(\gamma-c)E+1-c\ge 0$ if
$\gamma>1,$ and when
$\gamma E^2-(\gamma-c)E+1-c\le 0$ if
$\gamma<1.$ We notice that real solutions to
$\gamma E^2-(\gamma-c)E+1-c=0$ exist only
if
\begin{equation}\label{delta}
  \Delta=\gamma^2-2(2-c)\gamma+c^2\ge 0\; \Leftrightarrow \;
\gamma\le\gamma_1\, \vee \,
\gamma\ge\gamma_2,
\end{equation}
with $\gamma_1\equiv 2-c-2\sqrt{1-c}$ and
$\gamma_2\equiv 2-c+2\sqrt{1-c}$ that fulfill
\begin{equation}\label{ineq}
  0<\gamma_1<c<1<\gamma_2.
\end{equation}
Condition \eqref{ce} is then satisfied by any $E>0$ when $\gamma=1$ or
when $\gamma>1$ and $\Delta<0,$ which, recalling \eqref{delta} and
\eqref{ineq}, jointly provide $1<\gamma<\gamma_2.$ Hence, we can
conclude that condition \eqref{ce} is fulfilled by any $E>0$ when
$\gamma\in[1,\gamma_2).$\\
If $\gamma>1$ and $\Delta\geq0$ (which jointly require
$\gamma\ge\gamma_2$) then \eqref{ce} is fulfilled for
$0<E\le E_1\equiv\frac{\gamma-c-\sqrt{\Delta}}{2\gamma}$ or for $E\ge
E_2\equiv\frac{\gamma-c+\sqrt{\Delta}}{2\gamma}.$\\
Conversely, if $\gamma\ge\gamma_2$ but $E\in(E_1,E_2),$ \eqref{ce}
does not hold true, and thus iii) is fulfilled for any $\omega.$ If
$\gamma<1$ and $\Delta>0$ (namely, if $\gamma<\gamma_1$), we have that
\eqref{ce} holds true for $E\in(E_1,E_2)\subset(-\infty,0)$ since
$\gamma E^2-(\gamma-c) E + 1-c$ is a convex parabola, increasing at
the positive intersection $1-c$ with the vertical axis, since,
recalling \eqref{ineq}, $-(\gamma-c)>0.$ Hence \eqref{ce} is never
satisfied for $E>0$ when $\gamma<\gamma_1$ and iii) is fulfilled again
for any $\omega.$ The same conclusion holds true if $\gamma<1$ and
$\Delta\leq0$ (namely, if $\gamma_1\le\gamma<1$).


Summarizing, we have that\footnote{We stress that, in order to
  simplify the description of the conditions under which ii)--iv) are
  fulfilled, we report just
  the cases with $E>0,$ so that also i) is satisfied.}:
\medskip

i) is fulfilled for $E>0;$

ii) is fulfilled for $0\le\omega<\omega_2$ and $0<E<2;$

iii) is fulfilled for any value of $\omega\in[0,1],$ if $\gamma\in(0,1),$ or if $\gamma\ge\gamma_2$ and $E\in(E_1,E_2);$

iii) is fulfilled for $1\ge\omega>\omega_3,$ if $\gamma\in[1,\gamma_2),$ or if $\gamma\ge\gamma_2$
and $E\in(0,E_1]\cup[E_2,+\infty);$

iv) is fulfilled for any value of $\omega\in[0,1],$ if $E>\frac{2\gamma+c-3}{\gamma+c};$

iv) is fulfilled for $1\ge\omega>\omega_4,$ if
$\gamma>(3-c)/2$ and $0<E\le\frac{2\gamma+c-3}{c+\gamma}.$

\medskip

Hence, each condition in \eqref{eq:sc} is fulfilled on an
  interval for $\omega.$  Since the intersection of intervals is
  an (empty or nonempty) interval, the set on
  which $S^\ast$ is locally asymptotically stable is connected. Thus, the only scenarios that may be possibly detected on increasing $\omega$ are the unconditionally unstable, unconditionally stable, destabilizing, stabilizing and mixed (meaning unstable-stable-unstable) scenarios.\\
In particular, the unconditionally unstable scenario occurs, e.g., when $E\equiv\mu \sigma \left(1- 2b^{2}\beta\right)<0,$
that is, for $b>{1}/{\sqrt{2 \beta}},$ as in this case \eqref{eq:sc1} is
violated and $S^{\ast}$ is unstable for each value of $\omega.$

Recalling assumption \eqref{a1}, the unconditionally stable scenario
  occurs, for instance, when
  $\gamma\in (0,1)$ and $E\in\bigl(0,\frac{2(1+c+2\gamma)}{d
    h+1+c+2\gamma}\bigr).$ In fact, under such conditions,
  recalling the definition of
    $\omega_2$ in \eqref{w2} and that
  $c\in(0,1),$ we have
  $2>\frac{2(1+c+2\gamma)}{d
    h+1+c+2\gamma}>E>0>\frac{c-1}{1+c}>\frac{2\gamma+c-3}{\gamma+c}$
  and $\omega_2>1.$\\
  Hence, the dynamical system is locally asymptotically stable at $S^{\ast}$ for each $\omega\in [0,1].$\\
 Recalling assumption \eqref{a1}, the destabilizing scenario occurs,
  for instance, when $\gamma\in (0,1)$ and
  $E\in\bigl(\frac{2(1+c+2\gamma)}{d
    h+1+c+2\gamma},2\bigr).$
  In fact, under such conditions, we have
  $2>E>\frac{2(1+c+2\gamma)}{d
    h+1+c+2\gamma}>0>\frac{c-1}{1+c}>\frac{2\gamma+c-3}{\gamma+c}$ and
  $0<\omega_2<1.$
  Hence, the system is locally asymptotically stable at $S^{\ast}$ for $\omega\in [0,\omega_2)$ and unstable for $\omega\in(\omega_2,1].$\\
  Still under assumption \eqref{a1}, the stabilizing scenario occurs,
  for instance, when $\gamma\in[1,(3-c)/2),\,E\in(0,2)$ and
  $dh\in(\max\{0,L\},\min\{R,1-c\}),$ where\footnote{We notice that
    the conditions on the parameters are compatible, e.g., when $E=1.$
    Indeed, in such case we obtain $dh\in(\gamma-1,1-c),$ and
    $\gamma-1<1-c$ is fulfilled since $\gamma<(3-c)/2<2-c,$ as
    $c\in(0,1).$}
		\begin{equation}\label{lr}
L=(\gamma-1)\left[\frac{1-c}{E}+E\gamma-\gamma+c\right],\quad R=\frac{(2-E)(1+c+2\gamma)}{E}.
\end{equation}
Namely, under such conditions, we have
$2>E>0>\frac{2\gamma+c-3}{\gamma+c},\,\gamma_2>(3-c)/2>\gamma\ge 1$
and $0<\omega_3<1<\omega_2,$ since $dh>L$ is equivalent to
$\omega_3<1$ while $dh<R$ is equivalent to $\omega_2>1.$\\
Hence, the system is locally asymptotically stable at
$S^{\ast}$ for $\omega\in (\omega_3,1]$ and
unstable for $\omega\in[0,\omega_3).$\\
Recalling \eqref{a1} and the definition of $R$ in \eqref{lr}, the
mixed scenario occurs, for instance, when
$\gamma\in[1,(3-c)/2),\,E\in(0,2),\,dh\in(R,1-c)$ and
$(E^2-E)\gamma^2+(3E-c-3-E^2+Ec)\gamma+E-3-c<0.$\,\footnote{We notice
  that the conditions on the parameters are consistent. Indeed, they
  are fulfilled, e.g., for $c=0.5,\,\gamma=1.1,dh=0.2$ and $E=1.9.$ In
  such case we find $\omega_3=0.936,\,\omega_2=0.986,$ and
  \eqref{a1} is satisfied.} Namely, under such conditions, we have
$2>E>0>\frac{2\gamma+c-3}{\gamma+c},\,\gamma_2>(3-c)/2>\gamma\ge
1,\,0<\omega_3<\omega_2<1.$ Hence, the system is locally
asymptotically stable at $S^{\ast}$ for
$\omega\in (\omega_3,\omega_2)$ and unstable for
$\omega\in[0,\omega_3)\cup(\omega_2,1].$
\end{proof}

\begin{proof}[Proof of Proposition \ref{th:stabhl}]
  We start showing that evaluating the Jacobian matrix of $G$ in \eqref{map} at
  either $S^L$ or $S^H$ we obtain the Jacobian matrices $J^L$ and
  $J^H$ whose structure is very similar to that of $J^{\ast}$
  in \eqref{jac}. Indeed, the first and the third rows of both $J^L$ and
  $J^H$ coincide with those of $J^{\ast}.$ Now we evaluate $j_{21}$ and
  $j_{22},$ respectively defined in \eqref{j21} and \eqref{j22}, at
  $S^L$ or at $S^H.$

  We start noting that, recalling (\ref{yp}) and (\ref{imp}), we can
  rewrite $j_{21}$ and $j_{22}$ as
  \[ j_{21} = \omega d\mu \sigma   \left( 1 - \frac{8 b^2 \beta e^{- 4
          Wb \beta}}{(e^{- 4 Wb \beta} + 1)^2} \right) g_P^{\prime}
    \left( - \mu \left( W + b - \frac{2 b}{e^{- 4 Wb \beta} + 1}
      \right) \right) \]
  and
  \[ j_{22} = \mu \sigma \left( \frac{8 b^2 \beta e^{- 4 Wb
          \beta}}{(e^{- 4 Wb \beta} + 1)^2} - 1 \right) g_P^{\prime}
    \left( - \mu \left( W + b - \frac{2 b}{e^{- 4 Wb \beta} + 1} \right)
    \right) + 1. \]
  From (\ref{imp}) we have
  \[ g_P^{\prime} \left( - \mu \left( W + b - \frac{2 b}{e^{- 4 Wb
            \beta} + 1} \right) \right) = g_P^{\prime} (0) = 1,
  \]
  while using (\ref{s}) we can write
  \[ \frac{8 b^2 \beta e^{- 4 Wb \beta}}{(e^{- 4 Wb \beta} + 1)^2} = 2
    \beta (b^2 - W^2), \] so that the Jacobian matrix evaluated at either
  $S^L$ or $S^H$
  can be written as
  \[ \left(
      \begin{array}{ccc}
        c+\gamma & \omega h  & - \gamma\\
        \omega d\mu \sigma  (1 + 2 \beta (W^2 - b^2)) & 1 - \mu \sigma (1 + 2
                                                        \beta (W^2 - b^2)) & 0\\
        1 & 0 & 0
      \end{array} \right) \]
  Recalling (\ref{s2}), at $S^{\ast}$ we have $W
  = 0,$ which provides \eqref{jac}, while at both \ $S^L$ and $S^H$ we have $W^2 > 0$. Introducing
  \[ E^i = \mu \sigma (1 + 2 \beta ((W^i)^2 - b^2)), \; i \in \{ L, H
    \}, \] where $W^L$ and $W^H$ are defined using $P^L$ and $P^H$ into
  \eqref{s2}, respectively, we have that the Jacobian matrices
  evaluated at the steady states are
\begin{equation}\label{jhl}
  J^i = \left( \begin{array}{ccc}
     c+\gamma  & \omega h & - \gamma\\
      \omega d E^i & 1 - E^i & 0\\
     1 & 0 & 0
   \end{array} \right), \quad i \in \{ \ast, L, H \},
\end{equation}
which are formally the same as $J^{\ast}$ in \eqref{jac} provided
that we replace $E$ with $E^i.$


Thanks to the symmetry of $S^L$ and $S^H$ w.r.t. $S^{\ast}$ proved in Proposition
\ref{ss} and by \eqref{s2}, for the same parameters' configuration we have $W^H=- W^L$, and
hence $E^H = E^L$. From now on we then only consider
$S^H.$ We fix all parameters but $\beta$ and, for clarity's sake, we make explicit the dependence of some quantities on
$\beta.$

We want to show that for each $\beta \in (0, 1 / (2 b^2))$ there is a
unique $\beta^H \in (1 / (2 b^2), + \infty),$ strictly decreasing as
$\beta$ increases, such that $J^{\ast} (\beta)$ and $J^H ( \beta^H)$
have the same eigenvalues. To this end, it suffices that
$E(\beta) = E^H (\beta^H)$, where $E$ is defined in Proposition
\ref{stabcond}, namely that
\begin{equation}
  \beta b^2 = \beta^H (b^2 - W^H (\beta^H)^2). \label{imp3}
\end{equation}
By Propositions \ref{ss} and \ref{scbeta}, we know that
$W^H (\beta^H),$ implicitly defined on $(1 / (2 b^2), + \infty $) by
(\ref{s}), is an increasing function whose image is $(0,
b),$ so that we can rewrite (\ref{s}) as
$\beta^H = \frac{1}{4b W^H (\beta^H) } \ln \left( \frac{b + W^H
    (\beta^H)}{b - W^H (\beta^H)} \right),$ by means of which
\eqref{imp3} can be recast into
\[ \beta b^2 = \frac{\ln \left( \frac{b + W^H(\beta^H)}{b - W^H (\beta^H)} \right)}{4b W^H (\beta^H) } (b^2 - (W^H (\beta^H))^2)
  = \frac{\ln\left( \frac{1 +
      \tilde{W}^H (\beta^H)}{1 - \tilde{W}^H (\beta^H)} \right)}{4 \tilde{W}^H (\beta^H) } (1 -
  (\tilde{W}^H (\beta^H))^2), \] where
$\tilde{W}^H (\beta^H) = W^H (\beta^H) / b$. A straightforward check
shows that function
$h : (0, 1) \rightarrow \mathbb{R},$ defined by
\[ h (z) = \frac{1}{4 z} \ln \left( \frac{1 + z}{1 - z} \right) (1 -
  z^2), \]
is a strictly decreasing function with
$\lim_{z \rightarrow 0^+} h (z) = {1}/{2}$ and $\lim_{z
  \rightarrow 1^-} h (z) = 0.$
This means that equation $h (z) = \beta b^2$ has exactly one solution
$z (\beta) \in (0, 1)$ for each $\beta \in (0, 1 / (2
b^2))$. Moreover, $z (\beta)$ is decreasing with respect to
$\beta$. From $\tilde{W}^H (\beta^H) = z (\beta)$ we obtain
$W^H (\beta^H) = b z (\beta) \in (0, b)$, which, recalling that $W^H$
is increasing and its image is $(0, b)$, proves that for each
$\beta \in (0, 1 / (2 b^2))$ there is exactly one
$\beta^H (\beta) \in (1 / (2 b^2), + \infty)$. Moreover,
$\beta^H (\beta)$ is strictly decreasing and $\lim_{\beta \rightarrow 0^{+}} \beta^H (\beta) = + \infty,$
$\lim_{\beta \rightarrow (1 / (2 b^2))^{-}} \beta^H (\beta) = 1 / (2 b^2).$
 Through
function $\beta^H (\beta)$ we can easily translate each scenario for
$S^{\ast}$ into a unique corresponding scenario
for $S^H.$ For instance, a destabilizing scenario for
$S^{\ast}$ occurs if and only if the stability
interval is $I_d=(0,\beta_b)$ with either $\beta_b<1/(2b^2)$ or
$\beta_b=1/(2b^2).$ The image of $I_d$ through $\beta^H (\beta)$ is
then $(\beta^H (\beta_b),+\infty),$ with either
$\beta^H (\beta_b)=1/(2b^2)$ (unconditionally stable scenario) or
$\beta^H (\beta_b)>1/(2b^2)$ (stabilizing scenario). A mixed scenario
for $S^{\ast}$ occurs if and only if the stability
interval is $I_m=(\beta_a,\beta_b)$ with either $\beta_b<1/(2b^2)$ or
$\beta_b=1/(2b^2).$ The image of $I_m$ through $\beta^H (\beta)$ is
then $(\beta^H (\beta_b),\beta^H (\beta_b)),$ with either
$\beta^H (\beta_b)=1/(2b^2)$ (destabilizing scenario) or
$\beta^H (\beta_a)>1/(2b^2)$ (mixed scenario). The last case of
unconditional stability is straightforward.

If we fix all the parameters and we study the behavior of the matrices in \eqref{jhl}
and of their eigenvalues
on increasing $\omega,$ we can immediately note that $J^{\ast}=J^{H}$
if we replace $\beta\in (0,1/(2b^2))$ with $\beta^H(\beta)\in(1/(2b^2),+\infty)$ via the strictly decreasing function
 $\varphi:(0,1/(2b^2))\to(1/(2b^2),+\infty),\,\beta\mapsto \varphi(\beta),$ whose existence has been shown along the proof.
This allows concluding.
\end{proof}


\end{document}